\documentclass[letterpaper,11pt,oneside,fullpage,final]{article}

\usepackage[margin=1.1in]{geometry}
\usepackage{epstopdf}
\usepackage{subfigure}
\usepackage{amsmath}
\usepackage{amssymb}
\usepackage{amsthm}
\usepackage{enumitem}
\usepackage{multirow}
\usepackage{bm}
\usepackage{xcolor}
\usepackage{algorithm,algorithmic}
\usepackage{natbib}
\usepackage{tikz}
\usetikzlibrary{positioning}
\usepackage{titlesec}
\titleformat{\subsubsection}[runin]{\bf}{}{}{}[:]

\theoremstyle{definition}
\newtheorem{theorem}{Theorem}[section]

\newtheorem{lemma}[theorem]{Lemma}

\DeclareMathOperator*{\argmin}{arg\,min}

\title{
\bf
Deep Neural Network Framework Based on Backward Stochastic Differential Equations for Pricing and Hedging American Options in High Dimensions
}

\author{
Yangang Chen\thanks{Department of Applied Mathematics, University of Waterloo, 200 University Avenue West, Waterloo, ON, N2L 3G1, Canada.}
\and
Justin W. L. Wan\thanks{David R. Cheriton School of Computer Science, University of Waterloo, 200 University Avenue West, Waterloo, ON, N2L 3G1, Canada.}
}

\begin{document}

\maketitle

\begin{abstract}
We propose a deep neural network framework for computing prices and deltas of American options in high dimensions. The architecture of the framework is a sequence of neural networks, where each network learns the difference of the price functions between adjacent timesteps. We introduce the least squares residual of the associated backward stochastic differential equation as the loss function.
Our proposed framework yields prices and deltas on the entire spacetime, not only at a given point. The computational cost of the proposed approach is quadratic in dimension, which addresses the curse of dimensionality issue that state-of-the-art approaches suffer.
Our numerical simulations demonstrate these contributions, and show that the proposed neural network framework outperforms state-of-the-art approaches in high dimensions.
\end{abstract}



\section{Introduction}
\label{sec:intro}

American options are among the most common derivatives in financial markets. In practical applications of hedging, we are required to compute not only an American option price, but also the derivatives of a price with respect to the underlying asset prices, called American option delta \citep{hull2003options}. Numerous approaches have been proposed for solving American option problems, such as binomial trees \citep{hull2003options}, numerically solving partial differential equations (PDEs) with free boundary conditions, with policy iterations or with penalty terms \citep{achdou2005computational,duffy2013finite,forsyth2002quadratic,reisinger2012use}, regression-based methods \citep{longstaff2001valuing,tsitsiklis1999optimal,kohler2010review}, stochastic mesh methods \citep{broadie2004stochastic}, sparse grids methods \citep{bungartz2004sparse,reisinger2007efficient,leentvaar2008pricing}, etc. When the dimension of an American option, i.e., the number of underlying assets, is greater than 3, numerical solution of PDEs is infeasible, as the complexity grows exponentially with the dimension.
When the dimension $d$ is moderate (e.g., $d\leq 20$), the regression-based Longstaff-Schwartz method \citep{longstaff2001valuing} is widely considered as the state-of-the-art approach for computing option prices. In addition, one can combine the Longstaff-Schwartz method with the methods proposed in \citet{broadie1996estimating,bouchard2012monte} and \citet{thom2009longstaff} to compute corresponding option deltas.
We note that these approaches only compute option prices and deltas at a given point (e.g., $t=0$)\footnote{Although one may consider using the Longstaff-Schwartz regressed values as an estimate of the spacetime prices, Figure 1 in \citet{bouchard2012monte} shows that using such regressed values as the spacetime solution is inaccurate. Alternatively, one may consider applying the Longstaff-Schwartz method repeatedly on all the spacetime points, where {\it every} point requires $M\to\infty$ samples. However, this is expensive.}.
However, we emphasize that price and delta at a given point are insufficient for a complete delta hedging process, which requires computing prices and deltas on the entire spacetime \citep[see][for explanations and concrete examples]{hull2003options,he2006calibration,kennedy2009dynamic}.
Furthermore, for the Longstaff-Schwartz method, a set of $\chi$-th degree polynomials is normally used as the basis for regression, which leads to $\chi$-th degree complexity (rather than exponential complexity). However, $\chi$ is required to go to infinity for convergence \citep{longstaff2001valuing,stentoft2004convergence}, which still results in a high complexity.

In this paper, we propose a deep neural network framework for solving high-dimensional American option problems. The major contributions of the proposed neural network framework are summarized as follows:
\begin{itemize}[leftmargin=10mm]
\item Our neural network architecture is new.
Assuming that there are $N$ discrete timesteps, we design a sequence of $N$ recursively-defined feedforward neural networks\footnote{Here the proposed ``recursively-defined" feedforward network is not the same as the Recurrent Neural Network (RNN) in the literature, which will be explained in Section \ref{subsec:NN}.}, where each network extracts the difference between the price functions of adjacent timesteps.
\item Our neural network formulation utilizes domain knowledge of American options, including smoothing the payoff at $t=T$, adding the payoff and the previous continuation price as features, etc.
\item We introduce the least squares residual of the associated backward stochastic differential equation (BSDE) as the loss function of neural networks. BSDE couples prices and deltas in one single equation, and thus evaluates both prices and deltas accurately.
\item Our proposed approach can evaluate both option prices and option deltas on the entire spacetime, not only at a given point.
\item The computational cost of the proposed neural network framework grows quadratically with the dimension $d$, in contrast to exponential growth as in the Longstaff-Schwartz method. In particular, our approach outperforms the Longstaff-Schwartz method when $d\geq 20$, in the sense that our proposed approach solves American option prices and deltas in as high as 200 dimension, while the Longstaff-Schwartz method fails to solve the problems due to the out-of-memory error and the worse-than-quadratic cost.
\end{itemize}

We note that this paper is not the only neural network framework for American option problems. Early research of neural networks in American options can be found in \citet{kohler2010pricing} and \citet{haugh2004pricing}. They consider using one-hidden-layer (shallow) feedforward neural networks for option pricing. However, the highest dimension considered in their numerical simulations is 10.
Very recently, several types of deep neural network approaches were proposed in \citet{sirignano2018dgm,weinan2017deep,beck2017machine,han2017overcoming} and \citet{fujii2017asymptotic}. They suggest that increasing the depth of neural networks is important in pushing the solutions to higher dimensions.
Similar to these approaches, our proposed framework is also a deep neural network approach. However, we emphasize that there are a few key differences between our proposed approach and the other deep neural network approaches.
\begin{itemize}[leftmargin=10mm]
\item {\it Different computed quantities:}
Our approach computes American option prices and deltas on the entire spacetime.
The approach in \citet{sirignano2018dgm} computes prices but not deltas.
The approaches in \citet{weinan2017deep,beck2017machine} and \citet{han2017overcoming} only consider European option prices, noting that European options are easier to price than American options.
Although \citet{fujii2017asymptotic} extends their methods to American options, the authors only compute the price at a given point.
In particular, we emphasize that only our paper discusses and simulates hedging options, which is beyond merely pricing options.
\item {\it Different network architectures:}
Our network architecture is a chain of recursively-defined networks that learn the difference of the price functions between adjacent timesteps; \citet{sirignano2018dgm} uses a long short-term neural network that learns the price function itself;
\citet{weinan2017deep,beck2017machine,han2017overcoming} and \citet{fujii2017asymptotic} consider a chain of isolated, independent feedforward networks.
\item {\it Different loss functions:}
The approach in \citet{sirignano2018dgm} defines the loss function by the residual of the Hamilton-Jacobi-Bellman partial differential equation emerging from the Black-Scholes theory. It involves computing the Hessian of the output price function, which is expensive in both time and memory, and is difficult to implement. Our framework uses the residual of one single BSDE as the loss function, which avoids computing the Hessian.
The approaches in \citet{weinan2017deep,beck2017machine,han2017overcoming} and \citet{fujii2017asymptotic} involve the integral form of multiple BSDEs, which is redundant for option pricing. In addition, their BSDEs are 
not used as loss functions.
\end{itemize}

The paper is organized as follows. Section \ref{sec:American} defines the American option problems. Section \ref{sec:BSDE} introduces the BSDE formation and the least squares residual loss function. Section \ref{sec:NN} describes the architecture and components of the proposed neural network model. Section \ref{sec:key} discusses the techniques that improve the performance of the framework. Section \ref{sec:cost} analyzes the computational cost. In Section \ref{sec:numerical}, we present numerical solutions of option prices and deltas to illustrate the advantage of our deep neural network framework. Section \ref{sec:conclusion} concludes the paper.


\section{American Options}
\label{sec:American}

Suppose an American option contains a basket of $d$ underlying assets. Let $\vec{S}=(S_1, \cdots, S_d)^T\in\mathbb{R}^d$ be the prices of the underlying assets. Note that $\vec{S}$ is a random variable. In order to distinguish random and deterministic variables, we will use capital and lowercase letters respectively. Let $t\in [0,T]$ be the time up to the expiry $T$. Let $r$ be the interest rate. Let $\delta_i$ and $\sigma_i$ ($i = 1, \cdots, d$) be the dividend and volatility of each underlying asset. Let $\rho\in\mathbb{R}^{d\times d}$ be a correlation matrix. Define $d$ correlated random variables
$
dW_i(t) = 
\sum_{j=1}^d L_{ij} \phi_j(t) \sqrt{dt},
$
where
$\phi_i(t) \sim \mathcal{N}(0,1)$
are independent standard normal random variables, and $L$ is the Cholesky factorization of the correlation matrix, i.e., $\rho = LL^T$.
Given an initial state $\vec{s}^0\in\mathbb{R}^d$, the prices of the underlying assets $\vec{S}$ evolve under the following stochastic differential equations (SDEs):
\begin{equation}
\label{eq:SDE_S}
\begin{array}{l}
dS_i(t) = (r-\delta_i) S_i(t)dt
+ \sigma_i S_i(t) dW_i(t),
\quad
i = 1, \cdots, d,
\\
\vec{S}(0) = \vec{s}^0.
\end{array}
\end{equation}
Let $f(\vec{s})$ be the payoff function of the option at the state $\vec{s}$, which usually takes the form of
\begin{equation}
\label{eq:payoff}
f(\vec{s}) = \max(g(\vec{s}), 0).
\end{equation}
For instance, if we denote the strike price as $K$, then the commonly-seen max call options have the payoff function of
$
f(\vec{s})=\max\left(\max_{i=1,...,d} (s_i) \, -K, \,0\right)
$.

Let $c(\vec{s},t)$ be the continuation price, i.e., the discounted option payoff provided that the option is not exercised at time $t$ and state $\vec{s}$:
\begin{equation}
\label{eq:AmerOp_def1}
c(\vec{s},t)
= \max_{\tau\in[t,T]}
\mathbb{E}
\left[ \,
\left.
e^{-r(\tau-t)} f(\vec{S}(\tau))
\, \right| \,
\vec{S}(t) = \vec{s}
\, \right],
\end{equation}
where $\tau$ is the stopping time. Then the American option price $v(\vec{s},t)$ is defined as
\begin{equation}
\label{eq:AmerOp_def2}
\begin{array}{rl}
v(\vec{s},t)
\,
= &
\max\left[
c(\vec{s},t),
f(\vec{s})
\right]
\\
= &
\left\{
\begin{array}{ll}
c(\vec{s},t),
&
\text{ if }
c(\vec{s},t) > f(\vec{s})
\text{, i.e., the option is continued at } (\vec{s},t),
\\
f(\vec{s}),
&
\text{ if }
c(\vec{s},t) \leq f(\vec{s})
\text{, i.e., the option is exercised at } (\vec{s},t).
\end{array}
\right.
\end{array}
\end{equation}

In practical application of hedging, we are also interested in the first derivative of the American option price,
$\vec{\nabla} v (\vec{s},t)
\equiv
\left(
\frac{\partial v}{\partial s_1}(\vec{s},t),
\cdots
\frac{\partial v}{\partial s_d}(\vec{s},t)
\right)^T$.
This is called the ``delta" of the American option. The objective of this paper is to solve for both the option price
$v(\vec{s},t)$
and the option delta
$\vec{\nabla} v (\vec{s},t)$ on the entire spacetime.


\section{Backward Stochastic Differential Equation (BSDE) Formulation}
\label{sec:BSDE}

\subsection{BSDE formulation}
\label{subsec:BSDE}

Our approach is to first convert the American option problem into a backward stochastic differential equation (BSDE) using the following theorem:

%
%

\begin{theorem}[BSDE formulation]
\label{thm:BSDE}
Assume that an American option is not exercised at time $[t, t + dt]$. Then the continuation price of an American option at time $t$ satisfies the following BSDE:
\begin{equation}
\label{eq:SDE_c}
dc(\vec{S},t)
= rc(\vec{S},t) dt
+ \sum_{i=1}^d
\sigma_i S_i(t)
\frac{\partial c}{\partial s_i}(\vec{S},t)
dW_i(t),
\end{equation}
where $\vec{S}$ satisfies the SDE (\ref{eq:SDE_S}), and $r$, $\sigma_i$ and $dW_i(t)$ are the same as in (\ref{eq:SDE_S}).
\end{theorem}

\begin{proof}
We refer interested readers to the proof in \citet{el1997backward} and \citet{leentvaar2008pricing}, which uses Ito's lemma. 
\end{proof}

The significance of the BSDE formulation (\ref{eq:SDE_c}) is two-fold. One is that it correlates the price $c(\vec{s},t)$ with the delta $\vec{\nabla} c (\vec{s},t)$. If the price is solved correctly, then (\ref{eq:SDE_c}) simultaneously yields the correct delta. A simultaneously correct evaluation of the price and the delta is essential for performing a complete hedging process.

The other significance is that the BSDE formulation allows a less expensive and more manageable neural network approach compared to other formulations. In fact, other than the BSDE formulation, American option problems can also be formulated as a Hamilton-Jacobi-Bellman partial differential equation (PDE) based on the Black-Scholes theory. \citet{sirignano2018dgm} considers a neural network approach for solving the PDE, which involves computing Hessian tensors. Unfortunately, a Hessian tensor is an $O(Md^2)$ tensor, where $M$ is the number of samples for a neural network. When $d$ is high, a Hessian tensor can be expensive to compute and store.
In addition, given a neural network, the automatic differentiation of a Hessian is nearly impossible to derive, which makes it difficult to implement using existing deep learning libraries.
However, unlike the PDE formulation, the BSDE formulation (\ref{eq:SDE_c}) does not contain a Hessian, which avoids the computation and storage of Hessian tensors. Instead, it only requires computing price tensors of size $O(M)$ and delta tensors of size $O(Md)$. In addition, delta tensors can be easily evaluated by the built-in automatic differentiation of Tensorflow \citep{abadi2016tensorflow}, i.e., ``tf.gradients".

In this paper, we use an Euler timestepping Monte Carlo method to simulate the SDEs (\ref{eq:SDE_S}) and the BSDE (\ref{eq:SDE_c}). Let $m = 1, \cdots, M$ be the indices of simulation paths, $n=0, \cdots, N$ be the indices of discrete timesteps from 0 to $T$, $\Delta t = \frac{T}{N}$, $t^n = n\Delta t$ be the timesteps, and $(\Delta W_i)^n_m = 
\sum_{j=1}^d L_{ij} (\phi_j)^n_m \sqrt{\Delta t}$.
We discretize (\ref{eq:SDE_S}) as
\begin{align}
\label{eq:Euler_S0}
&
(S_i)^0_m = s_i^0,
&
i = 1, \cdots, d;
\\
\label{eq:Euler_S}
&
(S_i)^{n+1}_m
= (1 + (r-\delta_i)\Delta t) (S_i)^n_m
+ \sigma_i (S_i)^n_m (\Delta W_i)^n_m,
&
n = 0, \cdots, N-1, \;\;
i = 1, \cdots, d.
\end{align}

We also discretize (\ref{eq:SDE_c}) as
\begin{equation}
\label{eq:Euler_c}
c(\vec{S}^{n+1}_m, t^{n+1})
= (1 + r\Delta t) c(\vec{S}^n_m, t^n)
+ \sum_{i=1}^d
\sigma_i (S_i)^n_m
\frac{\partial c}{\partial s_i}(\vec{S}^n_m, t^n)
(\Delta W_i)^n_m,
\quad
n = N-1, \cdots, 0.
\end{equation}

Theorem \ref{thm:BSDE} assumes that an American option is not exercised at time $[t, t + dt]$. More generally, if we allow the option to be exercised at any time after $t$, then we can replace
$c(\vec{S}^{n+1}_m, t^{n+1})$
on the left hand side of (\ref{eq:Euler_c}) by
$v(\vec{S}^{n+1}_m, t^{n+1})$. In addition, we add the expiry condition $v(\vec{s}, T) = f(\vec{s})$ into the discretization. This yields a complete discretized system for the BSDE:
\begin{align}
\label{eq:Euler_vT}
&
v(\vec{S}^N_m, t^N) = f(\vec{S}^N_m),
\hspace{8.7cm}
n = N.
\\
\label{eq:Euler_v}
&
\text{Solve }
(1 + r\Delta t) c(\vec{S}^n_m, t^n)
+ \sum_{i=1}^d
\sigma_i (S_i)^n_m
\frac{\partial c}{\partial s_i}(\vec{S}^n_m, t^n)
(\Delta W_i)^n_m,
= v(\vec{S}^{n+1}_m, t^{n+1})
\text{ for } c(\vec{S}^n_m, t^n),
\\
\label{eq:Euler_v1}
&
\text{and then compute }
v(\vec{S}^n_m, t^n)
= \max\left[
c(\vec{S}^n_m, t^n),
f(\vec{S}^n_m)
\right],
\hspace{2.8cm}
n = N-1, \cdots, 0.
\end{align}

To sketch the idea of solving the discretized BSDE, let (\ref{eq:Euler_S0})-(\ref{eq:Euler_S}) generate samples of underlying asset prices $\{\vec{S}^n_m\}$ for all $n$'s and $m$'s. Then one starts with $n=N$, computes the expiry condition (\ref{eq:Euler_vT}), and then performs backward timestepping from $n=N-1$ to $n=0$ using (\ref{eq:Euler_v})-(\ref{eq:Euler_v1}), which yields $\{v(\vec{S}^n_m, t^n)\}$ for all $n$'s and $m$'s. Eventually, at $n=0$, noting that $\vec{S}^0_m=\vec{s}^0$ by (\ref{eq:Euler_S0}), we obtain the option price $v(\vec{s}^0,0)$ and the option delta $\vec{\nabla} v (\vec{s}^0,0)$.

\subsection{Least squares solution for the discretized BSDE}
\label{subsec:LeastSq}

Consider only the $n$-th timestep $t^n$, and introduce a short notation for the corresponding price and delta functions as
$v^n(\vec{s}) \equiv v(\vec{s},t^n)$
and
$\vec{\nabla} v^n(\vec{s})
\equiv \vec{\nabla} v(\vec{s},t^n)$.
Solving (\ref{eq:Euler_v}) requires finding a $d$-dimensional function $c^n(\vec{s})$ where both the function $c^n(\vec{s})$ itself and its derivative $\vec{\nabla} c^n(\vec{s})$ satisfy (\ref{eq:Euler_v}). This is challenging, especially when $d$ is large.

In this paper, we consider finding an approximation of the continuous price function. We let the approximation satisfy (\ref{eq:Euler_v}) in a least squares sense. More specifically, define the residual of (\ref{eq:Euler_v}) as the difference between the left and right hand sides:
\begin{equation}
\label{eq:LeastSq1}
\begin{array}{r}
\displaystyle
\mathcal{R}[c^n]_m
\equiv
(1 + r\Delta t) c^n(\vec{S}^n_m)
+ \sum_{i=1}^d
\sigma_i (S_i)^n_m
\frac{\partial c^n}{\partial s_i}(\vec{S}^n_m)
(\Delta W_i)^n_m
-v^{n+1}(\vec{S}^{n+1}_m),
\\
m = 1, \cdots, M.
\end{array}
\end{equation}
Then our goal is to find an approximation $y^n$ to the actual continuation function $c^n$ that minimizes the least squares residual:
\begin{equation}
\label{eq:LeastSq2}
c^n
\approx
(y^n)^*
\equiv
\argmin_{y^n}
\left(
\sum_{m=1}^M
\mathcal{R}[y^n]_m^2
\right).
\end{equation}


\section{Neural Network Formulation}
\label{sec:NN}

Finding the optimal approximate function in the least squares sense (\ref{eq:LeastSq2}) is non-trivial. One approach is to use a parameterized function to represent the approximate function $y^n$. Then the optimization problem in terms of function space is converted to the optimization problem in terms of parameter space, which is more manageable.

One well-known example of the parameterized approach is the Longstaff-Schwartz method \citep{longstaff2001valuing}. More specifically, the continuation price is approximated by a $\chi$-th degree polynomial. We note that unlike our approach, their objective is not to minimize the least squares residual of the BSDE (\ref{eq:LeastSq1}), but to minimize the least squares difference between the discounted payoffs and the parameterized polynomials. In practical implementation of the Longstaff-Schwartz method, we let $\chi\ll d$, which means that the number of the polynomial basis is
$\left(
\begin{smallmatrix}
d+\chi \\ d
\end{smallmatrix}
\right)
\approx \frac{1}{\chi!}d^{\chi}$.
However, convergence of the Longstaff-Schwartz method to the exact American option prices requires the number of the basis tending to infinity, i.e., $\chi\to\infty$ \citep{longstaff2001valuing,stentoft2004convergence}, which results in a high computational cost. In addition, a pre-defined, static polynomial basis may not be the optimal choice for American options.

\subsection{Sequence of neural networks}
\label{subsec:NN}

Our approach is to use neural networks to represent the approximate continuation price function $y^n$. A neural network is a nonlinear parameterization where the basis is dynamic, i.e., the optimal basis is learned during the training process \citep{goodfellow2016deep}. The main advantage of neural network formulation is that the complexity does not grow exponentially with the dimension $d$.

The architecture of neural network determines the proximity between the global minimum of the loss function and the true underlying price function, the landscape of the loss function, and the level of difficulty for optimization algorithms to find the global minimum. These directly impact the accuracy of the approximate price function. There exist many neural network architectures, such as feedforward, convolutional, or recurrent networks. We refer interested readers to \citet{goodfellow2016deep} for a review of these standard network architectures. However, these standard networks are not designed for solving American option problems.

In this paper, we propose a sequence of $N$ networks
$\{y^n (\vec{s};\Omega^n)
\,|\,
n = N-1, \cdots, 1,0\}$,
where $\Omega^n$ is the trainable parameter set of the $n$-th network. Each individual network $y^n (\vec{s};\Omega^n)$ approximates the price function at the $n$-th timestep $c^n(\vec{s})$. The design of each individual network is motivated by the fact that the approximate function of the $n$-th timestep,
$y^n (\vec{s};\Omega^n)$,
should differ from
$y^{n+1} (\vec{s};\Omega^{n+1})$
by a function of magnitude $O(\Delta t)$. Mathematically, it means that
\begin{align}
\label{eq:NN_def_payoff1}
&
y^N (\vec{s})
= f (\vec{s}),
&
&
n = N;
\\
\label{eq:NN_def_recursive_1}
&
y^n (\vec{s};\Omega^n)
= y^{n+1} (\vec{s};\Omega^{n+1})
+ \Delta t \cdot \mathcal{F}(\vec{s};\Omega^n),
&
&
n = N-1, \cdots, 0;
\end{align}
where
$\mathcal{F}(\vec{s};\Omega^n)$
is the difference between the approximate functions at the two adjacent timesteps, or the ``remainder" that we aim to find. We note that the sequence of networks (\ref{eq:NN_def_recursive_1}) is defined in a recursive sense. In addition, the sequence of networks is backward in time, i.e., the timestep $n$ decreases from $N-1$ to 0. Hence, in this paper, we use the ``previous", ``current" and ``next" timesteps to refer to the $(n+1)$-th, $n$-th and $(n-1)$-th timesteps, respectively.

Regarding each remainder network
$\mathcal{F}(\vec{s};\Omega^n)$,
we parameterize it by an $L$-layer feedforward network with batch normalizations. In the following part, we drop the timestep index $n$ temporarily, and use superscript with square brackets for the layer index $l = 0, \cdots, L$. Let the dimensions of the layers be
$\{d^{[l]} \, | \, l = 0, \cdots, L\}$.
Let the input of the neural network be
$\vec{x}^{[0]} = \vec{s} \in \mathbb{R}^{d^{[0]}}$,
where the input dimension is $d^{[0]}=d$.
Then we construct an $L$-layer feedforward neural network as follows:
\begin{itemize}[leftmargin=10mm]
\item For the hidden layers, $l = 1, \cdots, L$:
\begin{align}
\label{eq:NN_layer_linear}
\text{linear transformation: }
&
\vec{z}^{[l]}
= \mathbf{W}^{[l]} \cdot \vec{x}^{[l-1]},
\\
\label{eq:NN_layer_batch}
\text{batch normalization: }
&
\vec{h}^{[l]}
= \text{bnorm}(\vec{z}^{[l]};
\vec{\beta}^{[l]}, \vec{\gamma}^{[l]},
\vec{\mu}^{[l]}, \vec{\sigma}^{[l]}),
\\
\label{eq:NN_layer_activate}
\text{rectified linear unit activation: }
&
\vec{x}^{[l]}
= \max(\vec{h}^{[l]}, 0),
\end{align}
where
\begin{equation}
\label{eq:NN_layer_batch2}
\text{bnorm}(\vec{x};
\vec{\beta}, \vec{\gamma},
\vec{\mu}, \vec{\sigma})
\equiv
\vec{\gamma}\cdot
\frac{\vec{x}-\vec{\mu}}{\vec{\sigma}}
+ \vec{\beta}
\end{equation}
is the batch normalization operator,
$\vec{x}^{[l]},\vec{z}^{[l]},\vec{h}^{[l]}
\in\mathbb{R}^{d^{[l]}}$
are hidden layer variables,
$\mathbf{W}^{[l]}
\in\mathbb{R}^{d^{[l]}\times d^{[l-1]}}$
are trainable weights,
$\vec{\mu}^{[l]},\vec{\sigma}^{[l]}
\in\mathbb{R}^{d^{[l]}}$
are moving averages of batch means and standard deviations, and
$\vec{\gamma}^{[l]},\vec{\beta}^{[l]}
\in\mathbb{R}^{d^{[l]}}$
are trainable scales and offsets.
The operations in (\ref{eq:NN_layer_batch})-(\ref{eq:NN_layer_batch2}) are evaluated element-wise. For instance, (\ref{eq:NN_layer_activate}) means
$
x_i^{[l]}
= \max(h_i^{[l]}, 0)
$
for all $i=1,\cdots, d^{[l]}$.
\item For the output layer:
\begin{equation}
\label{eq:NN_layer_output}
\mathcal{F}(\vec{s};\Omega^n)
= \vec{\omega} \cdot \vec{x}^{[L]} + b,
\end{equation}
where
$\vec{\omega}\in\mathbb{R}^{d^{[L]}}$,
$b\in\mathbb{R}$ are trainable weight and bias.
\end{itemize}

In addition, we propose adding a scaling parameter $\alpha^n$ to each neural network (\ref{eq:NN_def_recursive_1}) and revise it as
\begin{equation}
\label{eq:NN_def_recursive_2}
y^n (\vec{s};\,\Omega^n)
= \alpha^n \left[
y^{n+1} (\vec{s};\,\Omega^{n+1})
+ \Delta t \cdot \mathcal{F}(\vec{s};\Omega^n)
\right],
\quad
n = N-1, \cdots, 0.
\end{equation}
We let $\alpha^n$ be trainable, or equivalently, $\alpha^n\in \Omega^n$. $\alpha^n$ is initialized as 1 before training, and is close to 1 during and after training. Introducing the trainable parameter $\alpha^n$ expands the function space the neural network can represent. A neural network with a larger function space is less likely to underfit, and thus more likely to have an accurate training result \citep{goodfellow2016deep}.

We remark that our proposed recursive architecture (\ref{eq:NN_def_recursive_2}) is different from the other architectures in the literature, particularly \citet{sirignano2018dgm}, where one single neural network is used to represent the spacetime price function. To justify our choice of the recursive architecture, we note that the true price functions $c^{n+1} (\vec{s})$ and $c^{n} (\vec{s})$ differ by a function of magnitude $O(\Delta t)$. In (\ref{eq:NN_def_recursive_2}), if we let $y^{n+1} (\vec{s};\Omega^{n+1})\approx c^{n+1} (\vec{s})$ and $\alpha^n \approx 1$, then regardless of the value of $\mathcal{F}(\vec{s};\Omega^n)$, $y^n (\vec{s};\Omega^n)$ will only differ from the true price function $c^n (\vec{s})$ by a magnitude of $O(\Delta t)$. Hence, before training starts, $y^n (\vec{s};\Omega^n)$ is already a good approximation of $c^n (\vec{s})$. This makes it more likely for the training to find the optimal solution that (almost) equals $c^n (\vec{s})$. Therefore, the recursive architecture is critical to the accuracy of the resulting prices and deltas.

\subsection{Smoothing payoff functions}
\label{subsec:smooth}

We note that most of the payoff functions in practical applications have the form of (\ref{eq:payoff}), which is not differentiable at $g(\vec{s})=0$. In other words, $y^{N} (\vec{s})$ in (\ref{eq:NN_def_payoff1}) is not differentiable. However, $y^{N-1} (\vec{s};\Omega^{N-1})$ as an approximation of the continuation price function is differentiable. Consequentially, the left and right hand sides of
$
y^{N-1} (\vec{s};\,\Omega^{N-1})
= \alpha^{N-1} \left[
y^{N} (\vec{s})
+ \Delta t \cdot \mathcal{F}(\vec{s};\Omega^{N-1})
\right],
$
are inconsistent in terms of differentiability. 
Such inconsistency makes it difficult to learn an accurate $\mathcal{F}(\vec{s};\Omega^{N-1})$, which negatively affects the accuracy of the trained $y^{N-1} (\vec{s};\Omega^{N-1})$, and furthermore, the accuracy of the trained $y^n (\vec{s};\Omega^n)$ in the subsequent timesteps. In this paper, we propose smoothing the function $y^{N} (\vec{s})$ in (\ref{eq:NN_def_payoff1}) as follows:
\begin{equation}
\label{eq:NN_def_payoff2}
y^{N} (\vec{s})
= f_{\kappa}(\vec{s})
\equiv \frac{1}{\kappa}
\ln\left(
1 + e^{\kappa g(\vec{s})}
\right),
\end{equation}
where $\kappa$ is a user-defined parameter. The operations in (\ref{eq:NN_def_payoff2}) are evaluated element-wise. $f_{\kappa}(\vec{s})$ converges to $f(\vec{s})$ when $\kappa\to\infty$, and is a good approximation of $f(\vec{s})$ when $\kappa$ is large. The significance of (\ref{eq:NN_def_payoff2}) is that $f_{\kappa}(\vec{s})$ is differentiable, which makes it easier to train an accurate $\mathcal{F}(\vec{s};\Omega^{N-1})$. In practice, we choose $\kappa=\frac{2}{\Delta t}$.
We note that smoothing payoff functions is a standard technique in the literature of binomial trees for option pricing \citep{heston2000rate}. However, to the best of our knowledge, this paper is the first to propose smoothing payoff functions among the literature of neural networks for option pricing.

\subsection{Feature selection}
\label{subsec:feature}

Feature selection, i.e., choosing the correct input features based on domain knowledge, has a great impact on the accuracy of neural network models \citep{goodfellow2016deep}. Naively one can simply set the input as the underlying asset prices $\vec{x}^{[0]} = \vec{s}$. In this paper, we consider adding two new features.

One new feature is the payoff function. It is suggested in \citet{kohler2010review} and \citet{firth2005high} that including the payoff in the nonlinear basis can improve the accuracy of the regression-based algorithms. In this paper, we consider using $g(\vec{s})$ in (\ref{eq:payoff}) as an input feature. The reason of using $g(\vec{s})$ rather than $f(\vec{s})$ is that the maximum operator in (\ref{eq:payoff}) is irreversible. In other words, $f(\vec{s})$ can be computed by $g(\vec{s})$ but not conversely. Hence, using $g(\vec{s})$ as the input contains more information than $f(\vec{s})$. The additional maximum operator in (\ref{eq:payoff}) can be learned by the activation function (\ref{eq:NN_layer_activate}) in the network.

The other new feature we consider adding is the output price function from the previous timestep, i.e., $y^{n+1} (\vec{s};\,\Omega^{n+1})$ in (\ref{eq:NN_def_recursive_2}). The intuition is that the solution at the $n$-th step should look similar to the solution at the $(n+1)$-th step. We note that this feature is similar but not exactly the same as the payoff function, which makes it useful as an additional feature. More specifically, when $n\approx N$, $y^{n+1}$ is approximately the same as but slightly smoother than the payoff function; when $n\ll N$, $y^{n+1}$ can be very different from the payoff function.

The accuracy of neural network models can be further improved by input normalization \citep{sola1997importance}. Effectively, we can combine the implementation of feature selection and input normalization by adding the following ``input layer" (denoted as $l=0$) before the hidden layer $l=1$:
\begin{align}
\label{eq:NN_layer_input_concat}
\text{feature concatenation: }
&
\vec{z}^{[0]}
= \left(
\vec{s}, \, g(\vec{s}), \,
y^{n+1} (\vec{s};\,\Omega^{n+1})
\right)^T \in \mathbb{R}^{d^{[0]}},
\\
\label{eq:NN_layer_input_batch}
\text{input normalization: }
&
\vec{x}^{[0]}
= \text{bnorm}(\vec{z}^{[0]};
\vec{\beta}^{[0]}, \vec{\gamma}^{[0]},
\vec{\mu}^{[0]}, \vec{\sigma}^{[0]}),
\end{align}
where the input dimension is changed to $d^{[0]}=d+2$ after the concatenation. We note that $\vec{\mu}^{[0]}$ and $\vec{\sigma}^{[0]}$ can be pre-computed from the entire training dataset, unlike $\vec{\mu}^{[l]}$ and $\vec{\sigma}^{[l]}$ in the hidden layers that are computed by moving averages of training batches.

To summarize Sections \ref{subsec:NN}-\ref{subsec:feature}, the architecture of the proposed neural network framework is defined by (\ref{eq:NN_def_recursive_2}) and (\ref{eq:NN_def_payoff2}), where the remainder network at each timestep
$\mathcal{F}(\vec{s};\Omega^n)$
is defined by the input layer (\ref{eq:NN_layer_input_concat})-(\ref{eq:NN_layer_input_batch}), the hidden layers (\ref{eq:NN_layer_linear})-(\ref{eq:NN_layer_activate}) and the output layer (\ref{eq:NN_layer_output}). The trainable parameters of the neural network framework are
$\{\Omega^n \, | \, n = N-1, \cdots, 0\}$,
where
\begin{equation}
\Omega^n \equiv
\{
(\mathbf{W}^{[l]})^n,
(\vec{\gamma}^{[l]})^n, (\vec{\beta}^{[l]})^n,
(\vec{\gamma}^{[0]})^n, (\vec{\beta}^{[0]})^n,
\vec{\omega}^n, b^n, \alpha^n
\,|\,
L = 1, \cdots, L\}.
\end{equation}
Figure \ref{fig:NN_1} illustrates the architecture of the proposed neural network framework.


\begin{figure}[h!]
\centering
\footnotesize
\includegraphics[height=5.3cm]{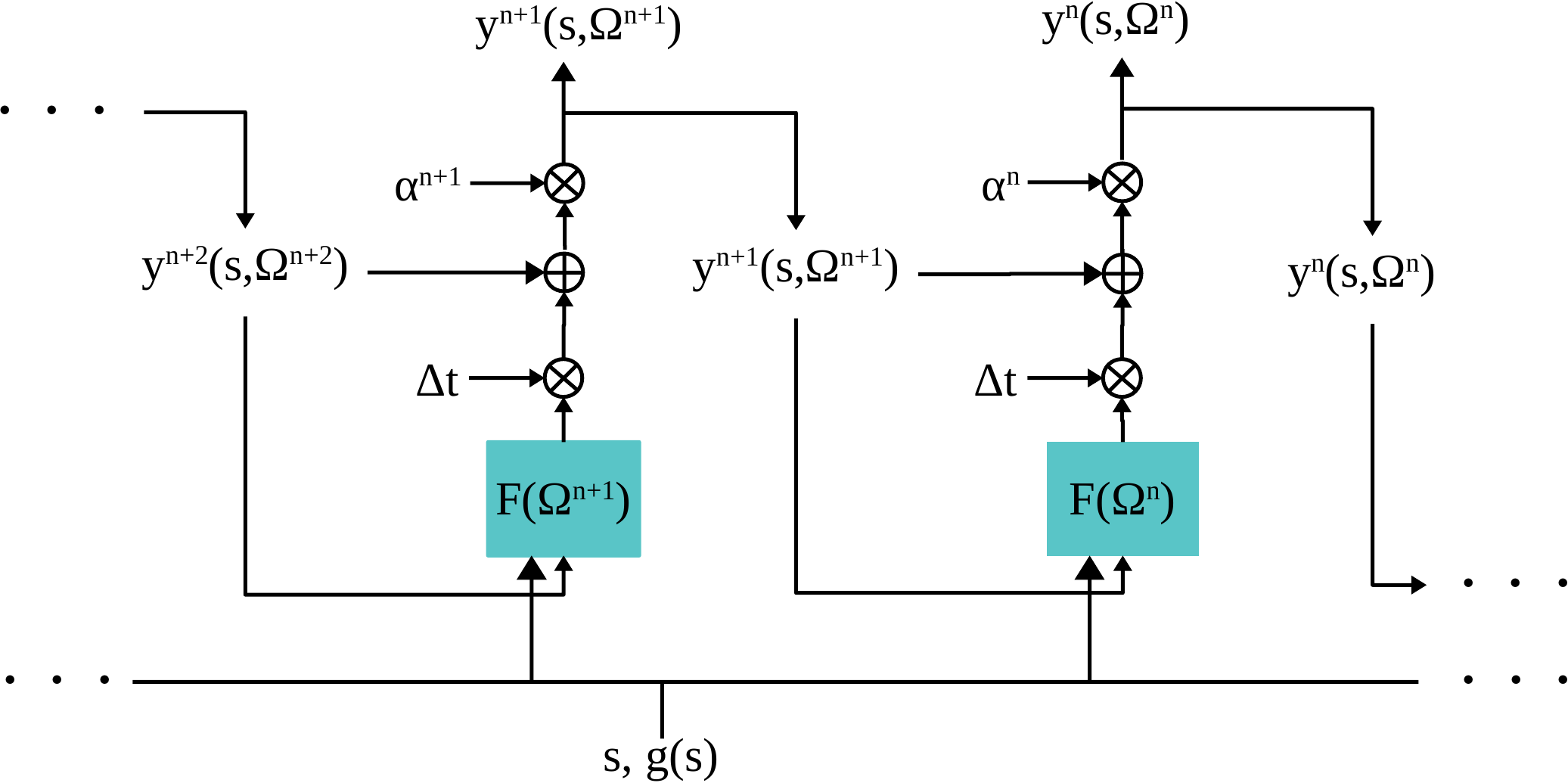}
\includegraphics[height=5.3cm]{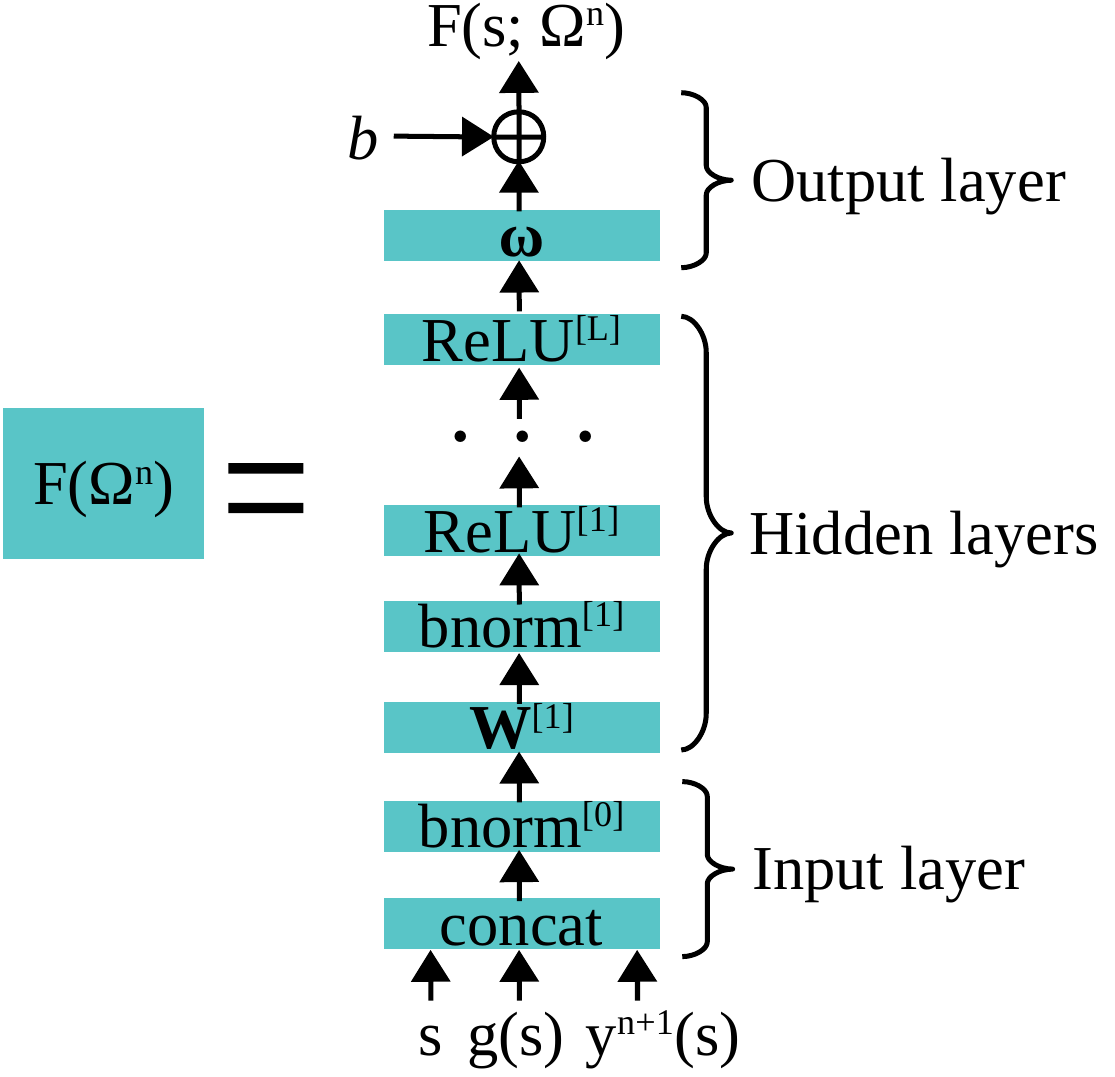}
\caption{\label{fig:NN_1}
The architecture of the proposed neural network framework defined by (\ref{eq:NN_def_recursive_2}) and (\ref{eq:NN_def_payoff2}), where the remainder network at each timestep $\mathcal{F}(\vec{s};\Omega^n)$ is defined by the input layer (\ref{eq:NN_layer_input_concat})-(\ref{eq:NN_layer_input_batch}), the hidden layers (\ref{eq:NN_layer_linear})-(\ref{eq:NN_layer_activate}) and the output layer (\ref{eq:NN_layer_output}). The symbols $\otimes$ and $\oplus$ represent multiplication and addition, respectively.
}
\end{figure}


\subsection{More efficient neural network sequence}
\label{subsec:NN2}

We discussed the advantage of the recursive architecture (\ref{eq:NN_def_recursive_2}) at the end of Section \ref{subsec:NN}. However, the recursive architecture is expensive when $N$ is large. More specifically, consider the $0$-th timestep, and consider using the sequence of the neural networks to compute the value of $y^0(\vec{s})$. By applying the recursive relation (\ref{eq:NN_def_recursive_2}), we have
\begin{equation}
\label{eq:y0_1}
y^0 (\vec{s})
= y^N(\vec{s})
+ \Delta t \cdot \sum_{\nu=1}^{N}
\mathcal{F}(\vec{s};\Omega^{N-\nu}),
\end{equation}
where for simplicity we set $\alpha^n=1$ for all timesteps. Equation (\ref{eq:y0_1}) shows that the computation of $y^0 (\vec{s})$ requires going through $N$ feedforward networks.

Here we propose a modified neural network architecture to reduce the computational cost. In Section \ref{subsec:NN}, we motivated the recursive relation (\ref{eq:NN_def_recursive_2}) based on the fact that the outputs of the two adjacent timesteps, $y^n(\vec{s})$ and $y^{n+1}(\vec{s})$, should differ by a function of magnitude $O(\Delta t)$. In fact, we can generalize this relation to any two timesteps $n$ and $n+j$ where $j \ll N$. That is, the outputs $y^n(\vec{s})$ and $y^{n+j}(\vec{s})$ should differ by a function of magnitude $O(\Delta t)$. Similar to (\ref{eq:NN_def_recursive_2}), we formulate this idea into the following recursive relation:
\begin{equation}
y^{n} (\vec{s};\,\Omega^{n})
= \alpha^{n} \left[
y^{n+j} (\vec{s};\,\Omega^{n+j})
+ j\Delta t \cdot \mathcal{F}(\vec{s};\Omega^{n})
\right].
\end{equation}
This generalization allows us to recur the feedforward networks at every few timesteps, rather than at every single timestep, and thus reduces the computational cost.

To be more precise, if we recur the feedforward networks at every $J$ timesteps ($J\ll N$), then we modify the sequence of the neural networks (\ref{eq:NN_def_recursive_2}) as follows:
\begin{equation}
\label{eq:NN_def_recursive_3}
\begin{array}{rl}
y^{n} (\vec{s};\,\Omega^{n})
= \alpha^{n} \left[
y^{n+\eta} (\vec{s};\,\Omega^{n+\eta})
+ \eta\Delta t \cdot \mathcal{F}(\vec{s};\Omega^{n})
\right],
&
\\
\text{where }
\eta \equiv [(N-n-1) \text{ mod } J] + 1,
&
n = N-1, \cdots, 0.
\end{array}
\end{equation}
Equivalently, we can enumerate (\ref{eq:NN_def_recursive_3}) as
\begin{equation}
\begin{array}{ll}
\text{at the $(n-1)$-th step:}
&
y^{n-1} (\vec{s};\,\Omega^{n-1})
\\
&
\,
= \alpha^{n-1} \left[
y^n (\vec{s};\,\Omega^n)
+ \Delta t \cdot \mathcal{F}(\vec{s};\Omega^{n-1})
\right],
\\
& \vdots
\\
\text{at the $(n-j)$-th step:}
&
y^{n-j} (\vec{s};\,\Omega^{n-j})
\\
&
\,
= \alpha^{n-j} \left[
y^n (\vec{s};\,\Omega^n)
+ j\Delta t \cdot \mathcal{F}(\vec{s};\Omega^{n-j})
\right],
\\
& \vdots
\\
\text{at the $(n-J)$-th step:}
&
y^{n-J} (\vec{s};\,\Omega^{n-J})
\\
&
\,
= \alpha^{n-J} \left[
y^n (\vec{s};\,\Omega^n)
+ J\Delta t \cdot \mathcal{F}(\vec{s};\Omega^{n-J})
\right],
\end{array}
\end{equation}
where $1\leq j\leq J$ and $n=N, N-J, N-2J, \cdots$. We remark that (\ref{eq:NN_def_recursive_2}) is simply a special case of (\ref{eq:NN_def_recursive_3}) with $J=1$. Figure \ref{fig:NN_2} illustrates the modified architecture with $J=3$. Readers can generalize the idea of Figure \ref{fig:NN_2} to any $J\ll N$.


\begin{figure}[t!]
\centering
\footnotesize
\includegraphics[height=5.3cm]{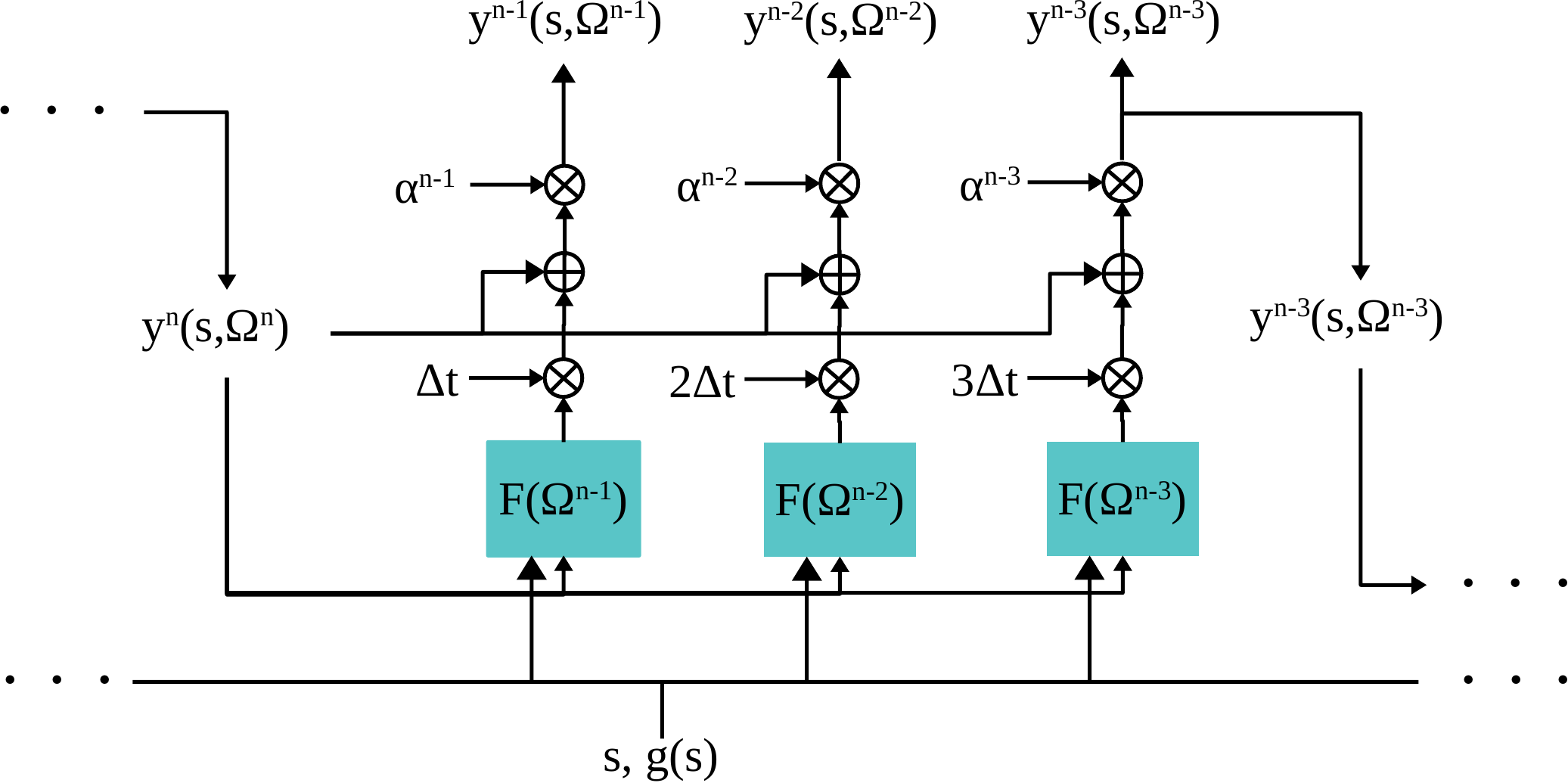}
\caption{\label{fig:NN_2}
The modified architecture of the proposed neural network framework defined by (\ref{eq:NN_def_payoff2}) and (\ref{eq:NN_def_recursive_3}), where $J=3$. Similar to Figure \ref{fig:NN_1}, the remainder network at each timestep $\mathcal{F}(\vec{s};\Omega^n)$ is defined by the input layer (\ref{eq:NN_layer_input_concat})-(\ref{eq:NN_layer_input_batch}), the hidden layers (\ref{eq:NN_layer_linear})-(\ref{eq:NN_layer_activate}) and the output layer (\ref{eq:NN_layer_output}).
}
\end{figure}


Regarding the choice of $J$, smaller $J$ yields more precise trained $y^n$ with higher computational cost; larger $J$ is computationally cheaper but the trained $y^n$ is less precise. In our numerical simulations, we choose $N=100$ and $J=4$.

To give an example of how the modified architecture reduces the computational cost, let us reconsider evaluating $y^0(\vec{s})$. By applying the recursive relation (\ref{eq:NN_def_recursive_3}), we have
\begin{equation}
\label{eq:y0_2}
y^0 (\vec{s})
= y^N(\vec{s})
+ J \Delta t \cdot \sum_{\nu=1}^{\lfloor N/J \rfloor}
\mathcal{F}(\vec{s};\Omega^{N-\nu J})
+ (N \text{ mod } J) \Delta t
\cdot \mathcal{F}(\vec{s};\Omega^0),
\end{equation}
where for simplicity we set $\alpha^n=1$ for all timesteps. Compared to (\ref{eq:y0_1}), using (\ref{eq:y0_2}) to compute $y^0 (\vec{s})$ only requires going through $\lceil N/J \rceil$ feedforward networks. In other words, the computation is $J$ times cheaper.

%

\subsection{Training the neural network}
\label{subsec:train}

Consider training the network at the $n$-th timestep for solving (\ref{eq:LeastSq1})-(\ref{eq:LeastSq2}). The training inputs are
\begin{equation}
\label{eq:train_input}
\{\vec{S}^n_m, \,
\Delta \vec{W}^n_m, \,
v^{n+1}(\vec{S}^{n+1}_m), \,
g(\vec{S}^n_m), \,
y^{n+\eta}(\vec{S}^n_m;(\Omega^{n+\eta})^*), \,
\vec{\nabla} y^{n+\eta}(\vec{S}^n_m;(\Omega^{n+\eta})^*)
\,|\, \forall m\},
\end{equation}
where the first three inputs are the required inputs of (\ref{eq:LeastSq1}), the last three inputs are the features introduced in Section \ref{subsec:feature}, $y^{n+\eta}$ is defined in (\ref{eq:NN_def_recursive_3}) and $(\Omega^{n+\eta})^*$ is the trained parameters from the previous timestep $n+\eta$. The training output is
$\{y^{n}(\vec{S}^n_m;\,\Omega^{n}), \,
\vec{\nabla} y^{n}(\vec{S}^n_m;\,\Omega^{n})
\,|\, \forall m\}$.
The loss function of the network is given by (\ref{eq:LeastSq1})-(\ref{eq:LeastSq2}), i.e., the least squares BSDE residual, which we rewrite as a function of the trainable parameters $\Omega^n$:
\begin{equation}
\label{eq:LeastSq3}
\begin{array}{rl}
\displaystyle
\mathcal{L}[\Omega^n]
\equiv
\sum_{m=1}^M
&
\displaystyle
\left[
(1 + r\Delta t) y^n(\vec{S}^n_m; \Omega^n)
\right.
\\
&
\displaystyle
\left.
+ \sum_{i=1}^d
\sigma_i (S_i)^n_m
\frac{\partial y^n}{\partial s_i}(\vec{S}^n_m; \Omega^n)
(\Delta W_i)^n_m
-v^{n+1}(\vec{S}^{n+1}_m)
\right]^2.
\end{array}
\end{equation}
We consider using the popular Adam optimizer \citep{kingma2014adam} to minimize the loss function (\ref{eq:LeastSq3}), which yields the set of optimal trainable parameters
\begin{equation}
\label{eq:LeastSq4}
(\Omega^n)^*
\equiv
\argmin_{\Omega^n} \,
\mathcal{L}[\Omega^n].
\end{equation}
Then, using the trained neural network, we can compute the estimated option price $y^n(\vec{s};\,(\Omega^n)^*)$ and delta $\vec{\nabla} y^n (\vec{s};\,(\Omega^n)^*)$.
In addition, we use the estimated option price to determine the exercise boundary as
\begin{equation}
\label{eq:exercise}
\xi^n(\vec{s})
= \left\{
\begin{array}{ll}
\text{continued,}
& \text{if }
y^n(\vec{s};\,(\Omega^n)^*) > f(\vec{s}),
\\
\text{exercised,}
& \text{otherwise.}
\end{array}
\right.
\end{equation}

In order to ensure the accuracy of training, we follow suggested good practices in the deep learning community \citep{goodfellow2016deep}. For instance, mini-batch optimization is used; the learning rate of the Adam optimizer is decayed to ensure convergence; gradient clipping is applied to avoid exploding gradients. In particular, we let the number of training steps be 600. At the $s$-th training step ($0\leq s\leq 600$), we let the moving average rate for $\vec{\mu}^{[l]}$ and $\vec{\sigma}^{[l]}$ in (\ref{eq:NN_layer_batch}) be
$\frac{1}{0.99}(0.01^{\max(\min(s/350,1),0)} - 0.01)$,
and let the learning rate for the Adam optimizer be
$0.01 \times 0.001^{\max(\min((s-150)/350,1),0)}$.


\section{Improving the Algorithm}
\label{sec:key}

Sections \ref{sec:BSDE}-\ref{sec:NN} describe the foundation of our algorithm. This section introduces a few techniques that improve the accuracy of resulting prices and deltas and the efficiency of the algorithm.

\subsection{The training input ``v"}
\label{subsec:define_v}

Consider the $n$-th timestep. The definition of $v^{n+1}(\vec{S}^{n+1}_m)$ in the training input (\ref{eq:train_input}) turns out to play a significant role in the accuracy of the trained continuation price $y^n$. More specifically, if the training input $v^{n+1}(\vec{S}^{n+1}_m)$ is incorrectly defined, which means that we feed incorrect values to the right hand side of (\ref{eq:Euler_v}), then the trained network $y^n$ would not represent the correct $c^n$.

Finding the correct definition of $v^{n+1}(\vec{S}^{n+1}_m)$ turns out to be non-trivial. One natural way of defining $v^{n+1}(\vec{S}^{n+1}_m)$ is to use the output prices of the trained network. More specifically, suppose
$y^{n+1} (\vec{s};\,(\Omega^{n+1})^*)$
is already trained. Then
\begin{equation}
\label{eq:v_def1}
v^{n+1}(\vec{S}^{n+1}_m)
= \left\{
\begin{array}{ll}
y^{n+1} (\vec{S}^{n+1}_m;\,(\Omega^{n+1})^*),
&
\text{if }
\xi^{n+1}(\vec{S}^{n+1}_m) = \text{continued},
\\
f(\vec{S}^{n+1}_m),
&
\text{if }
\xi^{n+1}(\vec{S}^{n+1}_m) = \text{exercised},
\end{array}
\right.
\end{equation}
where $\xi^{n+1}$ is defined in (\ref{eq:exercise}).
However, in practice, due to the finite number of samples and training steps, training error in the network $y^{n+1}$ is inevitable, which means that $v^{n+1}(\vec{S}^{n+1}_m)$ might contain error after applying (\ref{eq:v_def1}). Consequentially, the error of the training input $v^{n+1}(\vec{S}^{n+1}_m)$ will propagate into $y^{n}$ after training the $n$-th network, and propagate into $v^{n}(\vec{S}^{n}_m)$ after applying (\ref{eq:v_def1}) again, and propagate into $y^{n-1}$, $v^{n-1}(\vec{S}^{n-1}_m)$, $y^{n-2}$, ..., after further backward timestepping. In other words, (\ref{eq:v_def1}) is not robust against the accumulation of training errors over timesteps and may result in bias.

In fact, such bias can be quantified using the following lemma:
\begin{lemma}[Quantifying bias]
\label{lm:bias}
Assume that $\{\vec{S}^\nu_m \,|\, 0\leq\nu\leq n,\forall m\}$ are known and fixed, i.e., assume that the stochastic process $\{\vec{S}^n\}$ is adapted to the filtration $\{\mathcal{F}_n\}$.
Let $\{\vec{S}^{n+1}_m \,|\, \forall m\}$ be another set generated under (\ref{eq:Euler_S}).  Then the prices $c^n$ and $v^{n+1}$ must satisfy
\begin{equation}
\label{eq:v_def2_pre}
c^n(\vec{S}^n_m)
= \mathbb{E} [e^{-r\Delta t} v^{n+1}(\vec{S}^{n+1}_m) \,|\, \vec{S}^n_m] + O(\sqrt{\Delta t}).
\end{equation}
\end{lemma}

\begin{proof}
Consider taking the conditional expectation of (\ref{eq:Euler_v}):
\begin{equation*}
(1 + r\Delta t) c^n(\vec{S}^n_m)
+ \sum_{i=1}^d
\sigma_i (S_i)^n_m
\frac{\partial c^n}{\partial s_i}(\vec{S}^n_m)
\mathbb{E} [(\Delta W_i)^n_m \,|\, \vec{S}^n_m]
= \mathbb{E} [v^{n+1}(\vec{S}^{n+1}_m) \,|\, \vec{S}^n_m] + O(\sqrt{\Delta t}),
\end{equation*}
where we add the term $O(\sqrt{\Delta t})$ at the end of the equation to reflect the discretization error of (\ref{eq:Euler_v}).
We note that $\{\vec{S}^n_m\}$, $\{c^n(\vec{S}^n_m)\}$ and $\{\frac{\partial c^n}{\partial s_i}(\vec{S}^n_m)\}$ are not random variables due to the filtration, and hence the only random variables are $\{\vec{S}^{n+1}_m\}$ and $\{(\Delta W_i)^n_m\}$.
Since
$\mathbb{E} [(\Delta W_i)^n_m \,|\, \vec{S}^n_m] = 0$
and
$1 + r\Delta t = e^{r\Delta t} + O(\Delta t^2)$,
we have
$
e^{r\Delta t} c^n(\vec{S}^n_m)
= \mathbb{E} [v^{n+1}(\vec{S}^{n+1}_m) \,|\, \vec{S}^n_m] + O(\sqrt{\Delta t})
$,
which gives (\ref{eq:v_def2_pre}).
\end{proof}

Lemma \ref{lm:bias} indicates that if $v^{n+1}(\vec{S}^{n+1}_m)$ is correctly evaluated, then $\mathbb{E} [e^{-r\Delta t} v^{n+1}(\vec{S}^{n+1}_m)]$ should match the true underlying continuation function $c^n(\vec{S}^n_m)$. After a few timesteps, if $\mathbb{E} [e^{-r\Delta t} v^{n+1}(\vec{S}^{n+1}_m)]$ deviates from $c^n(\vec{S}^n_m)$, then it indicates an accumulation of training errors from the previous timesteps.


\begin{figure}[h!]
\centering
\footnotesize
\includegraphics[height=5cm]{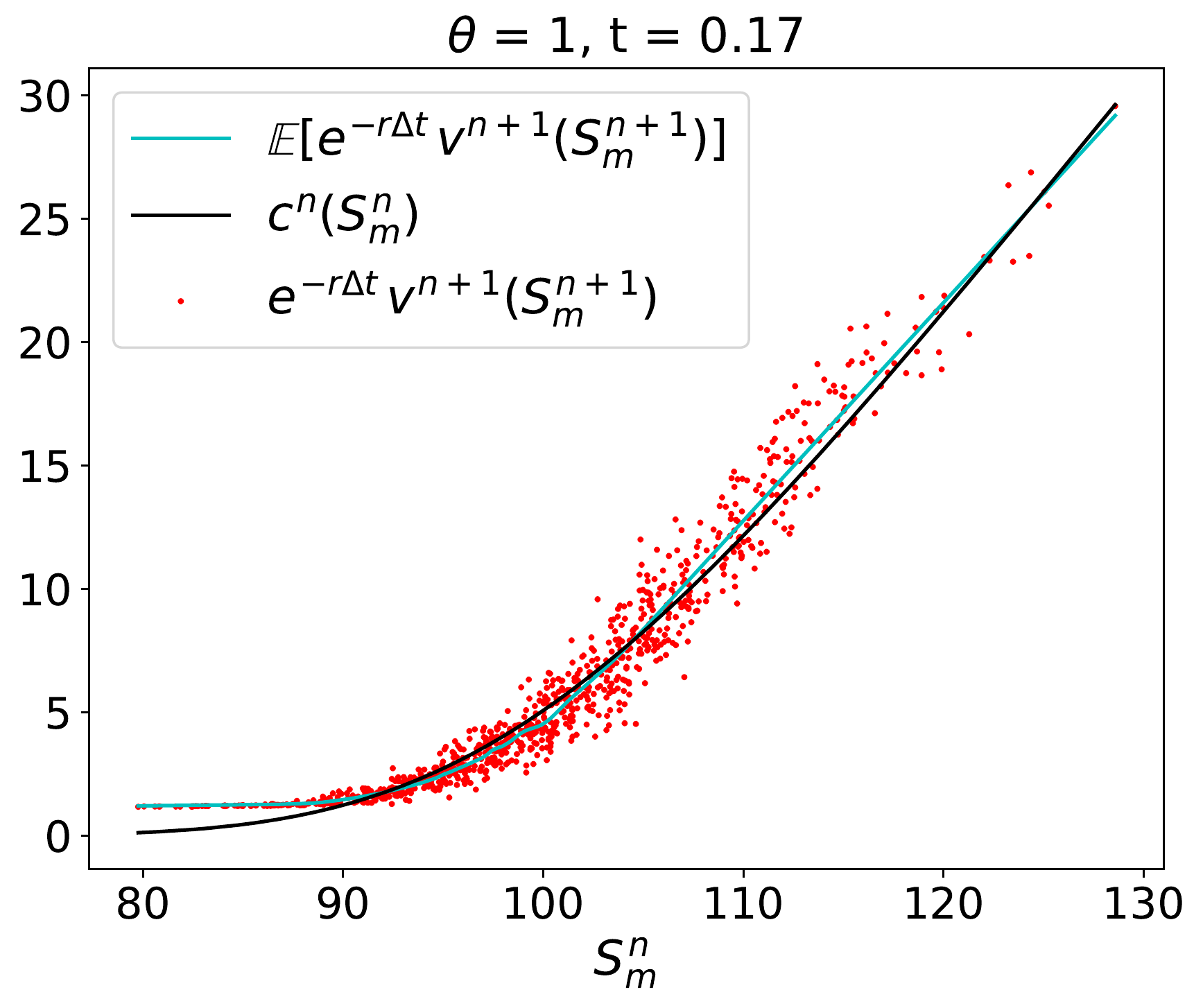}
\includegraphics[height=5cm]{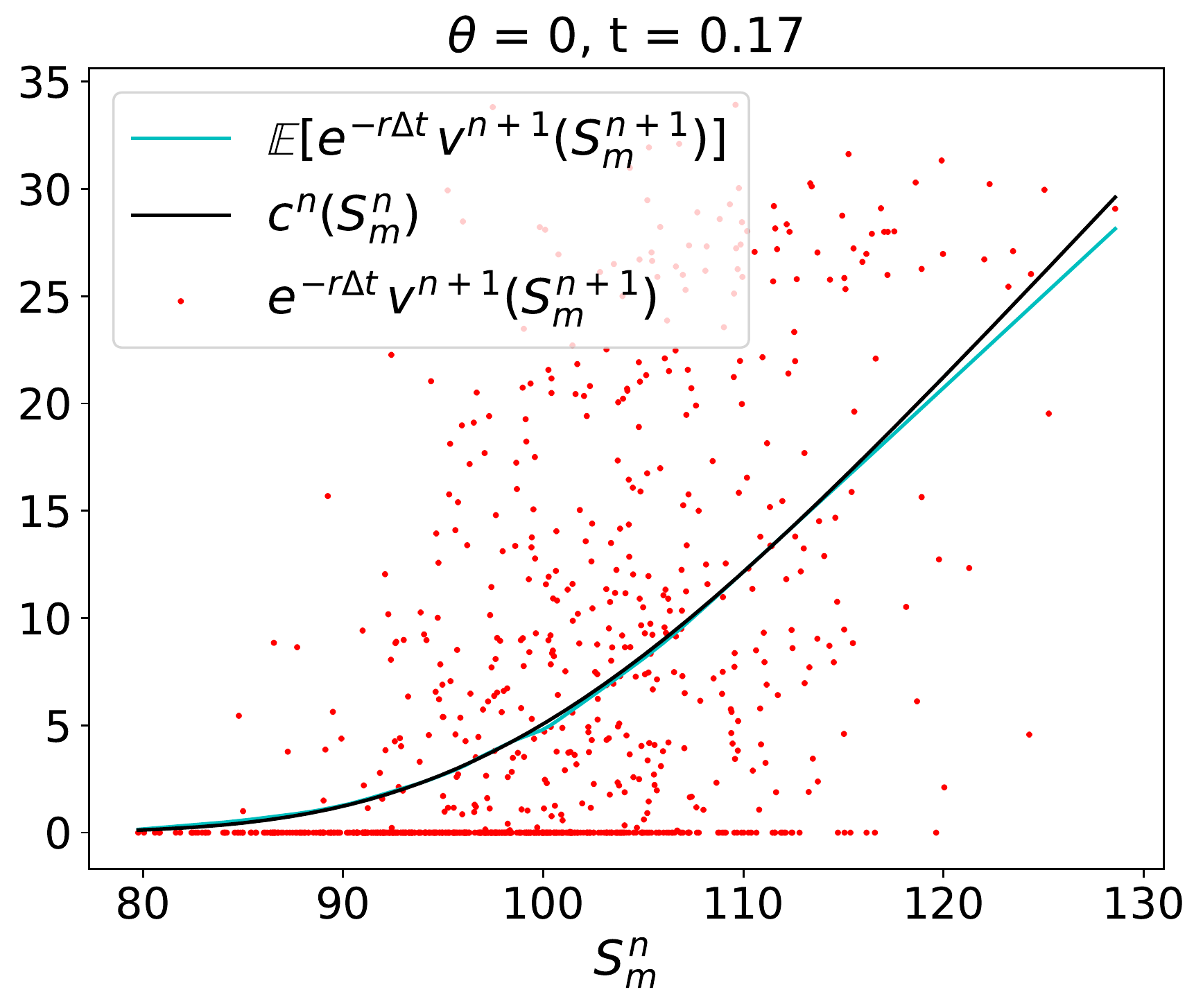}
\includegraphics[height=5cm]{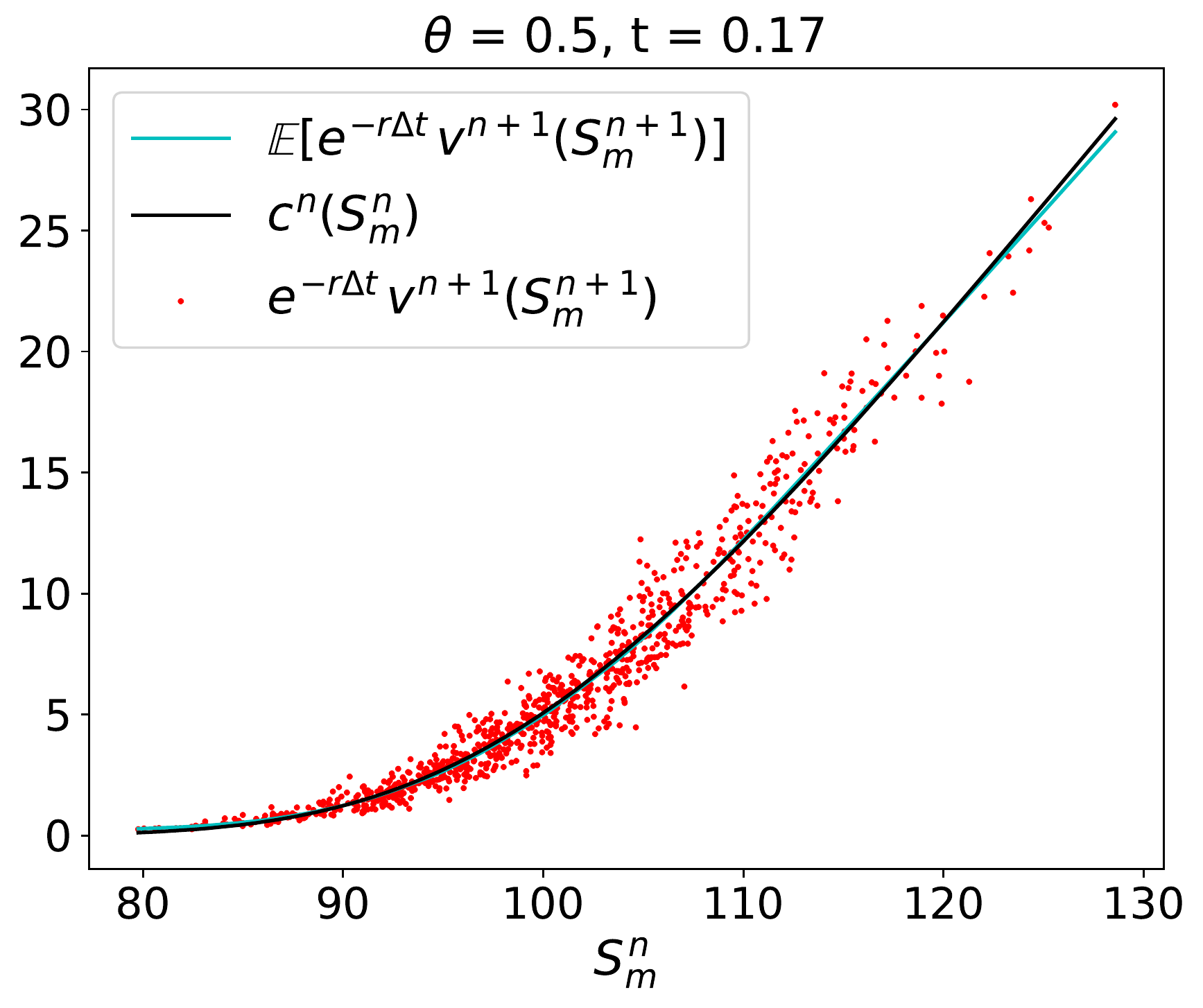}
\caption{\label{fig:theta}
The values of $c^{n}(\vec{S}^{n}_m)$ (black line), $e^{-r\Delta t} v^{n+1}(\vec{S}^{n+1}_m)$ (red dots) and $\mathbb{E}[e^{-r\Delta t} v^{n+1}(\vec{S}^{n+1}_m)]$ (blue line) under different definitions of $v^{n+1}(\vec{S}^{n+1}_m)$.
(Top left) The values under the definition of (\ref{eq:v_def1}), which shows a bias;
(Top right) The values under the definition of (\ref{eq:v_def2}), which shows a variance;
(Bottom) The values under the definition of (\ref{eq:v_def3}) with $\theta=0.5$, where both bias and variance are reduced.
}
\end{figure}


Figure \ref{fig:theta} shows a concrete example of the bias.  Consider a simulation of a one-dimensional American option, where $T=0.5$, $N=100$ and the true continuation function $c^n$ can be computed by finite difference methods. Consider using (\ref{eq:v_def1}) to define the training input $v^{n+1}(\vec{S}^{n+1}_m)$ at every timestep. As shown in the top left plot of Figure \ref{fig:theta}, when the simulation proceeds to $n=33$, there is a clear deviation of $\mathbb{E} [e^{-r\Delta t} v^{n+1}(\vec{S}^{n+1}_m)]$ (blue line)\footnote{To assess $\mathbb{E} [e^{-r\Delta t} v^{n+1}(\vec{S}^{n+1}_m)]$, we start with a fixed set of $\{\vec{S}^n_m\}$. For each point of $\vec{S}^n_m$, we generate multiple $\vec{S}^{n+1}_m$'s by (\ref{eq:Euler_S}), denoted as $\{\vec{S}^{n+1}_{m;m'}|m'=1,\cdots,M'\}$; compute $\{v(\vec{S}^{n+1}_{m;m'})\}$; and then compute the imperial average:
$
\mathbb{E} [e^{-r\Delta t} v^{n+1}(\vec{S}^{n+1}_m)]
\approx
e^{-r\Delta t}\frac{1}{M'}\sum_{m'}v(\vec{S}^{n+1}_{m;m'})
$.
}
from $c^n(\vec{S}^n_m)$ (black line) around $\vec{S}^n_m=80$. 

In fact, we can use the relation (\ref{eq:v_def2_pre}) to avoid the bias caused by the definition (\ref{eq:v_def1}). More specifically, let $\vec{S}^{n+1}_m$ be a continued point. Then
$v^{n+1}(\vec{S}^{n+1}_m)
=c^{n+1}(\vec{S}^{n+1}_m)
=\mathbb{E} [e^{-r\Delta t} v^{n+2}(\vec{S}^{n+2}_m)]$. This motivates us to redefine the training input $v^{n+1}(\vec{S}^{n+1}_m)$ as follows:
\begin{equation}
\label{eq:v_def2}
v^{n+1}(\vec{S}^{n+1}_m)
= \left\{
\begin{array}{ll}
e^{-r\Delta t} \, v^{n+2}(\vec{S}^{n+2}_m),
&
\text{if }
\xi^{n+1}(\vec{S}^{n+1}_m) = \text{continued},
\\
f(\vec{S}^{n+1}_m),
&
\text{if }
\xi^{n+1}(\vec{S}^{n+1}_m) = \text{exercised}.
\end{array}
\right.
\end{equation}
We note that (\ref{eq:v_def2}) is actually the ``discounted payoffs" used in \citet{longstaff2001valuing}. They use (\ref{eq:v_def2}) as the target prices for regression.

The top right plot of Figure \ref{fig:theta} considers again the same simulation, where the definition of $v^{n+1}(\vec{S}^{n+1}_m)$ is changed to (\ref{eq:v_def2}). The deviation of $\mathbb{E} [e^{-r\Delta t} v^{n+1}(\vec{S}^{n+1}_m)]$ (blue line) from $c^n(\vec{S}^n_m)$ (black line) around $\vec{S}^n_m=80$ disappears. The blue and black lines agree well with each other. This shows that using the definition (\ref{eq:v_def2}) does not introduce bias as does the definition (\ref{eq:v_def1}).
However, the noisy red dots show that using the definition (\ref{eq:v_def2}) results in a big variance of $e^{-r\Delta t} v^{n+1}(\vec{S}^{n+1}_m)$. This poses a risk for the model to fit the noise, which may still result in an inaccurate trained $y^{n}$.

In this paper, we define $v^{n+1}(\vec{S}^{n+1}_m)$ as the linear combination of the two definitions (\ref{eq:v_def1}) and (\ref{eq:v_def2}):
\begin{equation}
\label{eq:v_def3}
v^{n+1}(\vec{S}^{n+1}_m)
= \left\{
\begin{array}{l}
\theta \,
y^{n+1} (\vec{S}^{n+1}_m;\,(\Omega^{n+1})^*)
+ (1-\theta)
e^{-r\Delta t} \, v^{n+2}(\vec{S}^{n+2}_m),
\\
\hspace{35.5mm}
\text{if }
\xi^{n+1}(\vec{S}^{n+1}_m) = \text{continued},
\\
f(\vec{S}^{n+1}_m),
\hspace{20mm}
\text{if }
\xi^{n+1}(\vec{S}^{n+1}_m) = \text{exercised},
\end{array}
\right.
\end{equation}
where $\theta \in [0,1]$ is a user-defined hyperparameter. This linear combination mitigates both the bias caused by the definition (\ref{eq:v_def1}) and the variance caused by the definition (\ref{eq:v_def2}). That is, the resulting $v^{n+1}(\vec{S}^{n+1}_m)$ would accumulate less training error over multiple timesteps, and meanwhile contain less noise. The bottom plot in Figure \ref{fig:theta} considers the same simulation, where the definition of $v^{n+1}(\vec{S}^{n+1}_m)$ is (\ref{eq:v_def3}) with $\theta=0.5$. We observe almost no deviation of $\mathbb{E}[e^{-r\Delta t} v^{n+1}(\vec{S}^{n+1}_m)]$ (blue line) from $c^n(\vec{S}^n_m)$ (black line), and a small variance of $e^{-r\Delta t} v^{n+1}(\vec{S}^{n+1}_m)$ (red dots), as expected. Hence, the definition (\ref{eq:v_def3}) can improve the accuracy of the trained networks.

\subsection{Weight reuse}

The trainable parameters $\Omega^n$ need to be initialized for each individual network from $n=N-1$ to $n=0$. Starting from the network at $n=N-1$,
we initialize $(\vec{\beta}^{[l]})^{N-1}$ and $b^{N-1}$ by zeros; $(\vec{\gamma}^{[l]})^{N-1}$ and $\alpha^{N-1}$ by ones; and $(\mathbf{W}^{[l]})^{N-1}$ and $\vec{\omega}^{N-1}$ by uniformly distributed random numbers in
$(
-1/\sqrt{d^{[l]}+d^{[l-1]}},
1/\sqrt{d^{[l]}+d^{[l-1]}}
)$,
as suggested in \citet{goodfellow2016deep}.
Move on to the consecutive networks at $n<N-1$. One can use the same idea to initialize their trainable parameters. However, we notice that when $\Delta t$ is sufficiently small, the networks at the $n$-th and $(n+1)$-th timesteps should be close. In other words, their optimal trainable parameters should be close, i.e., $(\Omega^{n+1})^*\approx(\Omega^{n})^*$. We can take advantage of this fact and use the values of the trained parameters $(\Omega^{n+1})^*$ as the initial values of the corresponding trainable parameters $\Omega^n$. Such ``weight reuse" provides a good initial guess before the training starts at the $n$-th timestep.
Hence, the training results will be more accurate.


\begin{figure}[b!]
\centering
\footnotesize
\includegraphics[height=4cm]{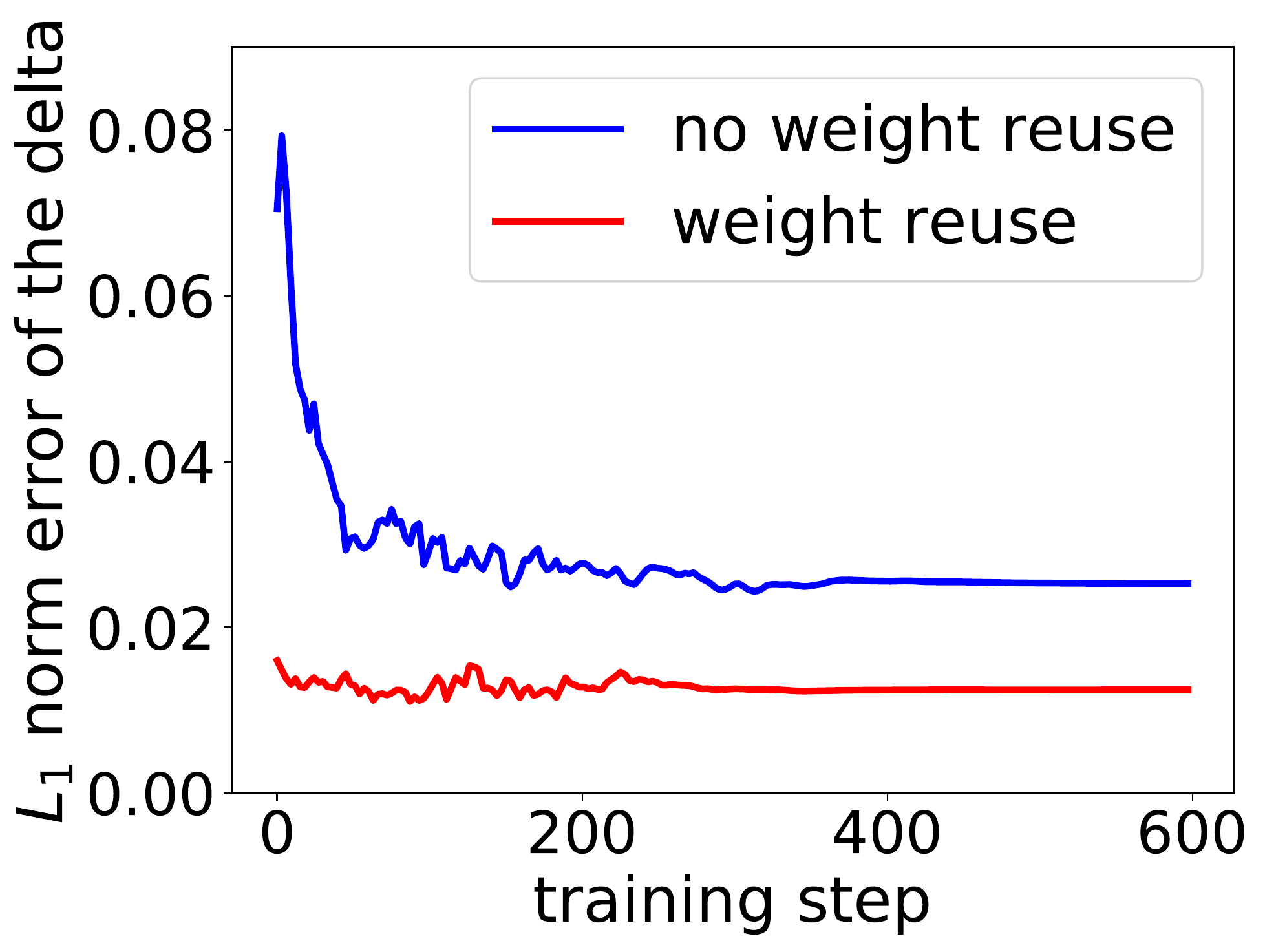}
\includegraphics[height=4cm]{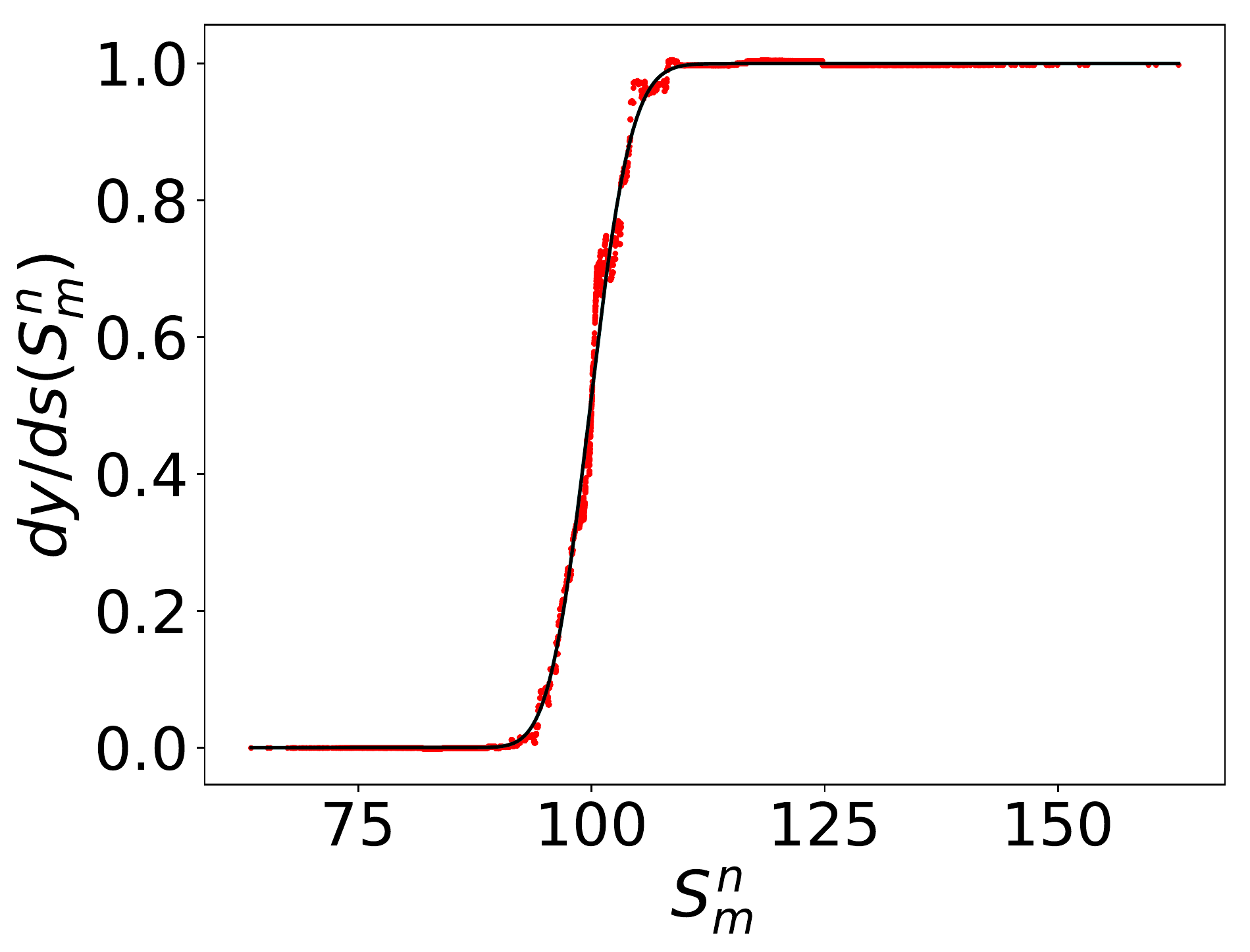}
\includegraphics[height=4cm]{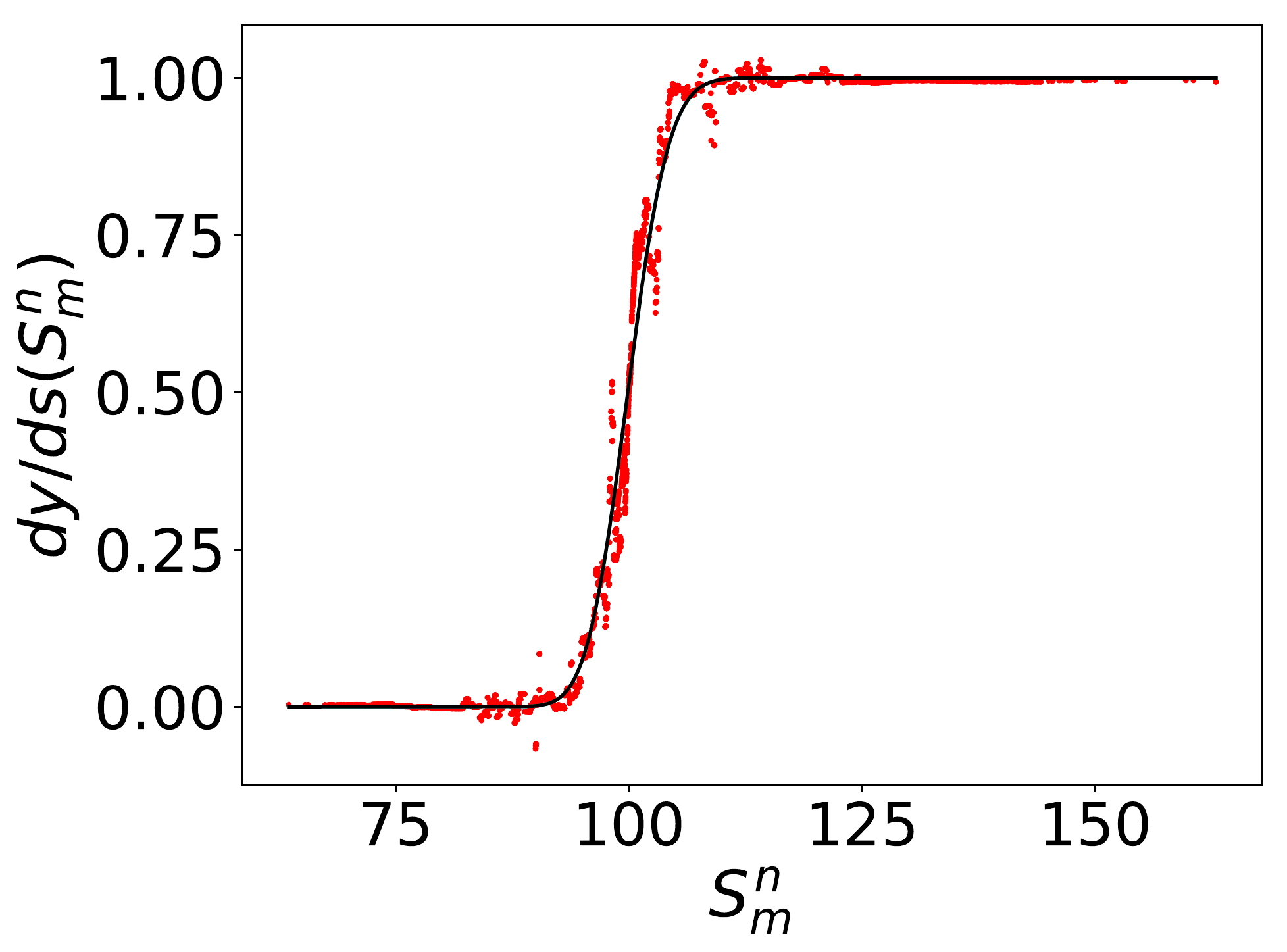}
\caption{\label{fig:reuse}
Example of the computed deltas with or without weight reuse.
(Left) The $L_1$ norm error of the computed delta over 600 training steps. Blue: the error with no weight reuse. Red: the error with weight reuse.
(Middle) The computed delta with weight reuse after 600 training steps. Black line: the exact delta computed by finite difference. Red dots: the sample values of the delta obtained from the network $y^n$.
(Right) The computed delta without weigh reuse after 600 training steps.
}
\end{figure}


Figure \ref{fig:reuse} demonstrates a concrete example on how weight reuse improves the training accuracy. Consider again a simulation of a one-dimensional American option with $T=0.5$, $N=50$. Consider a particular timestep $n=47$. We computed the delta $\frac{dy^n}{ds}(S^n_m)$ of 180000 sample points. The first plot shows the evolution of the $L_1$ norm error of the computed delta over 600 training steps. The error with weight reuse (red line) is significantly lower than the error without weight reuse (blue line). The second plot shows that after 600 training steps, the computed delta with weight reuse (red dots) agrees with the exact delta (black line). As a comparison, the third plot shows that after 600 training steps, the computed delta without weight reuse (red dots) still has a large fluctuation and does not match the exact delta (black line) well.

\subsection{Ensemble of neural networks}
\label{subsec:ensemble}

It is well-known that ensemble learning, which is a combination of the multiple machine learning models, usually outperforms individual models \citep{goodfellow2016deep}. Inspired by this, we consider ``ensemble of neural networks".

To describe the details, at each timestep (e.g., the $n$-th timestep), we construct $C$ networks
$
\{y^n(\vec{s};\,\Omega^n_c) \,|\, c=1,\cdots, C\}
$
instead of one network.
All the $C$ networks have the same architecture as defined in Sections \ref{subsec:NN}-\ref{subsec:NN2}.
The difference is that their trainable parameters
$
\{\Omega^n_c \,|\, c=1,\cdots, C\}
$ are initialized by different set of numbers.
Then the $C$ networks are trained by different input data. To do this, we generate $CM$ input samples (\ref{eq:train_input}) with $m=1,\cdots,CM$, split them into $C$ copies, and then use each copy of the input samples to train each of the $C$ networks. Consequentially, the trained results of the $C$ networks are independent of each other, i.e.,
$
\{(\Omega^n_c)^* \,|\, c=1,\cdots, C\}
$
are distinct from each other.
Then after training, we compute the averages across the ensemble:
\begin{equation}
y^n(\vec{s})
=\frac{1}{C} \sum_{c=1}^C
y^n(\vec{s};\,(\Omega_c^n)^*),
\quad
\vec{\nabla} y^n(\vec{s})
=\frac{1}{C} \sum_{c=1}^C
\vec{\nabla} y^n(\vec{s};\,(\Omega_c^n)^*),
\end{equation}
for the prices and deltas, respectively.
Eventually, we use the ensemble-average prices to determine the exercise boundary at the $n$-th timestep by (\ref{eq:exercise}) before proceeding to the $(n-1)$-th timestep.

Such ensemble technique yields more precise prices, deltas and thus more precise exercise boundaries.
We note that the computation across different ensembles can be parallelized. In practice, we find that $C=3$ is a good choice, in the sense that the accuracy is improved compared with $C=1$ without dramatically increasing computational cost.

\subsection{Price and delta at $t=0$}
\label{subsec:initial}

Our neural network formulation yields prices and deltas on the entire spacetime domain. In practical applications, the price and the delta at $t=0$, $v(\vec{s}^0, 0)$ and $\vec{\nabla} v(\vec{s}^0, 0)$, are of particular interest. We can extract their values from the trained neural network at $t=0$. Here we discuss how to further improve the accuracy of their values.

Our approach is to use the expectation values of the Monte Carlo paths, subject to the exercise boundary computed by our neural network formulation. More specifically, given the $m$-th path, the trained neural networks determine its stopping time, denoted as $\tau_m$. Then the price at $t=0$ can be computed by the mean of the discounted payoffs:
\begin{equation}
\label{eq:initial_price}
v(\vec{s}^0, 0)
= \frac{1}{CM}\sum_{m=1}^{CM}
e^{-r\tau_m} f(\vec{S}_m(\tau_m)).
\end{equation}
Regarding the delta at $t=0$, we can use the method in \citet{thom2009longstaff}, which is an adaptation of ``pathwise derivative method" \citep{broadie1996estimating} to American options:
\begin{equation}
\label{eq:initial_delta}
\frac{\partial v}{\partial s_i}(\vec{s}^0, 0)
= \frac{1}{CM}\sum_{m=1}^{CM}
\left(
e^{-r\tau_m}
\sum_{j=1}^d
\frac{\partial f}{\partial s_j}(\vec{S}_m(\tau_m))
\frac{\partial (S_j)_m}{\partial (s_i)^0}
\right).
\end{equation}
When the underlying asset prices evolve under (\ref{eq:SDE_S}), we have
$\frac{\partial (S_j)_m}{\partial (s_i)^0}
= \frac{(S_j)_m}{(s_i)^0} \delta_{ij}$.
We note that the pathwise derivative approach may not be applicable if $\frac{\partial (S_j)_m}{\partial (s_i)^0}$ is not evaluable (e.g., the underlying asset prices do not evolve under (\ref{eq:SDE_S})) or if the payoff function is not differentiable. For such non-applicable cases, we can still obtain the deltas from our trained neural network at $t=0$.

Using (\ref{eq:initial_price})-(\ref{eq:initial_delta}) to compute the price and the delta at $t=0$ is also observed in other Monte Carlo style pricing approaches, including the Longstaff-Schwartz algorithm. However, we emphasize that our approach differs from the others. More specifically, (\ref{eq:initial_price})-(\ref{eq:initial_delta}) are not computable unless combined with an algorithm that can determine the exercise boundary on the entire spacetime.
In this paper, our neural network framework is used to determine the exercise boundary before applying (\ref{eq:initial_price})-(\ref{eq:initial_delta}). In Section \ref{sec:numerical}, we will demonstrate that our neural network formulation yields a more accurate exercise boundary, and thus more accurate prices and deltas at $t=0$, compared to the Longstaff-Schwartz algorithm.

\subsection{Final algorithm}
\label{subsec:algorithm}

The final version of the proposed algorithm is summarized in Algorithm \ref{alg:SDE-NN}. We note that in Algorithm \ref{alg:SDE-NN}, we store
$\{y^n(\vec{S}^n_m), \vec{\nabla} y^n(\vec{S}^n_m)
\,|\, \forall n, \forall m\}$
on the entire spacetime (i.e., for all $m$'s and $n$'s). The reason is that we are interested in a complete delta hedging simulation, which requires sample values of both prices and deltas on the entire spacetime. The implementation of Algorithm \ref{alg:SDE-NN} uses an overwriting strategy for more efficient memory. We note, however, that if an algorithm user does not need sample values from the entire spacetime, then only the storage of the training outputs
$\{y^n(\vec{S}^n_m),\vec{\nabla} y^n(\vec{S}^n_m)
\,|\, \forall m\}$
and the training inputs
$\{y^{n+\eta}(\vec{S}^n_m),\vec{\nabla} y^{n+\eta}(\vec{S}^n_m)
\,|\, \forall m\}$
at the current timestep (i.e., for all $m$'s and for a given $n$) is necessary.

\begin{algorithm}[h!]
\caption{Neural network pricing and hedging under BSDE formulation}
\label{alg:SDE-NN}
\begin{algorithmic}[1]
\STATE{ {\bf Parameters} }
\STATE{ \qquad $C$: the number of networks in network ensemble }
\STATE{ \qquad $M$: the number of samples per ensemble }
\STATE{ \qquad $N$: the number of timesteps }
\STATE{ \qquad $J$: the number of timesteps between the network recurrence }
\STATE{ }
\STATE{ Initialize the underlying asset prices
        $\{\vec{S}^0_m \equiv \vec{s}^{0}
        \, | \, \forall m
        \text{ (i.e., } m=1,\cdots, CM)\}$. }
\FOR{ $n = 1, \cdots, N$ }
	\STATE{ Use (\ref{eq:Euler_S0})-(\ref{eq:Euler_S}) to generate 
	        $CM$ Monte Carlo trajectories of 
	        the underlying asset prices 
	        $\{\vec{S}^n_m \, | \, \forall m\}$. }
\ENDFOR
\STATE{ }
\STATE{ Use (\ref{eq:NN_def_payoff2}) to compute
        the expiry option prices and option deltas
        \\
        \qquad
        $\{Y^\nu_m = y^N(\vec{S}^\nu_m)
        \, | \, 0\leq\nu\leq N, \forall m\}$
        and
        $\{\vec{Z}^\nu_m = \vec{\nabla} y^N(\vec{S}^\nu_m)
        \, | \, 0\leq\nu\leq N, \forall m\}$. }
\STATE{ Initialize
        $\{v^N(\vec{S}^N_m)
        \, | \, \forall m\}$
        by (\ref{eq:Euler_vT}). }
\STATE{ }
\FOR{ $n = N-1, \cdots, 0$ }
    \FOR{ $c = 1, \cdots, C$ }
	    \STATE{ Initialize the neural network
	            $y^n (\vec{s};\,\Omega^n_c)$
	            defined by (\ref{eq:NN_def_recursive_3}),
	            where the input layer is
	            (\ref{eq:NN_layer_input_concat})-(\ref{eq:NN_layer_input_batch}),
	            the hidden layers are
	            (\ref{eq:NN_layer_linear})-(\ref{eq:NN_layer_activate})
	            and the output layer is
	            (\ref{eq:NN_layer_output}). }
	    \STATE{ {\bf Training:}
	            minimize the least squares residual
	            (\ref{eq:LeastSq3})-(\ref{eq:LeastSq4}),
	            using the training input (\ref{eq:train_input}). }
	    \STATE{ Result: the trained neural network
	            $y^n (\vec{s};\,(\Omega^n_c)^*)$. }
	\ENDFOR
    \STATE{ }
	\IF{ $(N-n) \text{ mod } J=0$}
        \STATE{ {\bf Ensemble evaluation}
            (all future timesteps):
            overwrite the option prices and deltas
            \\
            \qquad
            $\{Y^\nu_m
            = \frac{1}{C}\sum_{c=1}^C
            y^n(\vec{S}^\nu_m; \,(\Omega^n_c)^*)
            \, | \, 0\leq\nu\leq n, \forall m\}$,
            \\
            \qquad
            $\{\vec{Z}^\nu_m
            = \frac{1}{C}\sum_{c=1}^C
            \vec{\nabla} y^n(\vec{S}^\nu_m; \,(\Omega^n_c)^*)
            \, | \, 0\leq\nu\leq n, \forall m\}$. }
    \ELSE
        \STATE{ {\bf Ensemble evaluation}
            (current timestep):
            overwrite the option prices and deltas
            \\
            \qquad
            $\{Y^n_m
            = \frac{1}{C}\sum_{c=1}^C
            y^n(\vec{S}^n_m; \,(\Omega^n_c)^*)
            \, | \, \forall m\}$,
            \\
            \qquad
            $\{\vec{Z}^n_m
            = \frac{1}{C}\sum_{c=1}^C
            \vec{\nabla} y^n(\vec{S}^n_m; \,(\Omega^n_c)^*)
            \, | \, \forall m\}$. }
	\ENDIF
	\STATE{ }
	\STATE{ Determine whether $\vec{S}^n_m$ 
	        is continued or exercised
	        using (\ref{eq:exercise}) for all $m$'s.}
	\STATE{ Update 
	        $\{v^n(\vec{S}^n_m) \, | \, \forall m\}$
	        by (\ref{eq:v_def3}). }
\ENDFOR
\STATE{ }
\STATE{ Result: samples of option price and delta functions 
		on the entire spacetime
        \\
        \qquad
        $\{ Y^n_m \gets \max(Y^n_m, f(\vec{S}^n_m))
        \, | \, \forall n, \forall m \}$
        and
        $\{ \vec{Z}^n_m
        \, | \, \forall n, \forall m \}$.
        }
\STATE{ Optional result: Recompute the option price and
        the option delta at $t=0$ using 
        (\ref{eq:initial_price}) and (\ref{eq:initial_delta}).
        }
\end{algorithmic}
\end{algorithm}


\section{Computational Cost}
\label{sec:cost}

In this section, we analyze the computational cost of the proposed algorithm, and compare it with the Longstaff-Schwartz algorithm. For the Longstaff-Schwartz algorithm, consider degree-$\chi$ monomial basis
\begin{equation}
\label{eq:LSM_basis}
\varphi_\chi(\vec{s})\equiv
\{s_1^{a_1}s_2^{a_2}\cdots s_d^{a_d}
\, | \, a_1+a_2+\cdots+a_d\leq \chi\},
\end{equation}
as proposed in \citet{longstaff2001valuing} and \citet{kohler2010review}. In practice, we choose $\chi\ll d$. Then the number of the monomial basis is
$\left(
\begin{smallmatrix}
d+\chi \\ d
\end{smallmatrix}
\right)
\approx \frac{1}{\chi!}d^{\chi}$.

\subsection{Memory}
\label{subsec:memory}

The proposed algorithm requires storing
\begin{itemize}[leftmargin=10mm]
\item the underlying asset prices
$\{\vec{S}^n_m \, | \, \forall n, \forall m \}$
on the entire spacetime, requiring $NMd$ floating point numbers;
\item the training outputs
$\{y^n(\vec{S}^n_m),
\vec{\nabla} y^n(\vec{S}^n_m)
\, | \, \forall m \}$
and the training inputs
$\{y^{n+\eta}(\vec{S}^n_m),
\\
\vec{\nabla} y^{n+\eta}(\vec{S}^n_m)
\, | \, \forall m \}$
at the current timestep, requiring $2(M+Md)$ floating point numbers.
\end{itemize}
Hence, the entire process requires a total memory of
$NMd+2(M+Md) \approx NMd$
floating point numbers.
As a comparison, the Longstaff-Schwartz method requires storing
$\{\vec{S}^n_m \, | \, \forall n, \forall m \}$
on the entire spacetime and storing
$\{\varphi_\chi(\vec{S}^n_m), y^n(\vec{S}^n_m)
\, | \, \forall m \}$
at the current timestep. This requires a total memory of
$NMd +
M\cdot \frac{1}{\chi!}d^{\chi}
+ M
\approx
NMd + \frac{1}{\chi!}Md^{\chi}
$
floating point numbers. We remind readers that convergence of the Longstaff-Schwartz method to the exact American option prices requires $\chi\to\infty$. As a result, the proposed neural network method is more memory efficient than the Longstaff-Schwartz method.

\subsection{Time}
\label{subsec:time}

Consider a given timestep $n$. The computational time is dominated by two stages:
\begin{itemize}[leftmargin=10mm]
\item Stage 1: Computing the training inputs (\ref{eq:train_input}), in particular,
$
\{
y^{n+\eta}(\vec{S}^n_m;(\Omega^{n+\eta})^*), \,
\\
\vec{\nabla} y^{n+\eta}(\vec{S}^n_m;(\Omega^{n+\eta})^*)
\,|\, \forall m
\}
$,
using the trained networks
$\{(\Omega^\nu)^*
\,|\,
\nu \geq n + \eta\}$.
\item Stage 2: Training, using the training inputs (\ref{eq:train_input}).
\end{itemize}

To derive the computational time of each stage, denote the maximal width of the $L$-layer neural network $\mathcal{F}$ as $d_{max}\equiv \max_{l=0,\cdots,L}d^{[l]}$. We note that matrix multiplication is the dominant operation in (\ref{eq:NN_layer_linear})-(\ref{eq:NN_layer_activate}). Hence, for each stage, the computational time {\it per neural network} is given by $c_1 MLd_{max}^2$ and $c_2 MLd_{max}^2$, where $c_1$ and $c_2$ are constants. Typically $c_1 \ll c_2$, because Stage 1 only involves computing the outputs of neural networks, while Stage 2 involves training. This seems to suggest that Stage 2 dominates Stage 1. However, we note that Stage 2 involves only one single network (i.e., the $n$-th network), while Stage 1 involves multiple networks from the previous timesteps. More specifically, following the same analysis as (\ref{eq:y0_2}), one can show that the computation of the training input $y^{n+\eta}(\vec{S}^n_m;(\Omega^{n+\eta})^*)$, given by
\begin{equation}
\label{eq:y0_3}
y^{n+\eta} (\vec{s})
= y^N(\vec{s})
+ J \Delta t \cdot \sum_{\nu=1}^{(N-n-\eta)/J}
\mathcal{F}(\vec{s};\Omega^{N-\nu J}),
\end{equation}
requires going through
$(N-n-\eta)/J \approx (N-n)/J$
feedforward networks. As a result,
the actual computational time for Stage 1 is $c_1 MLd_{max}^2 \cdot \frac{N-n}{J}$.

Furthermore, if we consider all the $N$ timesteps, then the total computational time is
\begin{equation}
\label{eq:cost}
\begin{array}{l}
\text{Stage 1: }
\displaystyle
\sum_{n=0}^N c_1 MLd_{max}^2 \cdot \frac{N-n}{J}
= \frac{c_1 N^2}{2J} MLd_{max}^2,
\\
\text{Stage 2: }
\displaystyle
\sum_{n=0}^N c_2 MLd_{max}^2
= c_2 NMLd_{max}^2.
\end{array}
\end{equation}
Equation (\ref{eq:cost}) suggests that when $N$ is large, Stage 1 is dominant. However, we can significantly reduce the computational time of Stage 1 by increasing $J$, as discussed in Section \ref{subsec:NN2}. In our numerical simulation, we chose $d_{max}=d+5$. Then the total computational time of the proposed algorithm is approximately $(\frac{c_1 N}{2J} + c_2) NMLd^2$, which is quadratic in the dimension $d$.

Regarding the Longstaff-Schwartz method, if we assume that the standard normal equation or QR factorization is used for solving regression problems, then the computational time is
$O\left(NM(\frac{1}{\chi!}d^{\chi})^2\right)
= O(NMd^{2\chi})$,
which is worse-than-quadratic in $d$.
Hence, the proposed neural network method is asymptotically more efficient than the Longstaff-Schwartz method in high dimensions.


\section{Numerical Results}
\label{sec:numerical}

In this section, we solve the American option problem (\ref{eq:SDE_S})-(\ref{eq:AmerOp_def2}) using our neural network described in Algorithm \ref{alg:SDE-NN}. We compute the price $v(\vec{s}^0,0)$ and the delta $\vec{\nabla} v(\vec{s}^0,0)$ at $t=0$ for given $\vec{s}^0=(s_1^0,\cdots,s_d^0)$ where $s_1^0=\cdots=s_d^0=0.9K, K$ or $1.1K$. We also compute the prices $v(\vec{s},t)$ and the deltas $\vec{\nabla} v(\vec{s},t)$ for sample paths of $(\vec{s},t)$ spread over the entire spacetime.

In our experiments, we set the strike price $K=100$, the number of the timesteps $N=100$, the number of timesteps between the network recurrence $J=4$, the smoothing parameter in (\ref{eq:NN_def_payoff2}) $\kappa=\frac{2}{\Delta t}$, the coefficient in (\ref{eq:v_def3}) $\theta=0.5$. At each timestep, we train an ensemble of $C=3$ neural networks, where each neural network has a depth of $L=7$ and a uniform width of $d^{[l]}=d+5$ across all the hidden layers. We let the number of samples per network be $M=240000$ (or the total number of samples be $CM=720000$), and let the batch size and the number of training steps be 400 and 600 respectively.
Each numerical experiment is implemented on one Cedar\footnote{Cedar is a Compute Canada cluster. For more details, see https://docs.computecanada.ca/wiki/Cedar and https://docs.computecanada.ca/wiki/Using\_GPUs\_with\_Slurm.} base-GPU node, which contains 4 NVIDIA P100-PCIE-12GB GPUs, 24 CPUs and 128GB memory. 

We compare the numerical results computed by our proposed method with those computed by the finite difference method, the Longstaff-Schwartz method and the method proposed in \citet{sirignano2018dgm}. For the Longstaff-Schwartz method, we choose degree-$\chi$ monomial basis (\ref{eq:LSM_basis}) with $\chi=4$.
Finite difference solutions with very fine grids are used as exact solutions. We note that this is feasible only if $d\leq 3$.

We note that when finite difference solutions are available, we can evaluate the absolute and percent errors of computed prices and deltas. More specifically, denote the finite difference solutions as $v_{exact}$. Then the percent errors of the price and the delta at $t=0$ are
\begin{equation}
\label{eq:metric1}
\frac{|v(\vec{s}^0,0)-v_{exact}(\vec{s}^0,0)|}
{|v_{exact}(\vec{s}^0,0)|}
\times 100\%,
\quad
\frac{\|\vec{\nabla} v(\vec{s}^0,0)-\vec{\nabla} v_{exact}(\vec{s}^0,0)\|_{L_2}}
{\|\vec{\nabla} v_{exact}(\vec{s}^0,0)\|_{L_2}}
\times 100\%;
\end{equation}
and the percent errors of the spacetime price and the spacetime delta are
\begin{equation}
\label{eq:metric2}
\frac{\sum_{m,n} |v(\vec{S}^n_m,t^n)-v_{exact}(\vec{S}^n_m,t^n)|}
{\sum_{m,n} |v_{exact}(\vec{S}^n_m,t^n)|}
\times 100\%,
\;
\frac{\sum_{m,n} \|\vec{\nabla} v(\vec{S}^n_m,t^n)-\vec{\nabla} v_{exact}(\vec{S}^n_m,t^n)\|_{L_2}}
{\sum_{m,n} \|\vec{\nabla} v_{exact}(\vec{S}^n_m,t^n)\|_{L_2}}
\times 100\%.
\end{equation}

In addition, we can evaluate the quality of the computed exercise boundaries. More specifically, each sample point $(\vec{S}^n_m,t^n)$ is classified as ``exercised" or ``continued" by either the proposed algorithm or other algorithms that we compare with. Meanwhile, the true ``exercised" or ``continued" class of each sample point can be determined by the finite difference method. Let ``exercised" class be the positive class, and denote the numbers of true positive, true negative, false positive and false negative samples as TP, TN, FP, FN, respectively. Then the quality of the exercise boundaries can be evaluated by the f1-score:
\begin{equation}
\label{eq:metric3}
\text{f1-score} \equiv \frac{2TP}{2TP + FP + FN}.
\end{equation}
The best (or worst) case of the f1-score is 1 (or 0), respectively. We note that another common metric to evaluate the quality of classification problems is the accuracy. Since in all our experiments, the positive class is skewed (around 3-17\%), the f1-score would be a better metric than the accuracy \citep[see][for explanations]{murphy2012machine}.

\subsection{Multi-dimensional geometric average options}
\label{subsec:multi_geometric}

Consider a $d$-dimensional ``geometric average" American call option, where $\rho_{ij}=\rho$ for $i\neq j$, $\sigma_i=\sigma$ for all $i$'s, and the payoff function is given by 
$f(\vec{s}) = \max\left[
\left( \prod_{i=1}^d s_i \right)^{1/d} - K, 0
\right]$.
Although such options are rarely seen in practical applications, they have semi-analytical solutions for benchmarking the performance of our algorithm in high dimensions. More specifically, it is shown in \citet{glasserman2013monte} and \citet{sirignano2018dgm} that such a $d$-dimensional option can be reduced to a one-dimensional American call option in the variable
$s' = 
\left( \prod_{i=1}^d s_i \right)^{1/d}
$,
where the effective volatility is
$\sigma' = 
\sqrt{\frac{1+(d-1)\rho}{d}} \sigma
$
and the effective drift is
$r-\delta+\frac{1}{2}(\sigma'^2-\sigma^2)$.
Hence, by solving the equivalent one-dimensional option using finite difference methods, one can compute the $d$-dimensional option prices and (sometimes) deltas\footnote{We note that solving the equivalent one-dimensional option is not sufficient for computing the $d$-dimensional delta except at the symmetric points $s_1=\cdots=s_d$. Interested readers can verify this by straightforward algebra.} accurately.

In the following Experiments 1-5, we consider the geometric average option in Section 4.3 of \citet{sirignano2018dgm}, where $\rho_{i,j}=0.75$, $\sigma=0.25$, $r=0$, $\delta=0.02$, $T=2$.


\begin{table}[h!]
\footnotesize
\begin{center}
(i) 7-dimensional geometric average call option
\\
\begin{tabular}{|c|c|c|c|c|c|}
\hline
\multirow{3}{*}[0em]{$s^0_i$}
& \multirow{3}{*}[0em]
{\begin{tabular}{@{}c@{}}exact price\\$v(\vec{s}^0,0)$\end{tabular}}
& \multicolumn{2}{@{}c@{}|}
{proposed method}
& \multicolumn{2}{@{}c@{}|}
{Longstaff-Schwartz}
\\
\cline{3-6}
& 
&
{\begin{tabular}{@{}c@{}}computed price\\$v(\vec{s}^0,0)$\end{tabular}}
& \begin{tabular}{@{}c@{}}percent\\error\end{tabular}
&
{\begin{tabular}{@{}c@{}}computed price\\$v(\vec{s}^0,0)$\end{tabular}}
& \begin{tabular}{@{}c@{}}percent\\error\end{tabular}
\\
\hline
90
&
5.9021
&
5.8822
&
0.34\%
&
5.8440
&
0.98\%
\\
\hline
100
&
10.2591
&
10.2286
&
0.30\%
&
10.1736
&
0.83\%
\\
\hline
110
&
15.9878
&
15.9738
&
0.09\%
&
15.8991
&
0.55\%
\\
\hline
\end{tabular}
\\
\smallskip
(ii) 13-dimensional geometric average call option
\\
\begin{tabular}{|c|c|c|c|c|c|}
\hline
\multirow{3}{*}[0em]{$s^0_i$}
& \multirow{3}{*}[0em]
{\begin{tabular}{@{}c@{}}exact price\\$v(\vec{s}^0,0)$\end{tabular}}
& \multicolumn{2}{@{}c@{}|}
{proposed method}
& \multicolumn{2}{@{}c@{}|}
{Longstaff-Schwartz}
\\
\cline{3-6}
& 
&
{\begin{tabular}{@{}c@{}}computed price\\$v(\vec{s}^0,0)$\end{tabular}}
& \begin{tabular}{@{}c@{}}percent\\error\end{tabular}
&
{\begin{tabular}{@{}c@{}}computed price\\$v(\vec{s}^0,0)$\end{tabular}}
& \begin{tabular}{@{}c@{}}percent\\error\end{tabular}
\\
\hline
90
&
5.7684
&
5.7719
&
0.06\%
&
5.5962
&
3.0\%
\\
\hline
100
&
10.0984
&
10.1148
&
0.16\%
&
9.9336
&
1.6\%
\\
\hline
110
&
15.8200
&
15.8259
&
0.04\%
&
15.6070
&
1.4\%
\\
\hline
\end{tabular}
\\
\smallskip
(iii) 20-dimensional geometric average call option
\\
\begin{tabular}{|c|c|c|c|c|c|}
\hline
\multirow{3}{*}[0em]{$s^0_i$}
& \multirow{3}{*}[0em]
{\begin{tabular}{@{}c@{}}exact price\\$v(\vec{s}^0,0)$\end{tabular}}
& \multicolumn{2}{@{}c@{}|}
{proposed method}
& \multicolumn{2}{@{}c@{}|}
{Longstaff-Schwartz}
\\
\cline{3-6}
& 
&
{\begin{tabular}{@{}c@{}}computed price\\$v(\vec{s}^0,0)$\end{tabular}}
& \begin{tabular}{@{}c@{}}percent\\error\end{tabular}
&
{\begin{tabular}{@{}c@{}}computed price\\$v(\vec{s}^0,0)$\end{tabular}}
& \begin{tabular}{@{}c@{}}percent\\error\end{tabular}
\\
\hline
90
&
5.7137
&
5.7105
&
0.06\%
&
5.2023
&
9.0\%
\\
\hline
100
&
10.0326
&
10.0180
&
0.15\%
&
9.5964
&
4.4\%
\\
\hline
110
&
15.7513
&
15.7425
&
0.06\%
&
15.2622
&
3.1\%
\\
\hline
\end{tabular}
\\
\smallskip
(iv) 100-dimensional geometric average call option
\\
\begin{tabular}{|c|c|c|c|c|c|}
\hline
\multirow{3}{*}[0em]{$s^0_i$}
& \multirow{3}{*}[0em]
{\begin{tabular}{@{}c@{}}exact price\\$v(\vec{s}^0,0)$\end{tabular}}
& \multicolumn{2}{@{}c@{}|}
{proposed method}
& \multicolumn{2}{@{}c@{}|}
{Longstaff-Schwartz}
\\
\cline{3-6}
& 
&
{\begin{tabular}{@{}c@{}}computed price\\$v(\vec{s}^0,0)$\end{tabular}}
& \begin{tabular}{@{}c@{}}percent\\error\end{tabular}
&
{\begin{tabular}{@{}c@{}}computed price\\$v(\vec{s}^0,0)$\end{tabular}}
& \begin{tabular}{@{}c@{}}percent\\error\end{tabular}
\\
\hline
90
&
5.6322
&
5.6154
&
0.30\%
&
OOM
&
OOM
\\
\hline
100
&
9.9345
&
9.9187
&
0.16\%
&
OOM
&
OOM
\\
\hline
110
&
15.6491
&
15.6219
&
0.17\%
&
OOM
&
OOM
\\
\hline
\end{tabular}
\end{center}
\caption{\label{tab:multi_geometric_vs_Longstaff_1}
Multi-dimensional geometric average call options: Computed prices at $t=0$, i.e., $v(\vec{s}^0,0)$.
OOM means ``out-of-memory".}
\end{table}



\begin{table}[h!]
\footnotesize
\begin{center}
(i) 7-dimensional geometric average call option
\\
\begin{tabular}{|c|c|c|c|c|}
\hline
\multirow{2}{*}[0em]{$s^0_i$}
& \multirow{2}{*}[0em]
{\begin{tabular}{@{}c@{}}exact delta
\\$\vec{\nabla}v(\vec{s}^0,0)$\end{tabular}}
& \multicolumn{2}{@{}c@{}|}
{proposed method}
& Longstaff-Schwartz
\\
\cline{3-5}
& 
& computed delta $\vec{\nabla}v(\vec{s}^0,0)$
& percent error
& percent error
\\
\hline
90
&
(0.0523,$\cdots$,0.0523)
&
(0.0516,$\cdots$,0.0516)
&
1.2\%
&
1.2\%
\\
\hline
100
&
(0.0722,$\cdots$,0.0722)
&
(0.0710,$\cdots$,0.0710)
&
1.7\%
&
1.6\%
\\
\hline
110
&
(0.0912,$\cdots$,0.0912)
&
(0.0901,$\cdots$,0.0901)
&
1.2\%
&
1.4\%
\\
\hline
\end{tabular}
\\
\smallskip
(ii) 13-dimensional geometric average call option
\\
\begin{tabular}{|c|c|c|c|c|}
\hline
\multirow{2}{*}[0em]{$s^0_i$}
& \multirow{2}{*}[0em]
{\begin{tabular}{@{}c@{}}exact delta
\\$\vec{\nabla}v(\vec{s}^0,0)$\end{tabular}}
& \multicolumn{2}{@{}c@{}|}
{proposed method}
& Longstaff-Schwartz
\\
\cline{3-5}
& 
& computed delta $\vec{\nabla}v(\vec{s}^0,0)$
& percent error
& percent error
\\
\hline
90
&
(0.0279,$\cdots$,0.0279)
&
(0.0277,$\cdots$,0.0277)
&
0.76\%
&
5.4\%
\\
\hline
100
&
(0.0387,$\cdots$,0.0387)
&
(0.0384,$\cdots$,0.0384)
&
0.83\%
&
3.7\%
\\
\hline
110
&
(0.0492,$\cdots$,0.0492)
&
(0.0486,$\cdots$,0.0486)
&
1.1\%
&
2.6\%
\\
\hline
\end{tabular}
\\
\smallskip
(iii) 20-dimensional geometric average call option
\\
\begin{tabular}{|c|c|c|c|c|}
\hline
\multirow{2}{*}[0em]{$s^0_i$}
& \multirow{2}{*}[0em]
{\begin{tabular}{@{}c@{}}exact delta
\\$\vec{\nabla}v(\vec{s}^0,0)$\end{tabular}}
& \multicolumn{2}{@{}c@{}|}
{proposed method}
& Longstaff-Schwartz
\\
\cline{3-5}
& 
& computed delta $\vec{\nabla}v(\vec{s}^0,0)$
& percent error
& percent error
\\
\hline
90
&
(0.0180,$\cdots$,0.0180)
&
(0.0179,$\cdots$,0.0179)
&
0.70\%
&
12.7\%
\\
\hline
100
&
(0.0251,$\cdots$,0.0251)
&
(0.0248,$\cdots$,0.0248)
&
1.2\%
&
8.3\%
\\
\hline
110
&
(0.0320,$\cdots$,0.0320)
&
(0.0316,$\cdots$,0.0316)
&
1.2\%
&
6.8\%
\\
\hline
\end{tabular}
\\
\smallskip
(iv) 100-dimensional geometric average call option
\\
\begin{tabular}{|c|c|c|c|c|}
\hline
\multirow{2}{*}[0em]{$s^0_i$}
& \multirow{2}{*}[0em]
{\begin{tabular}{@{}c@{}}exact delta
\\$\vec{\nabla}v(\vec{s}^0,0)$\end{tabular}}
& \multicolumn{2}{@{}c@{}|}
{proposed method}
& Longstaff-Schwartz
\\
\cline{3-5}
& 
& computed delta $\vec{\nabla}v(\vec{s}^0,0)$
& percent error
& percent error
\\
\hline
90
&
(0.00359,$\cdots$,0.00359)
&
(0.00357,$\cdots$,0.00357)
&
0.58\%
&
OOM
\\
\hline
100
&
(0.00502,$\cdots$,0.00502)
&
(0.00495,$\cdots$,0.00495)
&
1.3\%
&
OOM
\\
\hline
110
&
(0.00639,$\cdots$,0.00639)
&
(0.00631,$\cdots$,0.00631)
&
1.3\%
&
OOM
\\
\hline
\end{tabular}
\end{center}
\caption{\label{tab:multi_geometric_vs_Longstaff_2}
Multi-dimensional geometric average call options: Computed deltas at $t=0$, i.e., $\vec{\nabla}v(\vec{s}^0,0)$.
Note that all the reported deltas in the table are length-$d$ vectors where all the elements are the same. The column ``Longstaff-Schwartz" is the Longstaff-Schwartz method combined with \citet{thom2009longstaff} and \citet{broadie1996estimating}. OOM means ``out-of-memory".}
\end{table}


\begin{table}[h!]
\footnotesize
\begin{center}
geometric average call option
\\
\begin{tabular}{|c|c|c|c|c|c|c|c|c|}
\hline
\multirow{2}{*}[0em]{$s^0_i$}
& \multicolumn{4}{@{}c@{}|}
{proposed method}
& \multicolumn{4}{@{}c@{}|}
{Longstaff-Schwartz}
\\
\cline{2-9}
& $d=7$
& $d=13$
& $d=20$
& $d=100$
& $d=7$
& $d=13$
& $d=20$
& $d=100$
\\
\hline
90
&
0.96
&
0.95
&
0.96
&
0.95
&
0.72
&
0.56
&
0.42
&
OOM
\\
\hline
100
&
0.95
&
0.95
&
0.97
&
0.97
&
0.75
&
0.61
&
0.47
&
OOM
\\
\hline
110
&
0.98
&
0.96
&
0.96
&
0.97
&
0.78
&
0.65
&
0.51
&
OOM
\\
\hline
\end{tabular}
\end{center}
\caption{\label{tab:multi_geometric_vs_Longstaff_3}
Multi-dimensional geometric average call options: The f1-score of the exercise boundary classification.
OOM means ``out-of-memory".}
\end{table}



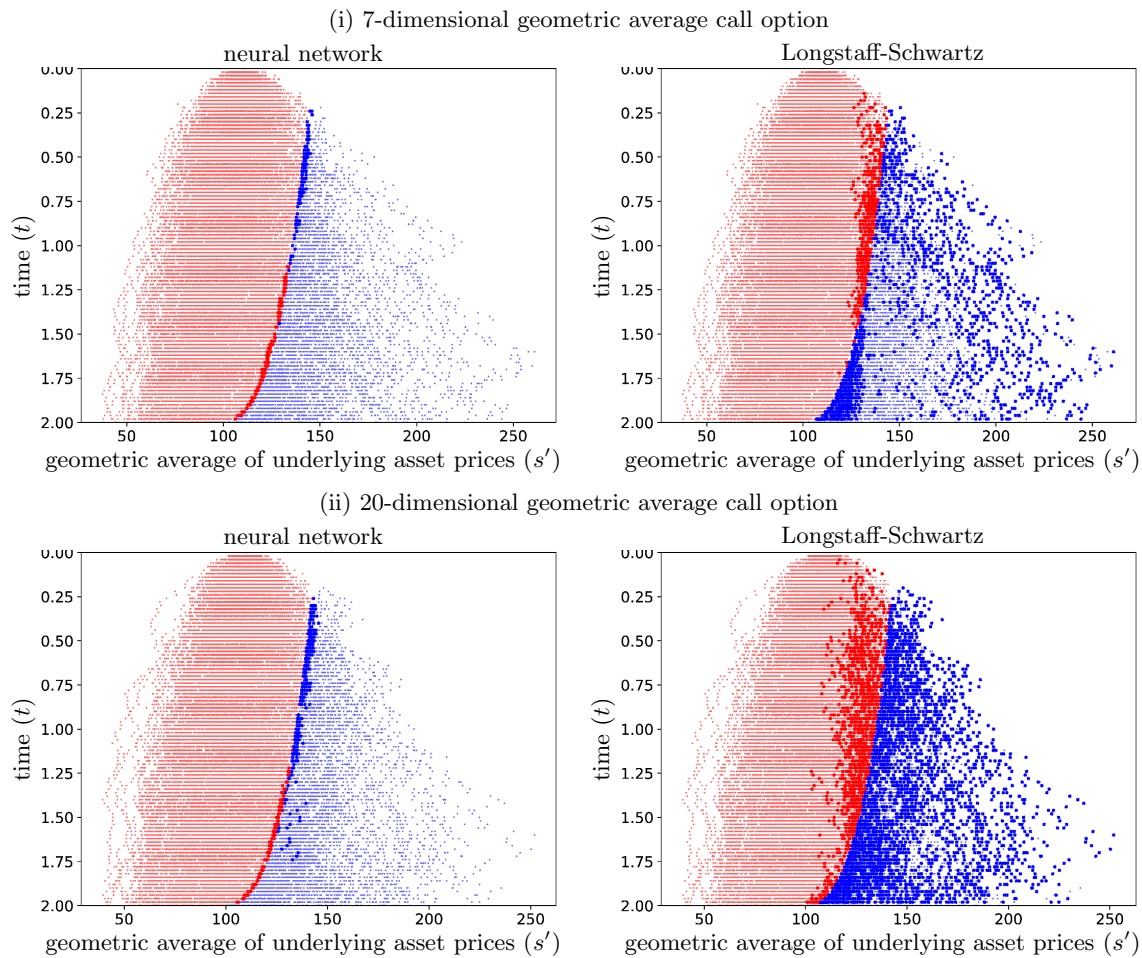
\begin{figure}[h!]
\centering
\footnotesize
(i) 7-dimensional geometric average call option
\\
\begin{tikzpicture}
\node (img1)
{\includegraphics[scale=0.4,trim={9mm 9mm 0 8mm},clip]{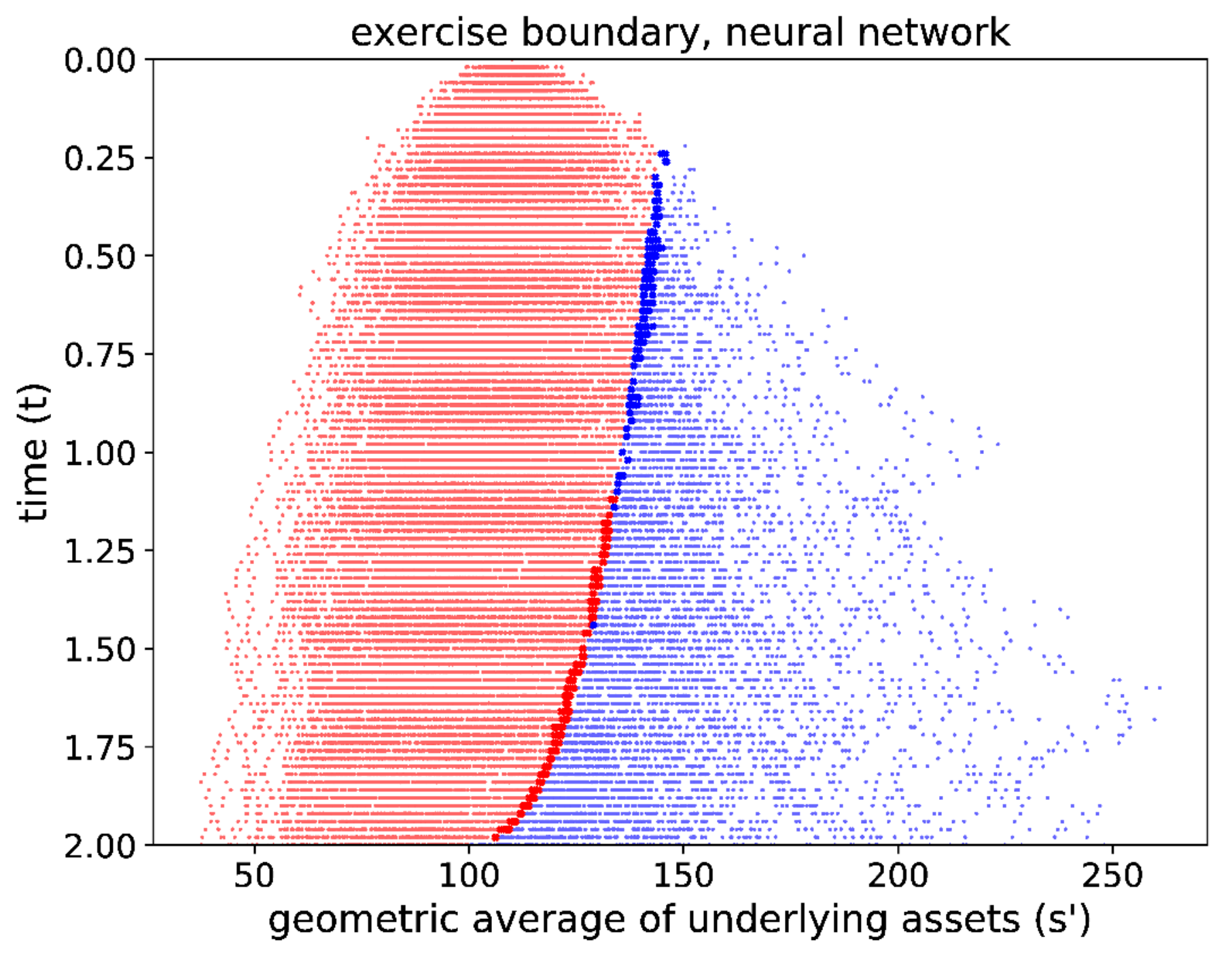}};
\node[below=of img1, node distance=0cm, yshift=1.2cm, xshift=0cm] {geometric average of underlying asset prices $(s')$};
\node[left=of img1, node distance=0cm, rotate=90, anchor=center, yshift=-0.9cm, xshift=0cm] {time ($t$)};
\node[above=of img1, node distance=0cm, yshift=-1.15cm, xshift=0cm] {neural network};
\end{tikzpicture}
\begin{tikzpicture}
\node (img1)
{\includegraphics[scale=0.4,trim={9mm 9mm 0 8mm},clip]{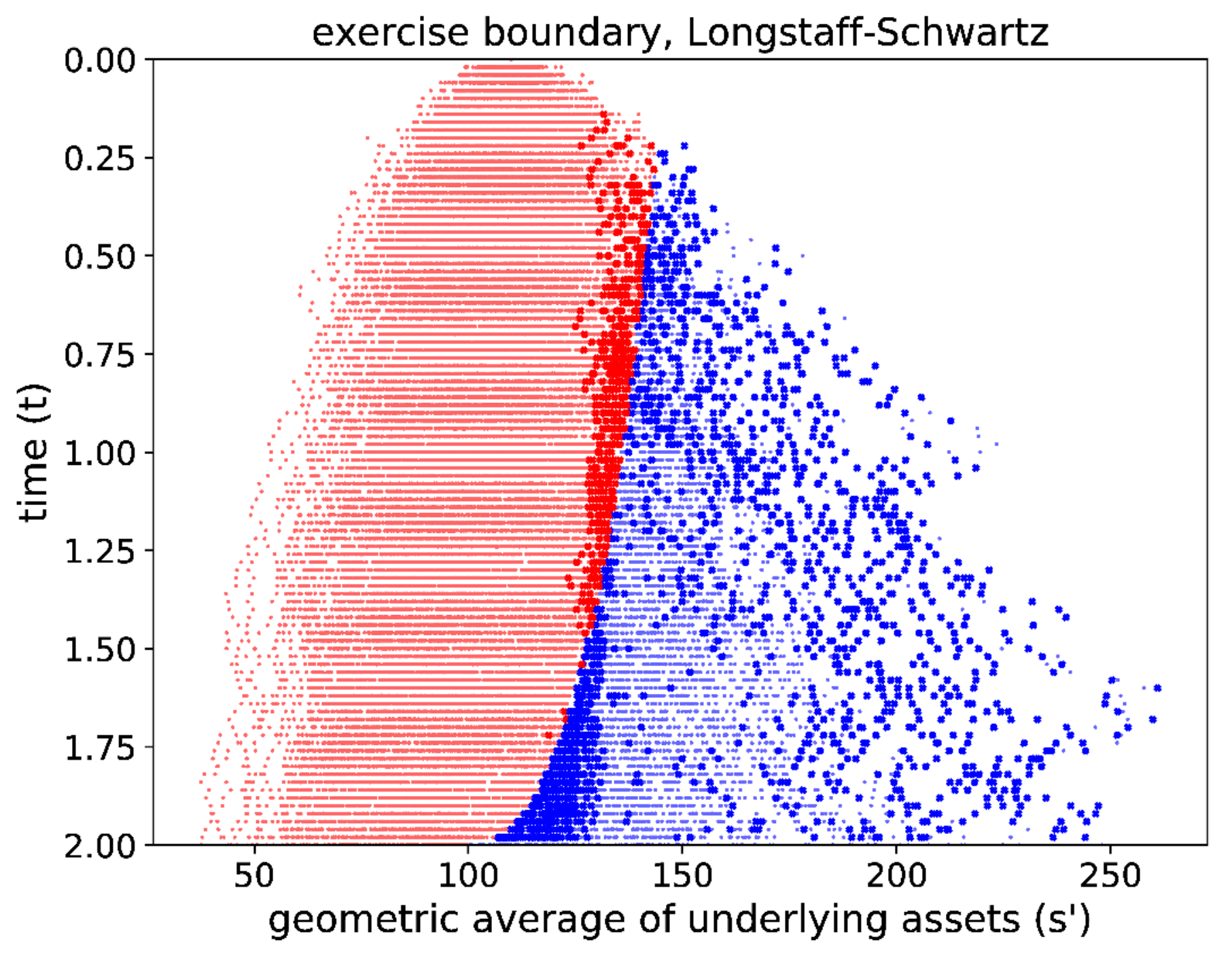}};
\node[below=of img1, node distance=0cm, yshift=1.2cm, xshift=0cm] {geometric average of underlying asset prices $(s')$};
\node[left=of img1, node distance=0cm, rotate=90, anchor=center, yshift=-0.9cm, xshift=0cm] {time ($t$)};
\node[above=of img1, node distance=0cm, yshift=-1.2cm, xshift=0cm] {Longstaff-Schwartz};
\end{tikzpicture}
\\
(ii) 20-dimensional geometric average call option
\\
\begin{tikzpicture}
\node (img1)
{\includegraphics[scale=0.4,trim={9mm 9mm 0 8mm},clip]{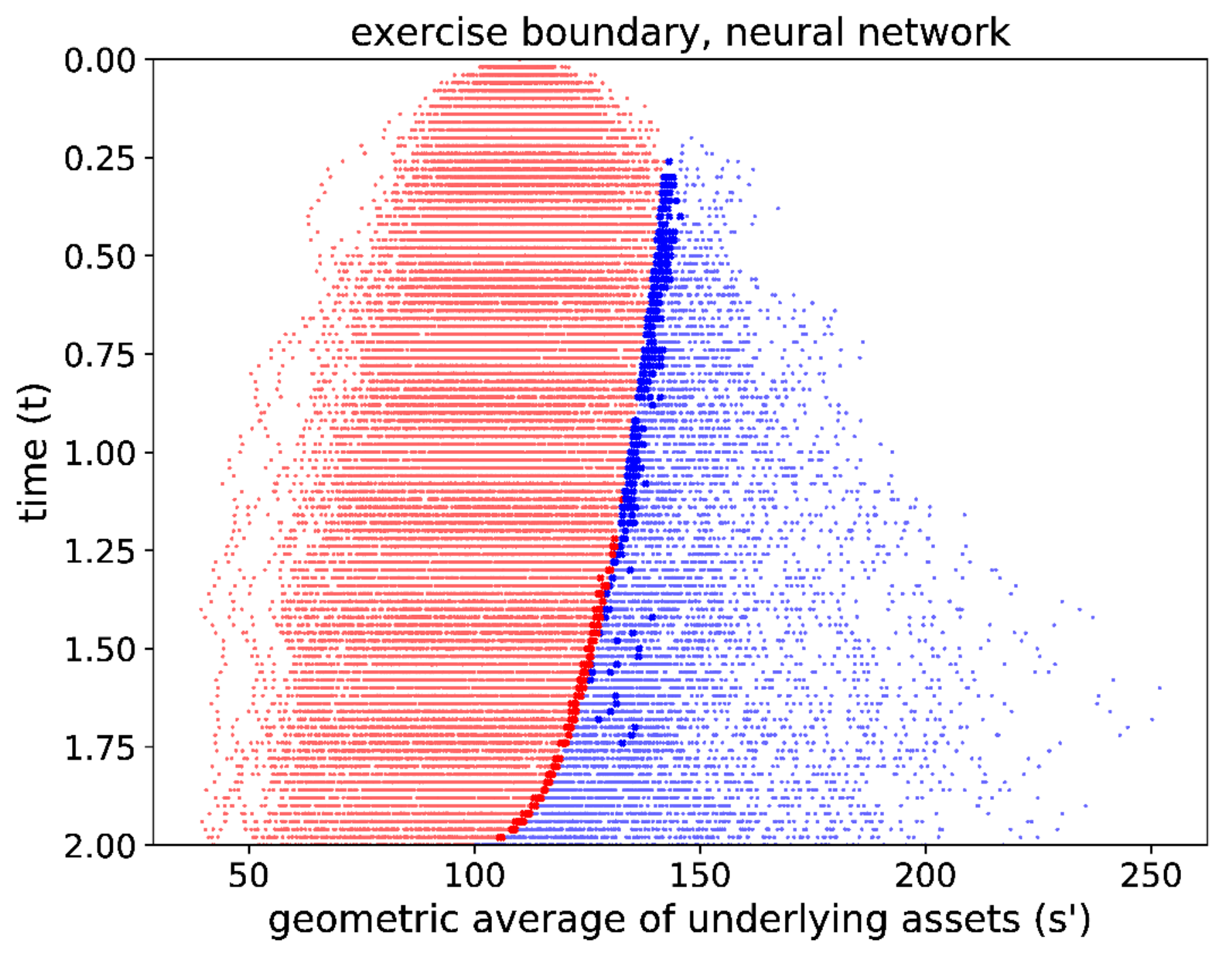}};
\node[below=of img1, node distance=0cm, yshift=1.2cm, xshift=0cm] {geometric average of underlying asset prices $(s')$};
\node[left=of img1, node distance=0cm, rotate=90, anchor=center, yshift=-0.9cm, xshift=0cm] {time ($t$)};
\node[above=of img1, node distance=0cm, yshift=-1.15cm, xshift=0cm] {neural network};
\end{tikzpicture}
\begin{tikzpicture}
\node (img1)
{\includegraphics[scale=0.4,trim={9mm 9mm 0 8mm},clip]{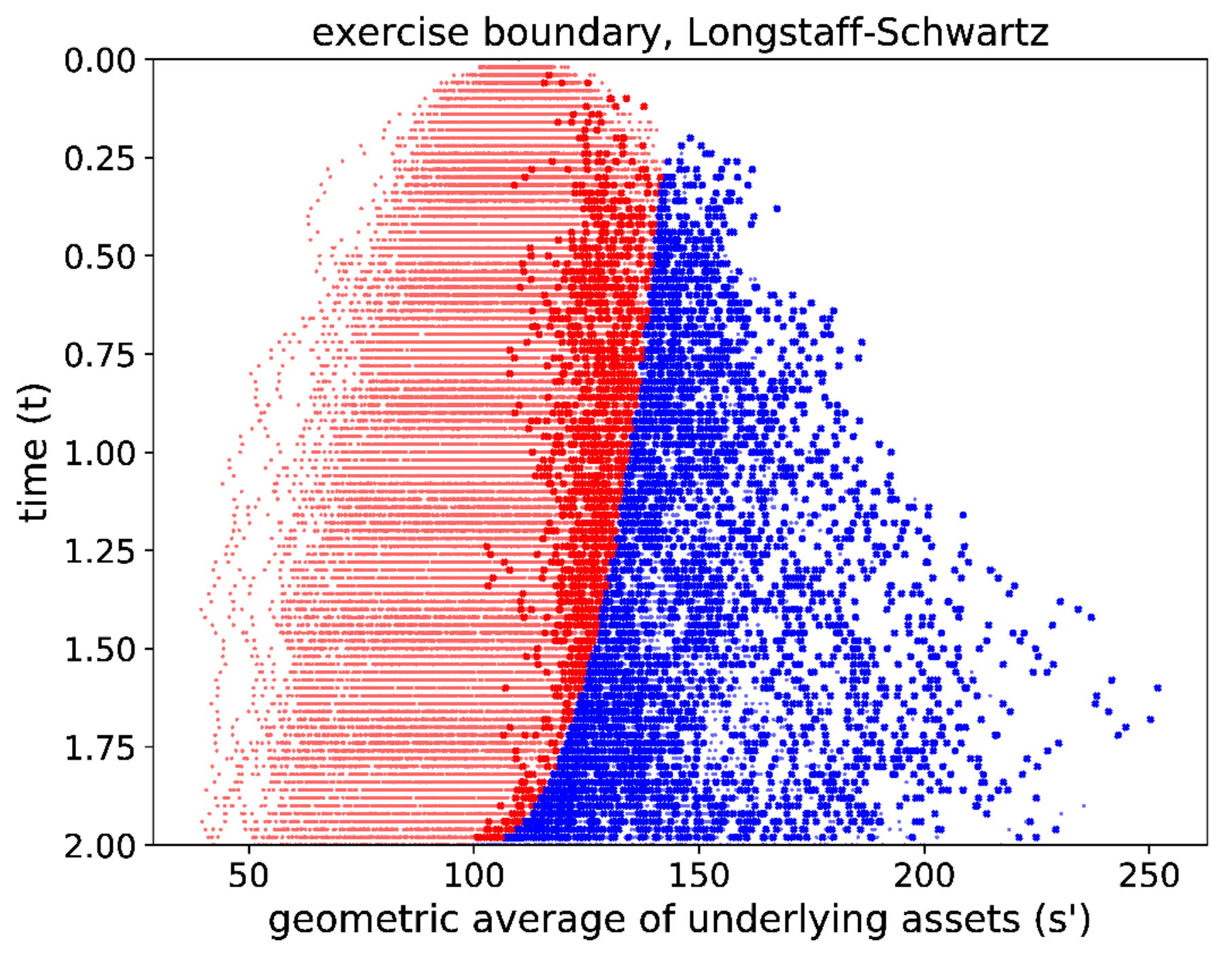}};
\node[below=of img1, node distance=0cm, yshift=1.2cm, xshift=0cm] {geometric average of underlying asset prices $(s')$};
\node[left=of img1, node distance=0cm, rotate=90, anchor=center, yshift=-0.9cm, xshift=0cm] {time ($t$)};
\node[above=of img1, node distance=0cm, yshift=-1.2cm, xshift=0cm] {Longstaff-Schwartz};
\end{tikzpicture}
\caption{\label{fig:multi_geometric_boundary}
Multi-dimensional geometric average call options: Comparison of exercise boundaries between the proposed neural network approach (top left and bottom left) and the Longstaff-Schwartz approach (top right and bottom right). All blue points: sample points that should be exercised; all red points: sample points that should be continued; bold dark blue points: sample points that should be exercised but are misclassified as continued; bold dark red points: sample points that should be continued but are misclassified as exercised.
}
\end{figure}


\subsubsection*{Experiment 1}
{\it Comparison between our proposed method and the Longstaff-Schwartz method.}
First we compare the computed prices at $t=0$; see Table \ref{tab:multi_geometric_vs_Longstaff_1}. Each sub-table includes: the exact prices computed by the Crank-Nicolson finite difference method with 1000 timesteps and 16385 space grid points, the prices and the corresponding percent errors computed by our proposed method, and the prices and the corresponding percent errors computed by the Longstaff-Schwartz method.
For the proposed method, the computed prices are accurate up to 2 decimal places; the percent errors are bounded by $0.34\%$, and remain approximately the same as the dimension increases.
As a comparison, for the Longstaff-Schwartz method, the percent errors deteriorate from $1\%$ to $9\%$ as the dimension increases from 7 to 20. If we keep increasing the dimension towards 100, the Longstaff-Schwartz method encounters an out-of-memory error, because, at $d=100$, it requires storing 
$\left(
\begin{smallmatrix}
d + \chi \\ d
\end{smallmatrix}
\right)CM = 3.3\times 10^{12}$
floating point numbers, or around 23TB of memory.

The Longstaff-Schwartz algorithm combined with the approaches in \citet{thom2009longstaff} and \citet{broadie1996estimating} can be used to compute the deltas at $t=0$. Table \ref{tab:multi_geometric_vs_Longstaff_2} compares the deltas at $t=0$ computed by our proposed approach with the ones computed by the Longstaff-Schwartz algorithm.
For the Longstaff-Schwartz algorithm, as the dimension increases from 7 to 20, the percent errors of the deltas worsen from $1.6\%$ to $12.7\%$; as the dimension continues to increase towards 100, an out-of-memory error occurs. However, for our proposed method, the computed deltas are accurate up to 3 decimal places; the percent errors do not increase with the dimension and stay below $1.7\%$.

Furthermore, we compare the exercise boundaries computed by the proposed neural network approach with the ones computed by the Longstaff-Schwartz approach.
Table \ref{tab:multi_geometric_vs_Longstaff_3} evaluates the f1-score of the exercise boundary classification, as defined in (\ref{eq:metric3}). For the proposed method, the f1-score remains around 0.95-0.98 as the dimension increases from 7 to 100. For the Longstaff-Schwartz algorithm, the f1-score drops from 0.78 to 0.42 as the dimension increases from 7 to 20. This illustrates a more precise exercise boundary determined by our proposed algorithm.

Figure \ref{fig:multi_geometric_boundary} visualizes the exercise boundaries computed by both algorithms. In order to visualize this, we start with $(\vec{s}^0,t^0)=(1.1K,0)$ and use the SDE (\ref{eq:Euler_S0})-(\ref{eq:Euler_S}) to generate sample points on the entire spacetime, i.e.,
$\{(\vec{S}^n_m,t^n) \, | \,
n = 0, ..., N; m = 1, ..., M
\}$;
we classify each sample point using either our proposed method, i.e., (\ref{eq:exercise}), or the Longstaff-Schwartz method; then we project these $(d+1)$-dimensional points onto the 2-dimensional points $\{({s'}^n_m,t^n)\}$, where
${s'}^n_m
= \left( \prod_{i=1}^d (S_i)^n_m \right)^{1/d}$
is the geometric average of the underlying asset prices $\vec{S}^n_m$.
We use bold dark blue to mark the sample points that should be exercised but are misclassified as continued, and bold dark red to mark the ones that should be continued but are misclassified as exercised. The plots show that the proposed neural network approach (top left and bottom left) has fewer misclassified sample points than the Longstaff-Schwartz approach (top right and bottom right). In other words, the proposed neural network approach yields more precise exercise boundaries.


\begin{table}[t!]
\footnotesize
\begin{center}
geometric average call option, $s_i^0=100$
\\
\begin{tabular}{|c|c|c|c|c|}
\hline
$d$
& exact price $v(\vec{s}^0,0)$
& mean of computed prices
& percent error
& 95\% CI
\\
\hline
7
&
10.2591
&
10.2468
&
0.12\%
&
$\pm$0.0161 ($\pm$0.16\%)
\\
\hline
13
&
10.0984
&
10.0822
&
0.16\%
&
$\pm$0.0201 ($\pm$0.20\%)
\\
\hline
20
&
10.0326
&
10.0116
&
0.21\%
&
$\pm$0.0173 ($\pm$0.17\%)
\\
\hline
100
&
9.9345
&
9.9163
&
0.18\%
&
$\pm$0.0038 ($\pm$0.04\%)
\\
\hline
\end{tabular}
\end{center}
\caption{\label{tab:multi_geometric_vs_Longstaff_1_CI}
Multi-dimensional geometric average call options: mean values and 95\% T-statistic confidence intervals (CIs) of the computed prices at $t=0$, i.e., $v(\vec{s}^0,0)$, using the proposed neural network method.}
\end{table}



\begin{table}[t!]
\footnotesize
\begin{center}
geometric average call option, $s_i^0=100$
\\
\begin{tabular}{|c|c|c|c|c|}
\hline
$d$
& exact delta $\nabla v(\vec{s}^0,0)$
& mean of computed deltas
& percent error
& 95\% CI of $\frac{\partial v}{\partial s_1} (\vec{s}^0,0)$
\\
\hline
7
&
(0.0722,$\cdots$,0.0722)
&
(0.0717,$\cdots$,0.0717)
&
0.67\%
&
$\pm 1.8\times 10^{-4}$ ($\pm$0.25\%)
\\
\hline
13
&
(0.0387,$\cdots$,0.0387)
&
(0.0384,$\cdots$,0.0384)
&
0.70\%
&
$\pm 7.3\times 10^{-5}$ ($\pm$0.19\%)
\\
\hline
20
&
(0.0251,$\cdots$,0.0251)
&
(0.0249,$\cdots$,0.0249)
&
0.78\%
&
$\pm 4.2\times 10^{-5}$ ($\pm$0.17\%)
\\
\hline
100
&
(0.00502,$\cdots$,0.00502)
&
(0.00498,$\cdots$,0.00498)
&
0.76\%
&
$\pm 8.9\times 10^{-6}$ ($\pm$0.18\%)
\\
\hline
\end{tabular}
\end{center}
\caption{\label{tab:multi_geometric_vs_Longstaff_2_CI}
Multi-dimensional geometric average call options: mean values of the computed deltas at $t=0$, i.e., $\nabla v(\vec{s}^0,0)$, using the proposed neural network method, and the corresponding 95\% T-statistic confidence intervals (CIs) of the first elements of deltas, i.e., $\frac{\partial v}{\partial s_1} (\vec{s}^0,0)$.}
\end{table}


\subsubsection*{Experiment 2}
{\it Confidence intervals by the proposed method.}
We repeat the experiments of computing the prices and deltas at $t=0$ (Tables \ref{tab:multi_geometric_vs_Longstaff_1}-\ref{tab:multi_geometric_vs_Longstaff_2}) for 9 times.
Tables \ref{tab:multi_geometric_vs_Longstaff_1_CI}-\ref{tab:multi_geometric_vs_Longstaff_2_CI} report the mean values of the computed prices and deltas, and the corresponding 95\% T-statistic confidence intervals. The last columns of the tables show that, for both the prices and the deltas, the deviations from the mean values remain a constant of $\pm 0.2\%$ as the dimension increases.


\begin{table}[h!]
\footnotesize
\begin{center}
(i) 7-dimensional geometric average call option
\\
\begin{tabular}{|c|c|c|c|c|}
\hline
\multirow{2}{*}[0em]{$s^0_i$}
& \multicolumn{2}{@{}c@{}|}
{spacetime price $v(\vec{s},t)$}
& \multicolumn{2}{@{}c@{}|}
{spacetime derivative
$\frac{\partial v}{\partial s'}(s',t)$}
\\
\cline{2-5}
& absolute error
& percent error
& absolute error
& percent error
\\
\hline
90
&
0.0688
&
1.2\%
&
0.0102
&
3.3\%
\\
\hline
100
&
0.0545
&
0.54\%
&
0.0102
&
2.3\%
\\
\hline
110
&
0.0450
&
0.29\%
&
0.0092
&
1.6\%
\\
\hline
\end{tabular}
\\
\smallskip
(ii) 13-dimensional geometric average call option
\\
\begin{tabular}{|c|c|c|c|c|}
\hline
\multirow{2}{*}[0em]{$s^0_i$}
& \multicolumn{2}{@{}c@{}|}
{spacetime price $v(\vec{s},t)$}
& \multicolumn{2}{@{}c@{}|}
{spacetime derivative
$\frac{\partial v}{\partial s'}(s',t)$}
\\
\cline{2-5}
& absolute error
& percent error
& absolute error
& percent error
\\
\hline
90
&
0.0540
&
0.94\%
&
0.0101
&
3.3\%
\\
\hline
100
&
0.0475
&
0.48\%
&
0.0106
&
2.4\%
\\
\hline
110
&
0.0465
&
0.30\%
&
0.0093
&
1.6\%
\\
\hline
\end{tabular}
\\
\smallskip
(iii) 20-dimensional geometric average call option
\\
\begin{tabular}{|c|c|c|c|c|}
\hline
\multirow{2}{*}[0em]{$s^0_i$}
& \multicolumn{2}{@{}c@{}|}
{spacetime price $v(\vec{s},t)$}
& \multicolumn{2}{@{}c@{}|}
{spacetime derivative
$\frac{\partial v}{\partial s'}(s',t)$}
\\
\cline{2-5}
& absolute error
& percent error
& absolute error
& percent error
\\
\hline
90
&
0.0567
&
1.00\%
&
0.0115
&
3.7\%
\\
\hline
100
&
0.0455
&
0.46\%
&
0.0111
&
2.5\%
\\
\hline
110
&
0.0397
&
0.26\%
&
0.0090
&
1.6\%
\\
\hline
\end{tabular}
\\
\smallskip
(iv) 100-dimensional geometric average call option
\\
\begin{tabular}{|c|c|c|c|c|}
\hline
\multirow{2}{*}[0em]{$s^0_i$}
& \multicolumn{2}{@{}c@{}|}
{spacetime price $v(\vec{s},t)$}
& \multicolumn{2}{@{}c@{}|}
{spacetime derivative
$\frac{\partial v}{\partial s'}(s',t)$}
\\
\cline{2-5}
& absolute error
& percent error
& absolute error
& percent error
\\
\hline
90
&
0.0534
&
0.96\%
&
0.0117
&
3.8\%
\\
\hline
100
&
0.0458
&
0.47\%
&
0.0107
&
2.4\%
\\
\hline
110
&
0.0480
&
0.31\%
&
0.0099
&
1.7\%
\\
\hline
\end{tabular}
\end{center}
\caption{\label{tab:multi_geometric_spacetime}
Multi-dimensional geometric average call options: Spacetime prices and deltas (in terms of absolute and percent errors) computed by our proposed method.
}
\end{table}



\begin{figure}[h!]
\centering
\footnotesize
100-dimensional geometric average call option
\\
\begin{tikzpicture}
\node (img1)
{\includegraphics[scale=0.35,trim={9mm 9mm 0 8mm},clip]{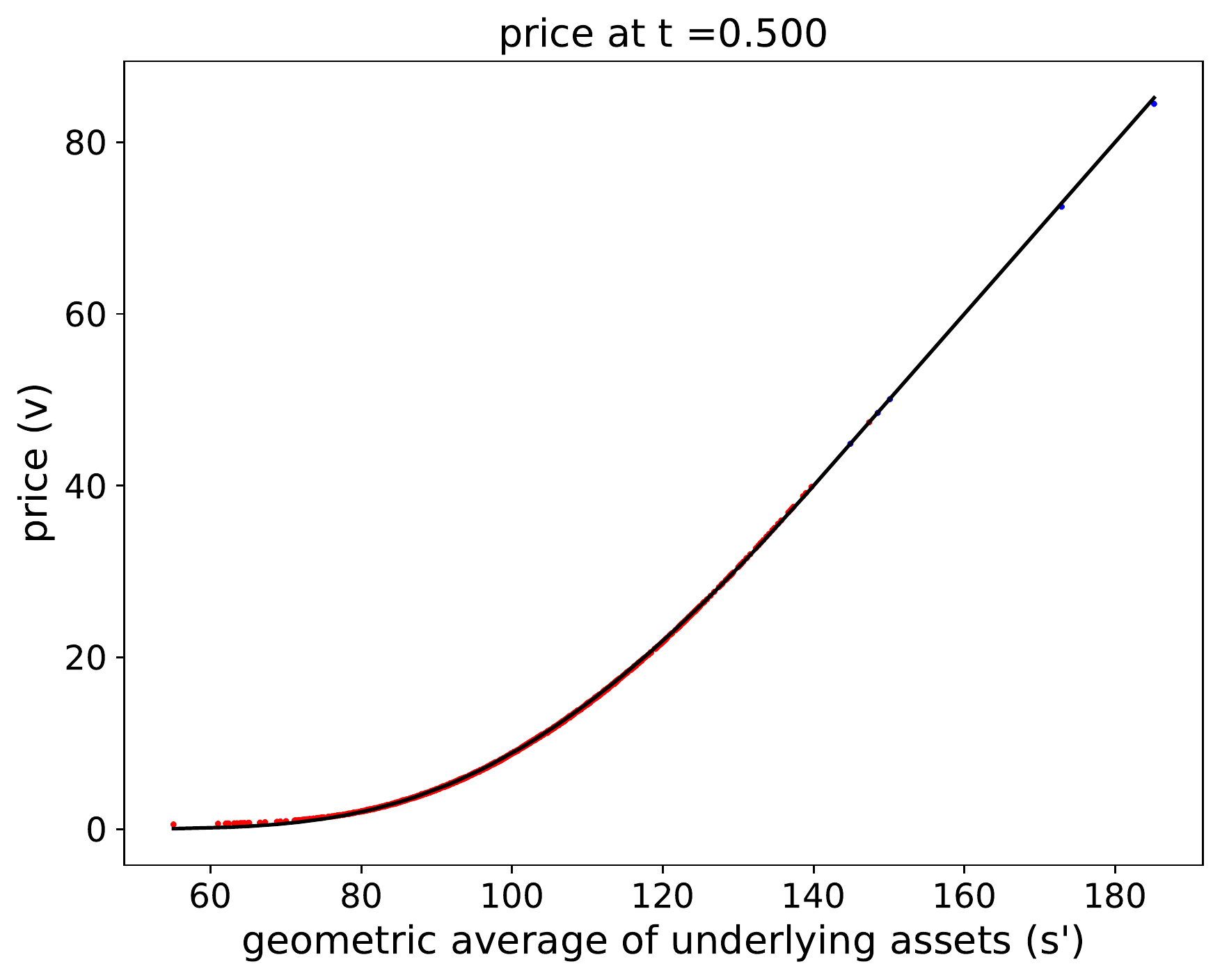}};
\node[below=of img1, node distance=0cm, yshift=1.2cm, xshift=0cm, font=\scriptsize] {geometric average of underlying asset prices $(s')$};
\node[left=of img1, node distance=0cm, rotate=90, anchor=center, yshift=-0.9cm, xshift=0cm] {price ($v$)};
\node[above=of img1, node distance=0cm, yshift=-1.2cm, xshift=0cm] {price at $t=0.5$};
\end{tikzpicture}
\begin{tikzpicture}
\node (img1)
{\includegraphics[scale=0.35,trim={9mm 9mm 0 8mm},clip]{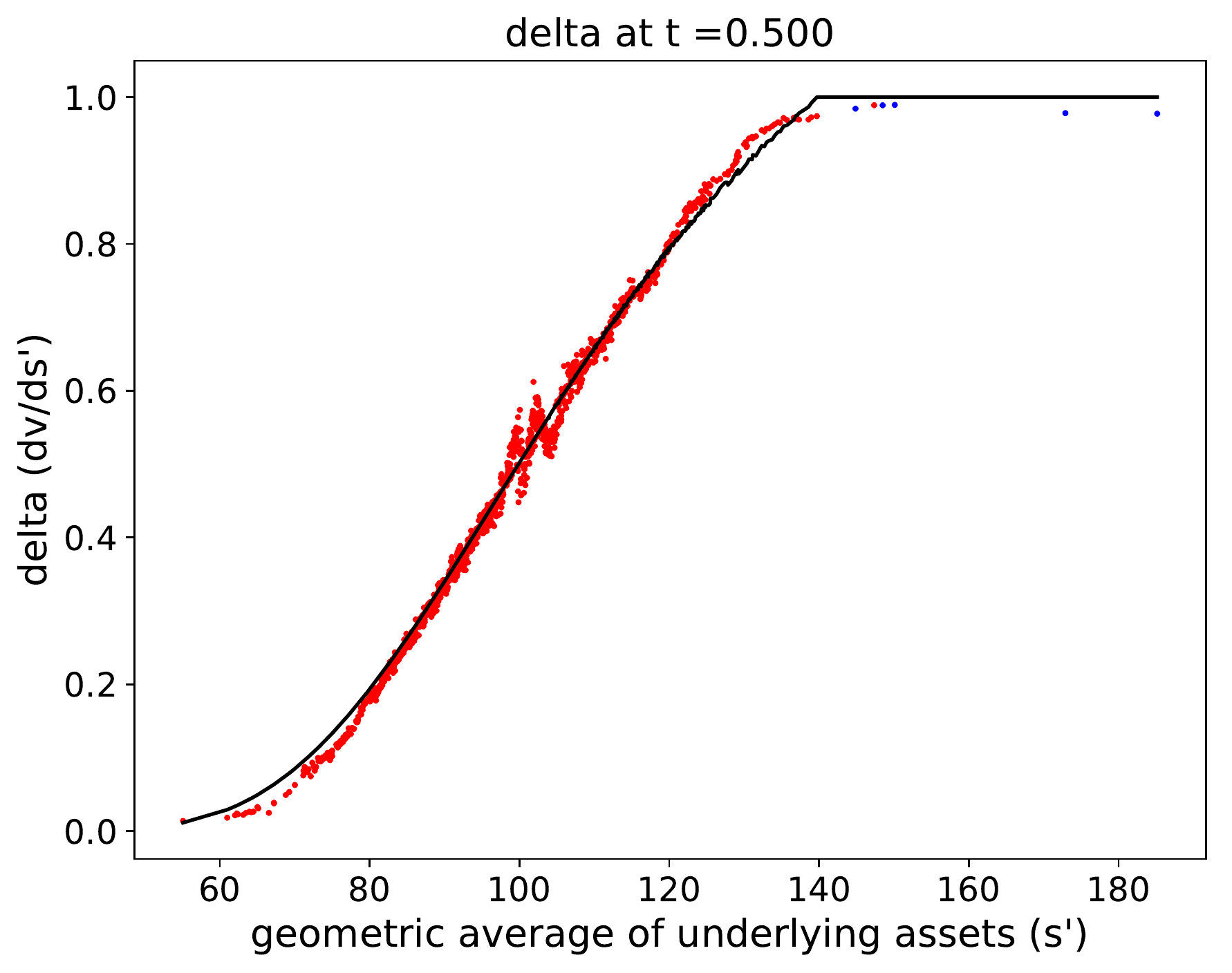}};
\node[below=of img1, node distance=0cm, yshift=1.2cm, xshift=0cm, font=\scriptsize] {geometric average of underlying asset prices $(s')$};
\node[left=of img1, node distance=0cm, rotate=90, anchor=center, yshift=-0.9cm, xshift=0cm] {delta ($dv/ds'$)};
\node[above=of img1, node distance=0cm, yshift=-1.15cm, xshift=0cm] {delta at $t=0.5$};
\end{tikzpicture}
\\
\begin{tikzpicture}
\node (img1)
{\includegraphics[scale=0.35,trim={9mm 9mm 0 8mm},clip]{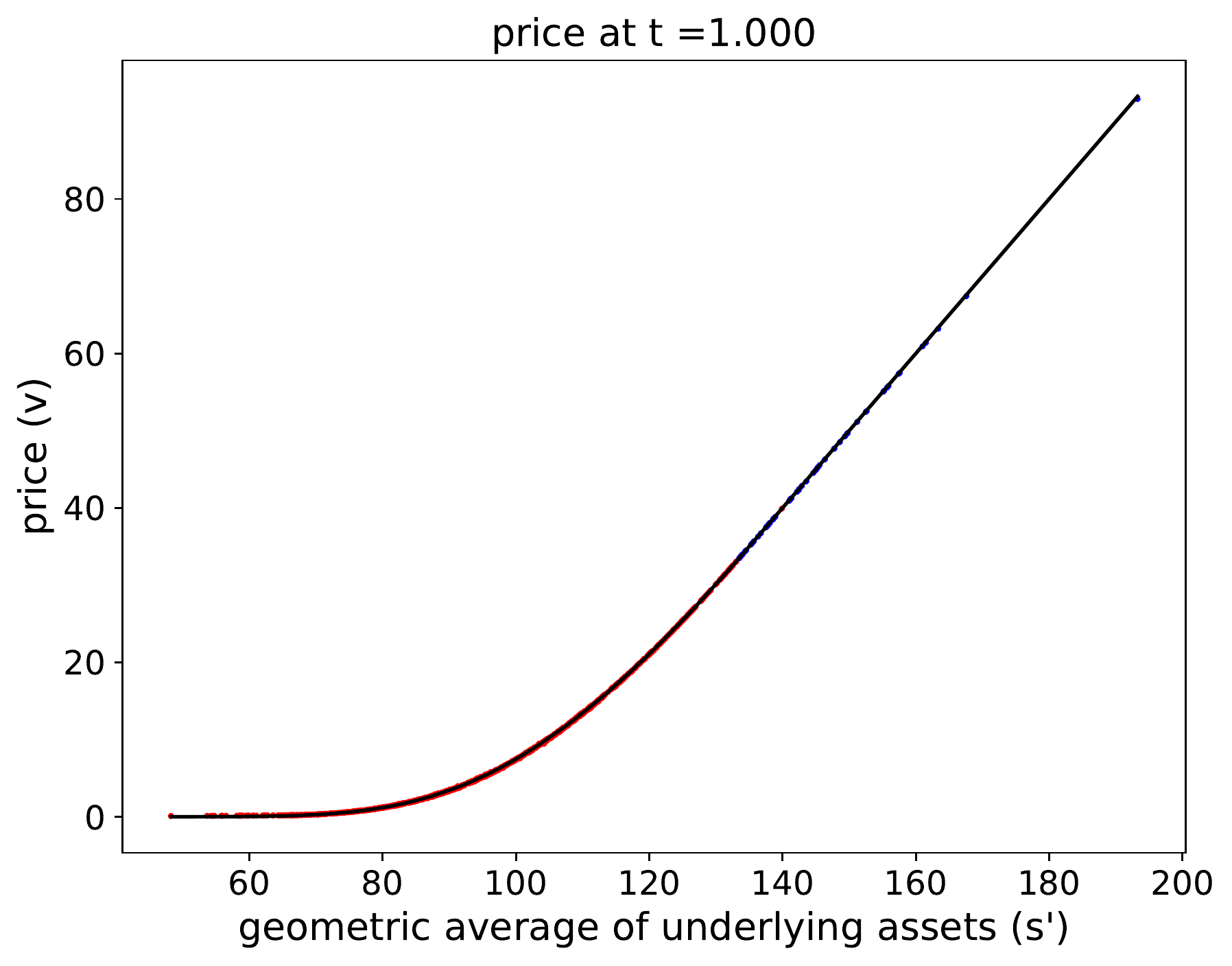}};
\node[below=of img1, node distance=0cm, yshift=1.2cm, xshift=0cm, font=\scriptsize] {geometric average of underlying asset prices $(s')$};
\node[left=of img1, node distance=0cm, rotate=90, anchor=center, yshift=-0.9cm, xshift=0cm] {price ($v$)};
\node[above=of img1, node distance=0cm, yshift=-1.2cm, xshift=0cm] {price at $t=1.0$};
\end{tikzpicture}
\begin{tikzpicture}
\node (img1)
{\includegraphics[scale=0.35,trim={9mm 9mm 0 8mm},clip]{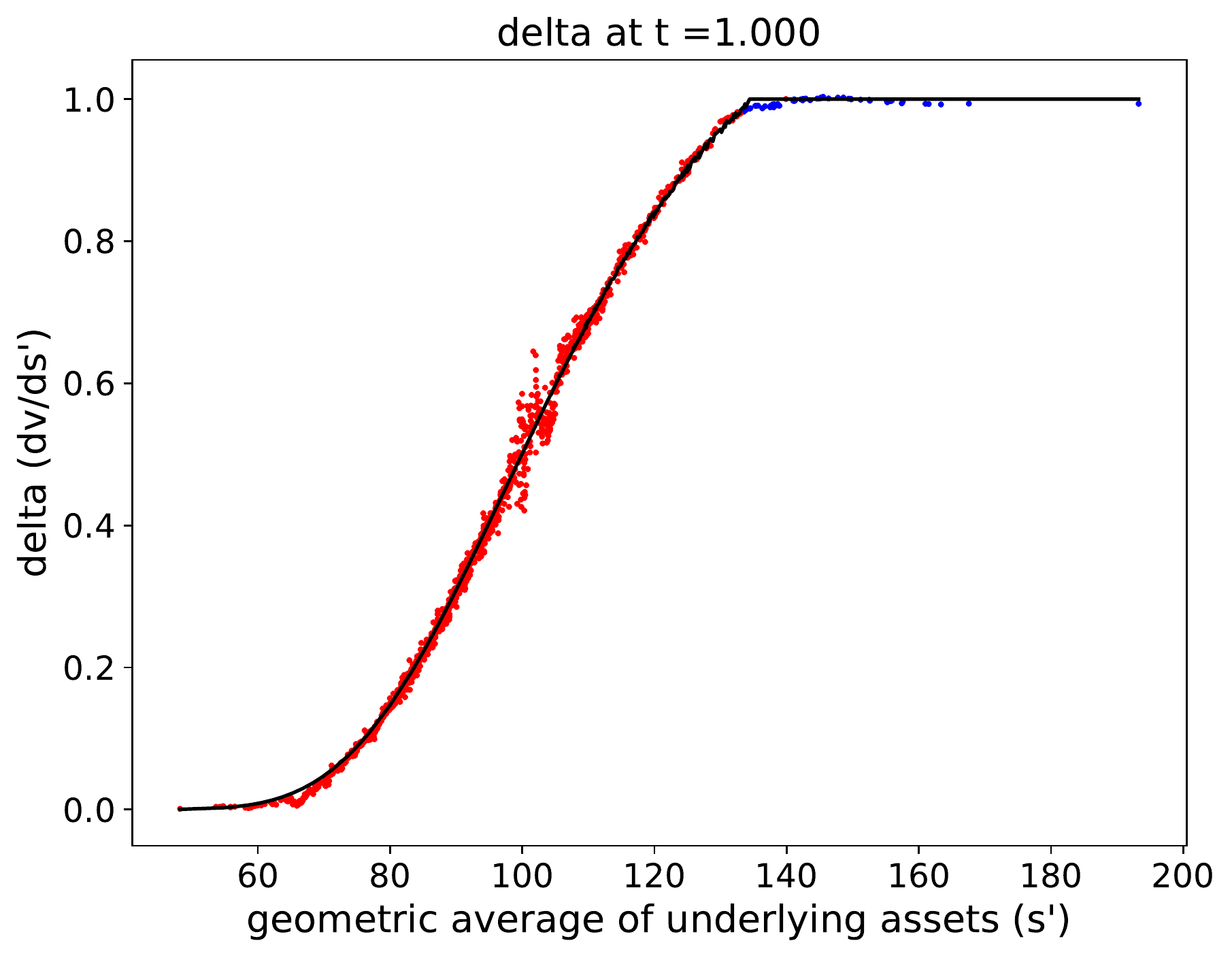}};
\node[below=of img1, node distance=0cm, yshift=1.2cm, xshift=0cm, font=\scriptsize] {geometric average of underlying asset prices $(s')$};
\node[left=of img1, node distance=0cm, rotate=90, anchor=center, yshift=-0.9cm, xshift=0cm] {delta ($dv/ds'$)};
\node[above=of img1, node distance=0cm, yshift=-1.15cm, xshift=0cm] {delta at $t=1.0$};
\end{tikzpicture}
\\
\begin{tikzpicture}
\node (img1)
{\includegraphics[scale=0.35,trim={9mm 9mm 0 8mm},clip]{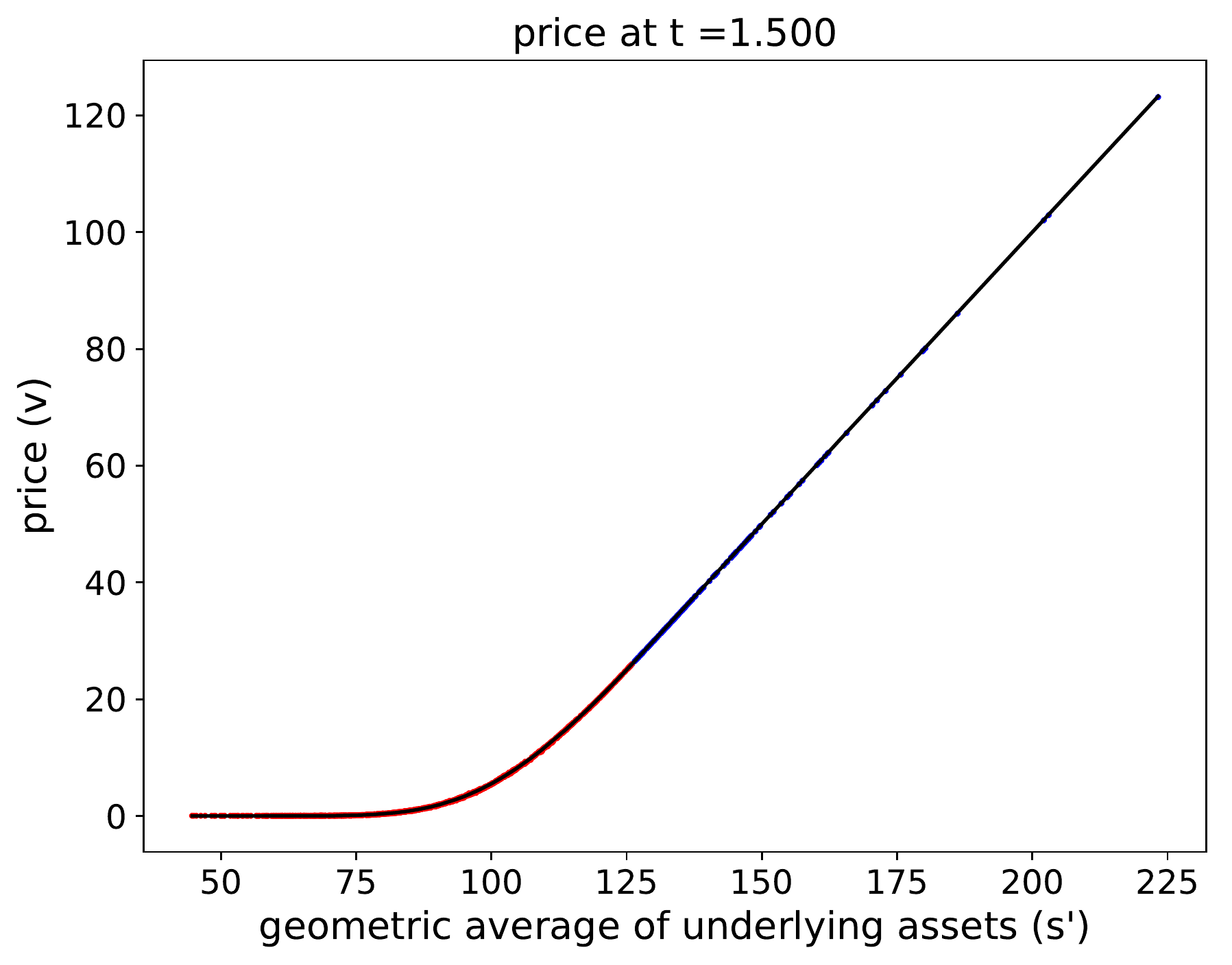}};
\node[below=of img1, node distance=0cm, yshift=1.2cm, xshift=0cm, font=\scriptsize] {geometric average of underlying asset prices $(s')$};
\node[left=of img1, node distance=0cm, rotate=90, anchor=center, yshift=-0.9cm, xshift=0cm] {price ($v$)};
\node[above=of img1, node distance=0cm, yshift=-1.2cm, xshift=0cm] {price at $t=1.5$};
\end{tikzpicture}
\begin{tikzpicture}
\node (img1)
{\includegraphics[scale=0.35,trim={9mm 9mm 0 8mm},clip]{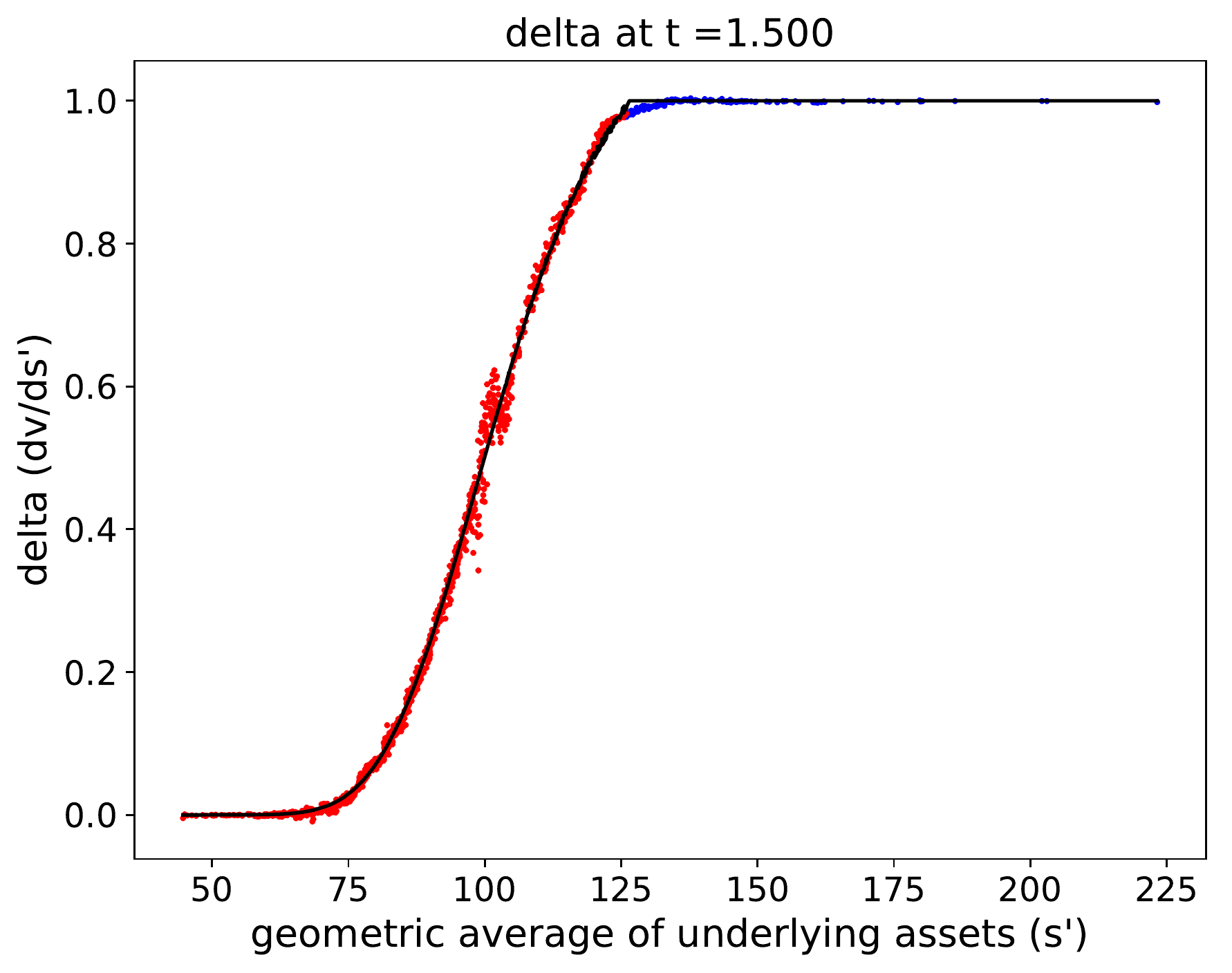}};
\node[below=of img1, node distance=0cm, yshift=1.2cm, xshift=0cm, font=\scriptsize] {geometric average of underlying asset prices $(s')$};
\node[left=of img1, node distance=0cm, rotate=90, anchor=center, yshift=-0.9cm, xshift=0cm] {delta ($dv/ds'$)};
\node[above=of img1, node distance=0cm, yshift=-1.15cm, xshift=0cm] {delta at $t=1.5$};
\end{tikzpicture}
\caption{\label{fig:multi_geometric_spacetime}
100-dimensional geometric average call option:
Prices (left subplots) and deltas (right subplots) computed by the proposed neural network approach at $t=$0.5, 1.0, 1.5. The blue/red dots are neural network output values of the exercised/continued sample points. The black lines
are the exact solutions computed by finite difference methods.
}
\end{figure}


\subsubsection*{Experiment 3}
{\it Evaluation of computed spacetime prices and deltas by the proposed method.}
Our proposed algorithm yields not only the prices and deltas at $t=0$, but also the prices and deltas on the entire spacetime, which are directly extracted from the output of the neural networks. We emphasize that the computation of spacetime prices and deltas using the Longstaff-Schwartz method is infeasible. The reason is that using the Longstaff-Schwartz method to compute prices and deltas on the entire spacetime would require repeating the algorithm at {\it every} sample point, noting that the Longstaff-Schwartz method at one sample point is already non-trivial. We also remark that although one may consider using the Longstaff-Schwartz regressed values as an estimate of the spacetime prices, Figure 1 in \citet{bouchard2012monte} shows that using such regressed values as the spacetime solution is inaccurate.

First we evaluate the absolute and percent errors of the spacetime price $v(\vec{s},t)$ and the derivative $\frac{\partial v}{\partial s'}(s',t)$ computed by our proposed method. Here we evaluate the errors of the derivative $\frac{\partial v}{\partial s'}(s',t)$ instead of the delta $\vec{\nabla} v(\vec{s},t)$, because the exact values of the former can be computed by finite difference method spacetime-wise, but not the latter. Table \ref{tab:multi_geometric_spacetime} shows that the absolute errors of the spacetime prices and derivatives are around 0.04-0.07 and 0.01 respectively, or in other words, the spacetime prices and derivatives are accurate up to 2 decimal places; the percent errors are less than 1.2\% and 3.8\%, respectively. We note that the percent errors of the spacetime prices and deltas (Table \ref{tab:multi_geometric_spacetime}) are slightly larger than the percent errors of the prices and deltas at $t=0$ (Tables \ref{tab:multi_geometric_vs_Longstaff_1}-\ref{tab:multi_geometric_vs_Longstaff_2}). This is expected, as the values at $t=0$ are computed by the improved approach described in Section \ref{subsec:initial}.

To visualize the spacetime solutions, we consider the 100-dimensional case, select three time slices $t=$ 0.5, 1.0, 1.5, and project the 100-dimensional sample points of $v(\vec{s},t)$ and $\vec{\nabla} v(\vec{s},t)$ to 1-dimensional points of $v(s',t)$ and $\frac{\partial v}{\partial s'}(s',t)$, as shown in Figure \ref{fig:multi_geometric_spacetime}. The spacetime option prices and deltas computed by the proposed neural network approach (the blue/red dots) agree well with the exact solutions by finite difference methods (black lines). We note that small fluctuations exist for the computed spacetime deltas (right subplots), especially near the strike price $K=100$. This is expected, as the deltas of the payoff functions are discontinuous at the strike price. Smoothing the payoff, as described in Section \ref{subsec:smooth}, can mitigate this issue, although it does not eliminate the fluctuations.

\subsubsection*{Experiment 4}
{\it Comparison between our proposed method and the method in \citet{sirignano2018dgm}.}
First we compare the computed prices at $t=0$; see Table \ref{tab:multi_geometric_vs_sirignano_1}. Up to 200 dimension is tested. In particular, by comparing the last two columns of the table, we observe that the percent errors computed by our method are bounded by $0.17\%$, while the ones computed by \citet{sirignano2018dgm} are bounded by $0.22\%$.

Next we compare the computed spacetime prices by the two approaches.
Figure \ref{fig:multi_geometric_vs_sirignano_1} compares the absolute errors of the spacetime prices.
To plot the figure, we start with $(\vec{s}^0,t^0)=(K,0)$ and use the SDE (\ref{eq:Euler_S0})-(\ref{eq:Euler_S}) to generate sample points on the entire spacetime, i.e.,
$\{(\vec{S}^n_m,t^n) \, | \,
n = 0, ..., N; m = 1, ..., M
\}$.
We compute the error at each sample point,
$e(\vec{S}^n_m,t^n)
\equiv |v(\vec{S}^n_m,t^n)-v_{exact}(\vec{S}^n_m,t^n)|$.
Then we project $\{e(\vec{S}^n_m,t^n)\}$ from $(d+1)$-dimensional to 2-dimensional space and get the sample points $\{e({s'}^n_m,t^n)\}$, where ${s'}^n_m$ is the geometric average of $\vec{S}^n_m$.
From the discrete data points $\{e({s'}^n_m,t^n)\}$, we use interpolation to obtain a continuous error function $e(s',t)$  and represent it by a heatmap (also known as filled contour plot), where the $x$ and $y$ axes are the time $t$ and the geometric average $s'$, and the color represents the magnitude of $e(s',t)$. The red, green and blue areas represent the areas where the samples have large, median and small errors, respectively. The white areas are the areas outside the convex hull of the sampled points, where no value of $e(s',t)$ can be interpolated from the sampled $\{e({s'}^n_m,t^n)\}$. We remark that this plotting procedure is the same as \citet{sirignano2018dgm}. Indeed, the right subplot of Figure \ref{fig:multi_geometric_vs_sirignano_1} is directly taken from \citet{sirignano2018dgm}. In addition, we note that the colored areas of the left and right subplots are not exactly the same. This is because
the points on (or near) the boundary of the convex hull are only sampled with a small probability and would have a large variation under the two independent stochastic sampling processes that generate the two subplots.


\begin{table}[b!]
\footnotesize
\begin{center}
geometric average call option, $s^0_i=100$
\\
\begin{tabular}{|c|c|c|c|c|}
\hline
\multirow{2}{*}[0em]{$d$}
& \multirow{2}{*}[0em]
{\begin{tabular}{@{}c@{}}exact price\\$v(\vec{s}^0,0)$\end{tabular}}
& \multicolumn{2}{@{}c@{}|}
{proposed method}
& \citet{sirignano2018dgm}
\\
\cline{3-5}
& 
& computed price $v(\vec{s}^0,0)$
& percent error
& percent error
\\
\hline
3
&
10.7185
&
10.7368
&
0.17\%
&
0.05\%
\\
\hline
20
&
10.0326
&
10.0180
&
0.15\%
&
0.03\%
\\
\hline
100
&
9.9345
&
9.9187
&
0.16\%
&
0.11\%
\\
\hline
200
&
9.9222
&
9.9088
&
0.14\%
&
0.22\%
\\
\hline
\end{tabular}
\end{center}
\caption{\label{tab:multi_geometric_vs_sirignano_1}
Multi-dimensional geometric average call options: Computed prices at $t=0$, i.e., $v(\vec{s}^0,0)$. $s^0_i=100$.
The percent errors reported in Table 1 of \citet{sirignano2018dgm} are also included in the last column of this table.
}
\end{table}



\begin{table}[b!]
\footnotesize
\begin{center}
geometric average call option, $s^0_i=100$
\\
\begin{tabular}{|c|c|c|c|}
\hline
\multirow{2}{*}[0em]{$d$}
& \multirow{2}{*}[0em]
{\begin{tabular}{@{}c@{}}exact delta
\\$\vec{\nabla}v(\vec{s}^0,0)$\end{tabular}}
& \multicolumn{2}{@{}c@{}|}
{proposed method}
\\
\cline{3-4}
& 
& computed delta $\vec{\nabla}v(\vec{s}^0,0)$
& percent error
\\
\hline
3
&
(0.1702,$\cdots$,0.1702)
&
(0.1683,$\cdots$,0.1683)
&
1.1\%
\\
\hline
20
&
(0.0251,$\cdots$,0.0251)
&
(0.0248,$\cdots$,0.0248)
&
1.2\%
\\
\hline
100
&
(0.00502,$\cdots$,0.00502)
&
(0.00495,$\cdots$,0.00495)
&
1.3\%
\\
\hline
200
&
(0.00251,$\cdots$,0.00251)
&
(0.00250,$\cdots$,0.00250)
&
0.53\%
\\
\hline
\end{tabular}
\end{center}
\caption{\label{tab:multi_geometric_vs_sirignano_2}
Multi-dimensional geometric average call options: Computed deltas at $t=0$, i.e., $\vec{\nabla}v(\vec{s}^0,0)$. $s^0_i=100$.
}
\end{table}



\begin{figure}[h!]
\begin{footnotesize}
\hspace{30mm}
proposed method
\hspace{40mm}
\citet{sirignano2018dgm}
\begin{center}
\includegraphics[height=5cm]{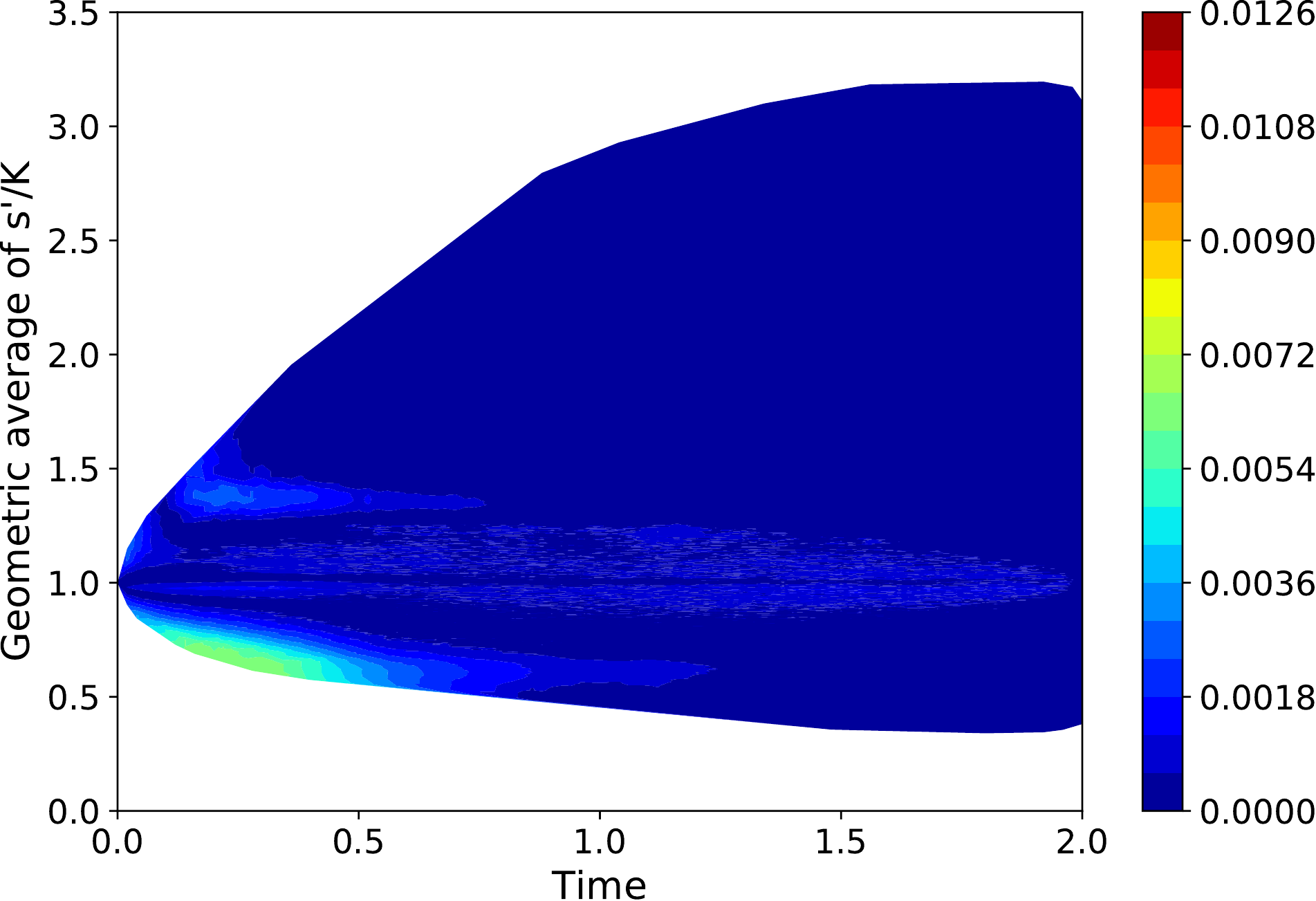}
\includegraphics[height=5cm]{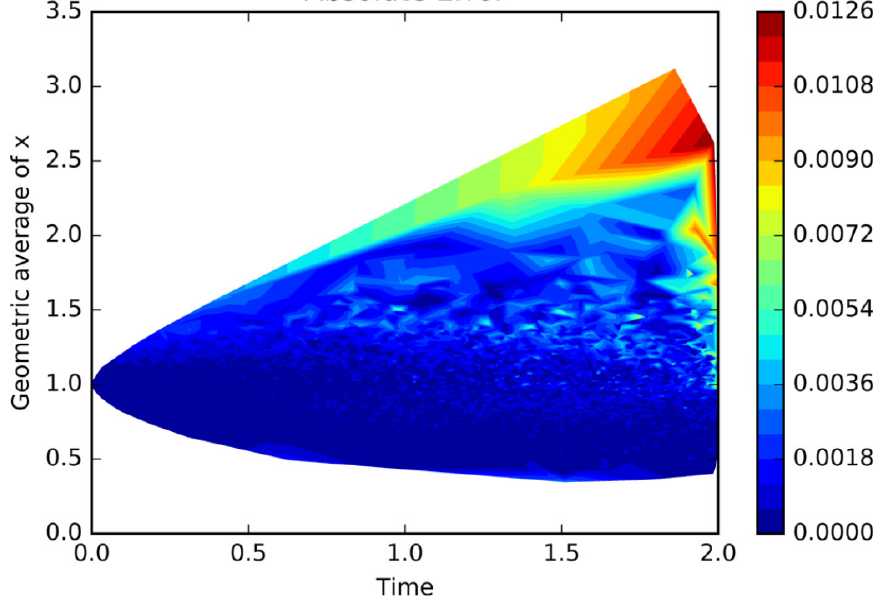}
\end{center}
\end{footnotesize}
\caption{\label{fig:multi_geometric_vs_sirignano_1}
20-dimensional geometric average call options: Heatmaps of the absolute errors of the computed spacetime prices. Left: absolute error computed by the proposed approach; right: absolute error computed by \citet{sirignano2018dgm}.
}
\end{figure}



\begin{figure}[h!]
\begin{footnotesize}
\hspace{30mm}
proposed method
\hspace{40mm}
\citet{sirignano2018dgm}
\begin{center}
\includegraphics[height=5cm]{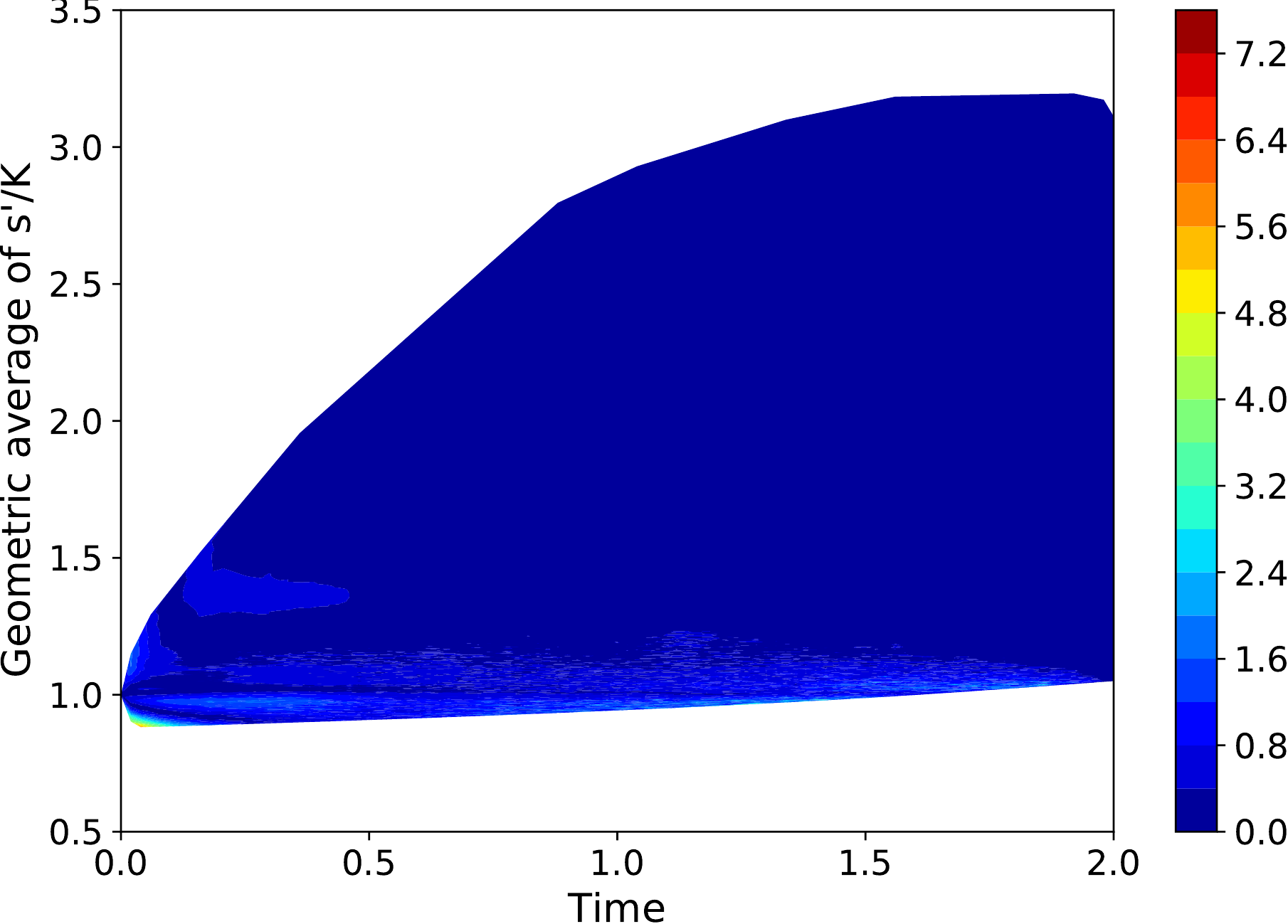}
\includegraphics[height=5cm]{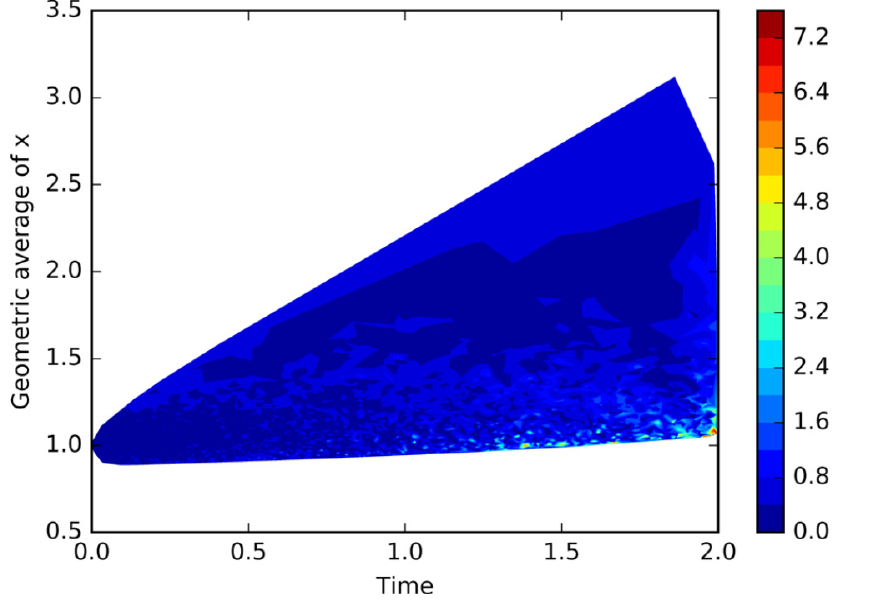}
\end{center}
\end{footnotesize}
\caption{\label{fig:multi_geometric_vs_sirignano_2}
20-dimensional geometric average call options: Heatmaps of the percent errors of the computed spacetime prices. Left: percent error computed by the proposed approach; right: percent error computed by \citet{sirignano2018dgm}.
}
\end{figure}


Figure \ref{fig:multi_geometric_vs_sirignano_1} (left) shows that the absolute error computed by our proposed approach is close to zero almost on the entire spacetime domain. The error is slightly larger near $(t,s'/K)\approx (0.2,0.7)$ and bounded by 0.0072. The reason why the error is slightly larger near $t=0$ is that our proposed approach computes the price in a backward manner, and hence the error may accumulate near $t=0$. As a comparison, Figure \ref{fig:multi_geometric_vs_sirignano_1} (right) shows that the error computed by \citet{sirignano2018dgm} has a larger error in most of the spacetime domain. In particular, the error reaches 0.0126 near $(t,s'/K)\approx (2.0,2.7)$, which is larger than the upper bound of our error, 0.0072.

Figure \ref{fig:multi_geometric_vs_sirignano_2} compares the heatmaps of the corresponding percent errors. Following \citet{sirignano2018dgm}, the percent errors are only plotted for the areas where $|v_{exact}(s',t)|>0.05$. Similar to Figure \ref{fig:multi_geometric_vs_sirignano_1}, Figure \ref{fig:multi_geometric_vs_sirignano_2} (left) shows that our proposed approach yields zero error almost everywhere, except that near $(t,s'/K)\approx (0.05,0.9)$ the error reaches $5.6\%$. Figure \ref{fig:multi_geometric_vs_sirignano_2} (right) shows that the approach in \citet{sirignano2018dgm} results in a larger error, particularly near $(t,s'/K)\approx (2.0,1.05)$, where the error reaches $7.2\%$.

We emphasize that \citet{sirignano2018dgm} does not compute deltas, whereas our proposed method does yield the deltas. Table \ref{tab:multi_geometric_vs_sirignano_2} reports the deltas at $t=0$ computed by our proposed method. The percent errors are bounded by $1.3\%$, and remain approximately the same as the dimension increases. Our approach also computes spacetime deltas, which has been discussed in Experiment 3 and is thus skipped here.


\begin{table}[b!]
\footnotesize
\begin{center}
geometric average call option
\\
\begin{tabular}{|c|c|c|c|c|c|c|c|c|}
\hline
\multirow{2}{*}[0em]{$s^0_i$}
& \multicolumn{2}{@{}c@{}|}
{$d=7$}
& \multicolumn{2}{@{}c@{}|}
{$d=13$}
& \multicolumn{2}{@{}c@{}|}
{$d=20$}
& \multicolumn{2}{@{}c@{}|}
{$d=100$}
\\
\cline{2-9}
& mean
& std
& mean
& std
& mean
& std
& mean
& std
\\
\hline
90
&
-0.0023 & 0.1788
&
0.0017 & 0.1827
&
-0.0003 & 0.1877
&
-0.0021 & 0.1908
\\
\hline
100
&
-0.0016 & 0.1159
&
0.0021 & 0.1170
&
-0.0007 & 0.1184
&
-0.0010 & 0.1184
\\
\hline
110
&
-0.0001 & 0.0757
&
0.0013 & 0.0755
&
0.0005 & 0.0751
&
-0.0009 & 0.0763
\\
\hline
\end{tabular}
\end{center}
\caption{\label{tab:multi_geometric_hedge}
Multi-dimensional geometric average call options: Computed means and standard deviations of the relative P\&Ls, subject to 100 hedging intervals.
}
\end{table}



\begin{figure}[b!]
\centering
\footnotesize
\begin{tikzpicture}
\node (img1)
{\includegraphics[scale=0.35,trim={9mm 9mm 0 8mm},clip]{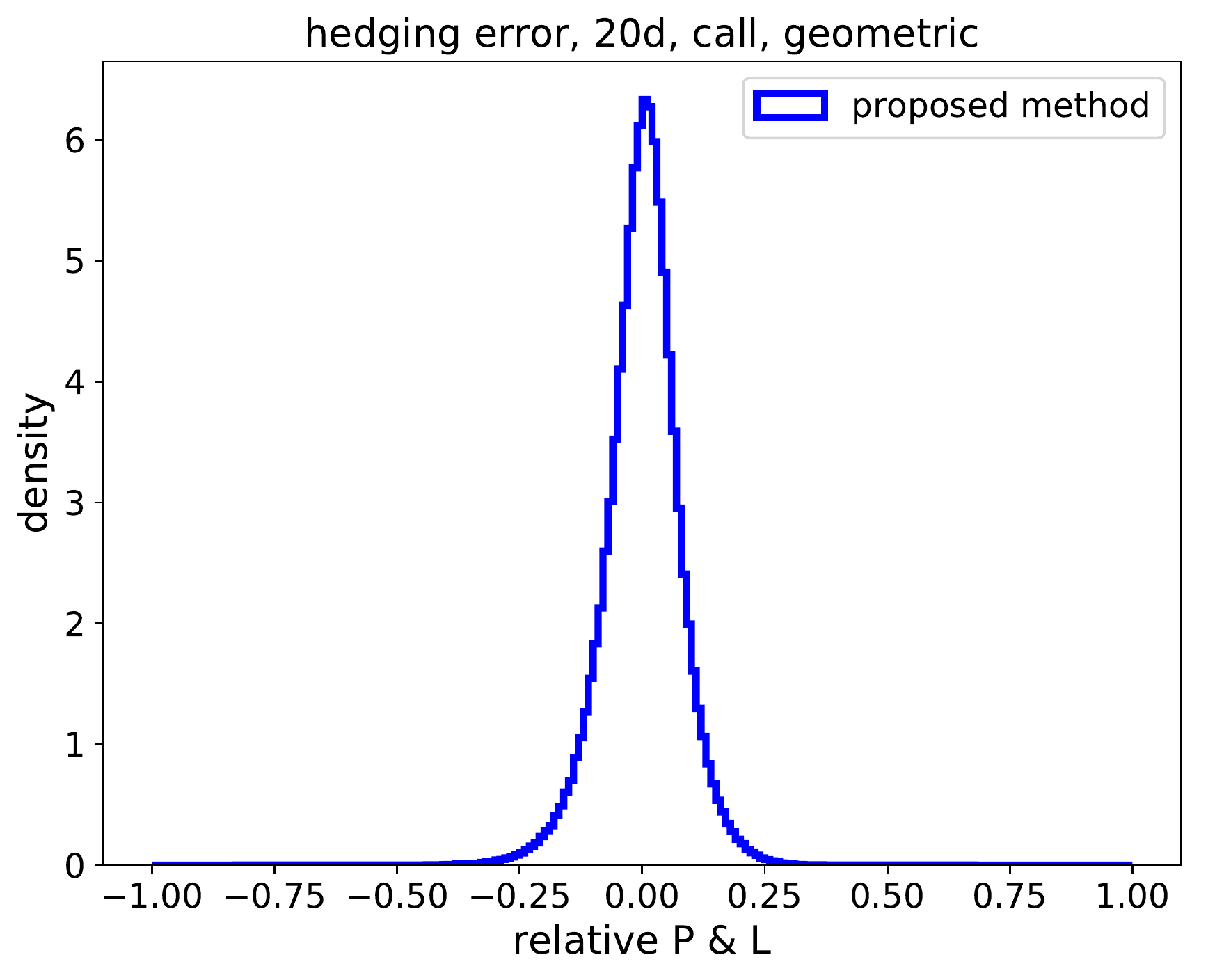}};
\node[below=of img1, node distance=0cm, yshift=1.2cm, xshift=0cm] {relative P\&L};
\node[left=of img1, node distance=0cm, rotate=90, anchor=center, yshift=-0.9cm, xshift=0cm] {density};
\node[above=of img1, node distance=0cm, yshift=-1.2cm, xshift=0cm] {20-dimensional geometric call};
\end{tikzpicture}
\begin{tikzpicture}
\node (img1)
{\includegraphics[scale=0.35,trim={9mm 9mm 0 8mm},clip]{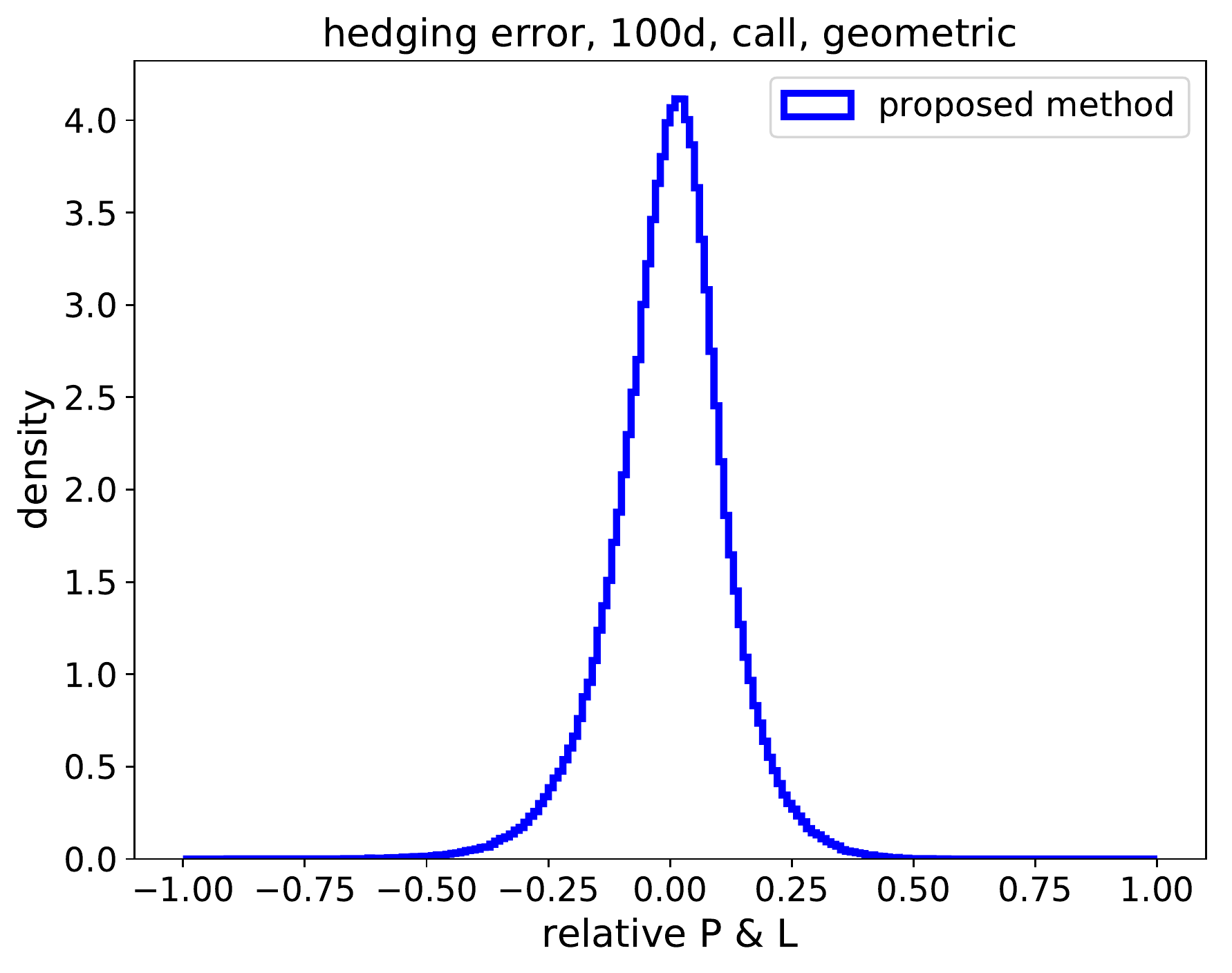}};
\node[below=of img1, node distance=0cm, yshift=1.2cm, xshift=0cm] {relative P\&L};
\node[left=of img1, node distance=0cm, rotate=90, anchor=center, yshift=-0.9cm, xshift=0cm] {density};
\node[above=of img1, node distance=0cm, yshift=-1.2cm, xshift=0cm] {100-dimensional geometric call};
\end{tikzpicture}
\caption{\label{fig:multi_geometric_hedge}
Multi-dimensional geometric call options: Distributions of the relative P\&Ls computed by the proposed neural network approach, subject to 100 hedging intervals.
}
\end{figure}
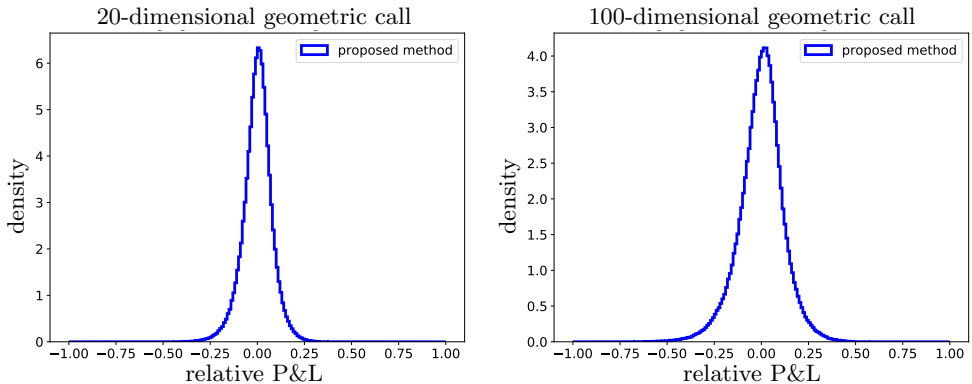


\subsubsection*{Experiment 5}
{\it Delta hedging.}
We perform delta hedging simulations over the period $[0,T]$ with our proposed method. We evaluate the quality of the approach using the distribution of the relative profit and loss \citep{forsyth2005introduction,he2006calibration}:
Relative P\&L $\equiv \frac{e^{-rT}\Pi_T}{v(\vec{s}^0,0)}$,
where $\Pi_T$ is the balance of an initially-zero hedging portfolio at the expiry $T$. For perfect hedging, the relative P\&L should be a Dirac delta function. Due to the discretization of time, the relative P\&L would be close to a normal distribution, where the mean is zero and the standard deviation is a small value depending on $\Delta t$. We emphasize that the computation of the relative P\&L must use both prices and deltas on the entire spacetime. Hence, none of the existing methods referenced in this paper, except our proposed method, are designed to compute the relative P\&L.

Table \ref{tab:multi_geometric_hedge} shows the means and the standard deviations of the relative P\&Ls for all the 720000 simulation paths, computed by our proposed method. The reported values are indeed close to zero. Figure \ref{fig:multi_geometric_hedge} illustrates the distributions of the relative P\&Ls. The resulting distributions are indeed approximately normal distributions with zero means. These results confirm the accuracy of the spacetime prices and the spacetime deltas computed by the proposed method.

\subsection{Multi-dimensional max options}
\label{subsec:multi_max}

Multi-dimensional max options are common in practical applications. In this section, we report simulation results for this type of options.


\begin{table}[b!]
\footnotesize
\begin{center}
2-dimensional max call option
\\
\begin{tabular}{|c|c|c|c|c|c|}
\hline
\multirow{3}{*}[0em]{$s^0_i$}
& \multirow{3}{*}[0em]
{\begin{tabular}{@{}c@{}}exact price\\$v(\vec{s}^0,0)$\end{tabular}}
& \multicolumn{2}{@{}c@{}|}
{proposed method}
& \multicolumn{2}{@{}c@{}|}
{Longstaff-Schwartz}
\\
\cline{3-6}
& 
&
{\begin{tabular}{@{}c@{}}computed price\\$v(\vec{s}^0,0)$\end{tabular}}
& \begin{tabular}{@{}c@{}}percent\\error\end{tabular}
&
{\begin{tabular}{@{}c@{}}computed price\\$v(\vec{s}^0,0)$\end{tabular}}
& \begin{tabular}{@{}c@{}}percent\\error\end{tabular}
\\
\hline
90
&
4.2122
&
4.1992
&
0.31\%
&
4.1748
&
0.89\%
\\
\hline
100
&
9.6333
&
9.6080
&
0.26\%
&
9.5646
&
0.71\%
\\
\hline
110
&
17.3487
&
17.3313
&
0.10\%
&
17.2751
&
0.42\%
\\
\hline
\end{tabular}
\end{center}
\caption{\label{tab:multi_2dmax_1}
2-dimensional max call option: Computed prices at $t=0$, i.e., $v(\vec{s}^0,0)$.
}
\end{table}



\begin{table}[b!]
\footnotesize
\begin{center}
2-dimensional max call option
\\
\begin{tabular}{|c|c|c|c|c|}
\hline
\multirow{2}{*}[0em]{$s^0_i$}
& \multirow{2}{*}[0em]
{\begin{tabular}{@{}c@{}}exact delta
\\$\vec{\nabla}v(\vec{s}^0,0)$\end{tabular}}
& \multicolumn{2}{@{}c@{}|}
{proposed method}
& Longstaff-Schwartz
\\
\cline{3-5}
& 
& computed delta $\vec{\nabla}v(\vec{s}^0,0)$
& percent error
& percent error
\\
\hline
90
&
(0.2062, 0.2062)
&
(0.2025, 0.2019)
&
1.9\%
&
5.2\%
\\
\hline
100
&
(0.3338, 0.3338)
&
(0.3300, 0.3324)
&
0.84\%
&
4.4\%
\\
\hline
110
&
(0.4304, 0.4304)
&
(0.4252, 0.4277)
&
0.96\%
&
3.3\%
\\
\hline
\end{tabular}
\end{center}
\caption{\label{tab:multi_2dmax_2}
2-dimensional max call option: Computed deltas at $t=0$, i.e., $\vec{\nabla}v(\vec{s}^0,0)$.
``Longstaff-Schwartz" is the Longstaff-Schwartz method combined with \citet{thom2009longstaff} and \citet{broadie1996estimating}.}
\end{table}



\begin{table}[b!]
\footnotesize
\begin{center}
2-dimensional max call option
\\
\begin{tabular}{|c|c|c|c|c|}
\hline
\multirow{2}{*}[0em]{$s^0_i$}
& \multicolumn{2}{@{}c@{}|}
{spacetime price $v(\vec{s},t)$}
& \multicolumn{2}{@{}c@{}|}
{spacetime delta
$\vec{\nabla}v(\vec{s},t)$}
\\
\cline{2-5}
& absolute error
& percent error
& absolute error
& percent error
\\
\hline
90
&
0.0563
&
1.3\%
&
0.0155
&
4.9\%
\\
\hline
100
&
0.0828
&
0.85\%
&
0.0180
&
3.4\%
\\
\hline
110
&
0.0678
&
0.39\%
&
0.0207
&
3.0\%
\\
\hline
\end{tabular}
\end{center}
\caption{\label{tab:multi_2dmax_3}
2-dimensional max call option: Spacetime prices and deltas (in terms of absolute and percent errors) computed by our proposed method.
}
\end{table}



\begin{table}[b!]
\footnotesize
\begin{center}
2-dimensional max call option
\\
\begin{tabular}{|c|c|c|c|c|c|}
\hline
$s^0_i$
& proposed method
& Longstaff-Schwartz
\\
\hline
90
&
0.93
&
0.74
\\
\hline
100
&
0.95
&
0.76
\\
\hline
110
&
0.94
&
0.79
\\
\hline
\end{tabular}
\end{center}
\caption{\label{tab:multi_2dmax_4}
2-dimensional max call option: The f1-score of the exercise boundary classification.
}
\end{table}



\begin{figure}[h!]
\centering
\footnotesize
2-dimensional max call option
\\
\begin{tikzpicture}
\node (img1)
{\includegraphics[scale=0.42,trim={9mm 9mm 0 8mm},clip]{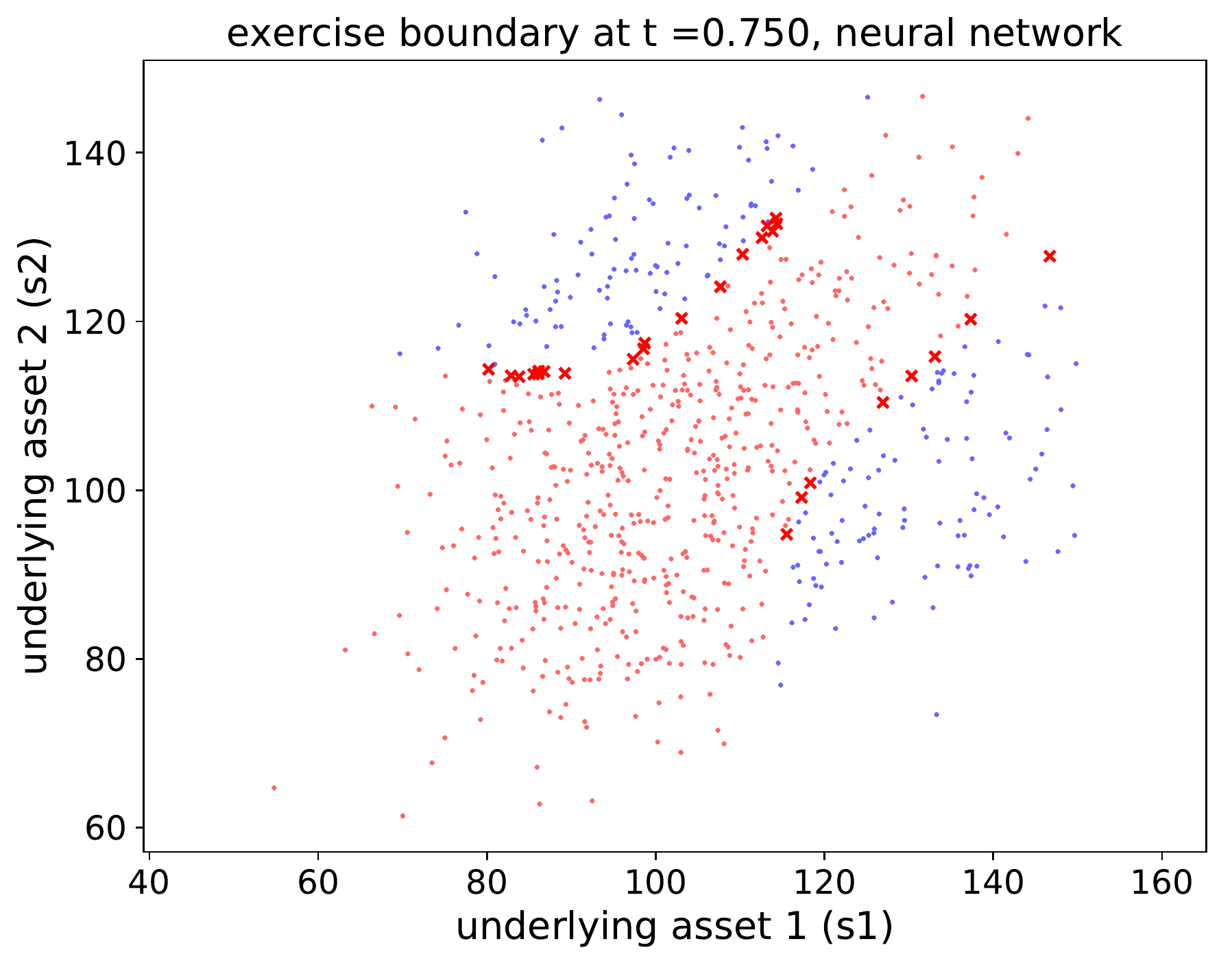}};
\node[below=of img1, node distance=0cm, yshift=1.2cm, xshift=0cm] {1st underlying asset price ($s_1$)};
\node[left=of img1, node distance=0cm, rotate=90, anchor=center, yshift=-0.9cm, xshift=0cm] {2nd underlying asset price ($s_2$)};
\node[above=of img1, node distance=0cm, yshift=-1.15cm, xshift=0cm] {neural network, $t=0.75$};
\end{tikzpicture}
\begin{tikzpicture}
\node (img1)
{\includegraphics[scale=0.42,trim={9mm 9mm 0 8mm},clip]{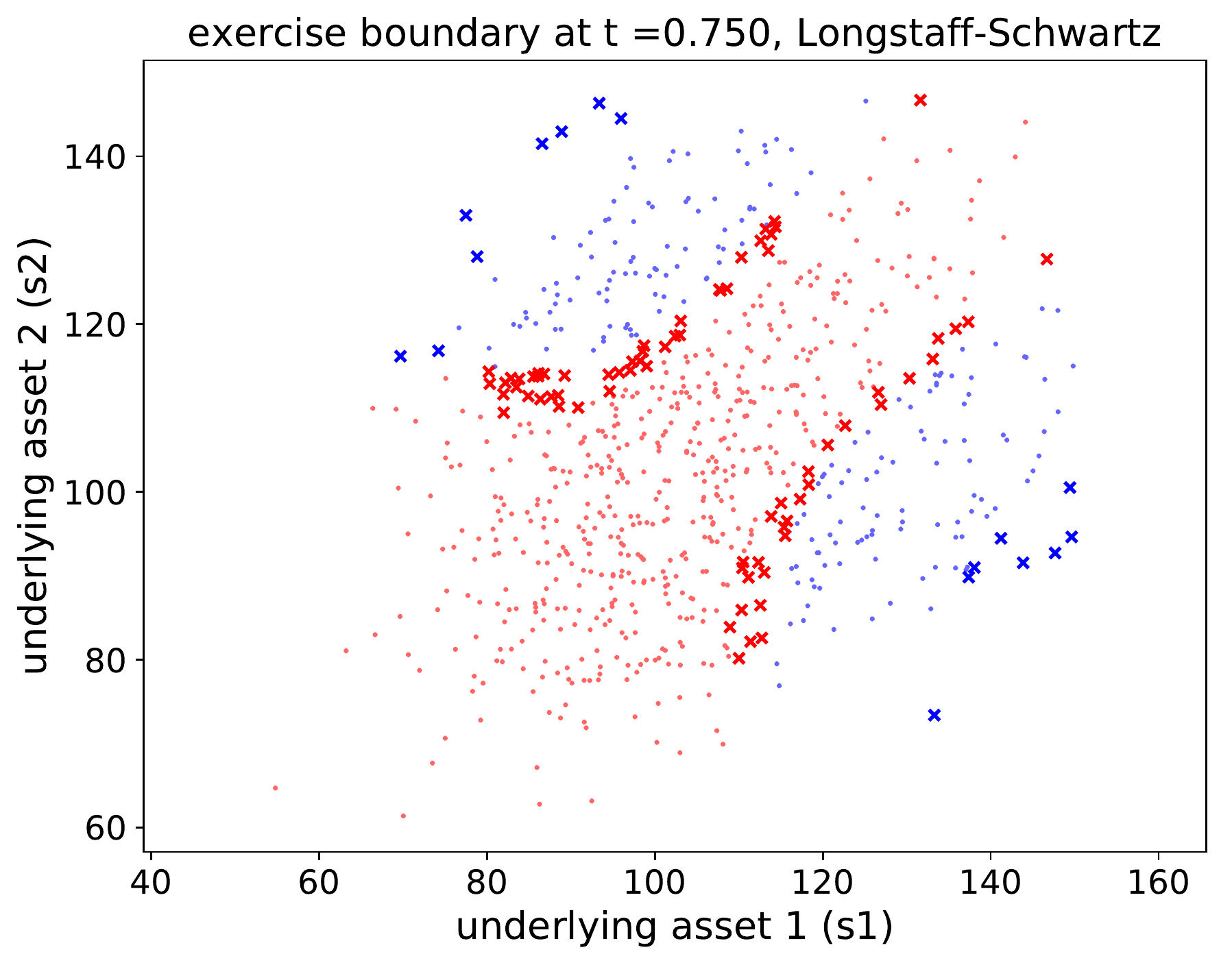}};
\node[below=of img1, node distance=0cm, yshift=1.2cm, xshift=0cm] {1st underlying asset price ($s_1$)};
\node[left=of img1, node distance=0cm, rotate=90, anchor=center, yshift=-0.9cm, xshift=0cm] {2nd underlying asset price ($s_2$)};
\node[above=of img1, node distance=0cm, yshift=-1.2cm, xshift=0cm] {Longstaff-Schwartz, $t=0.75$};
\end{tikzpicture}
\\
\begin{tikzpicture}
\node (img1)
{\includegraphics[scale=0.42,trim={9mm 9mm 0 8mm},clip]{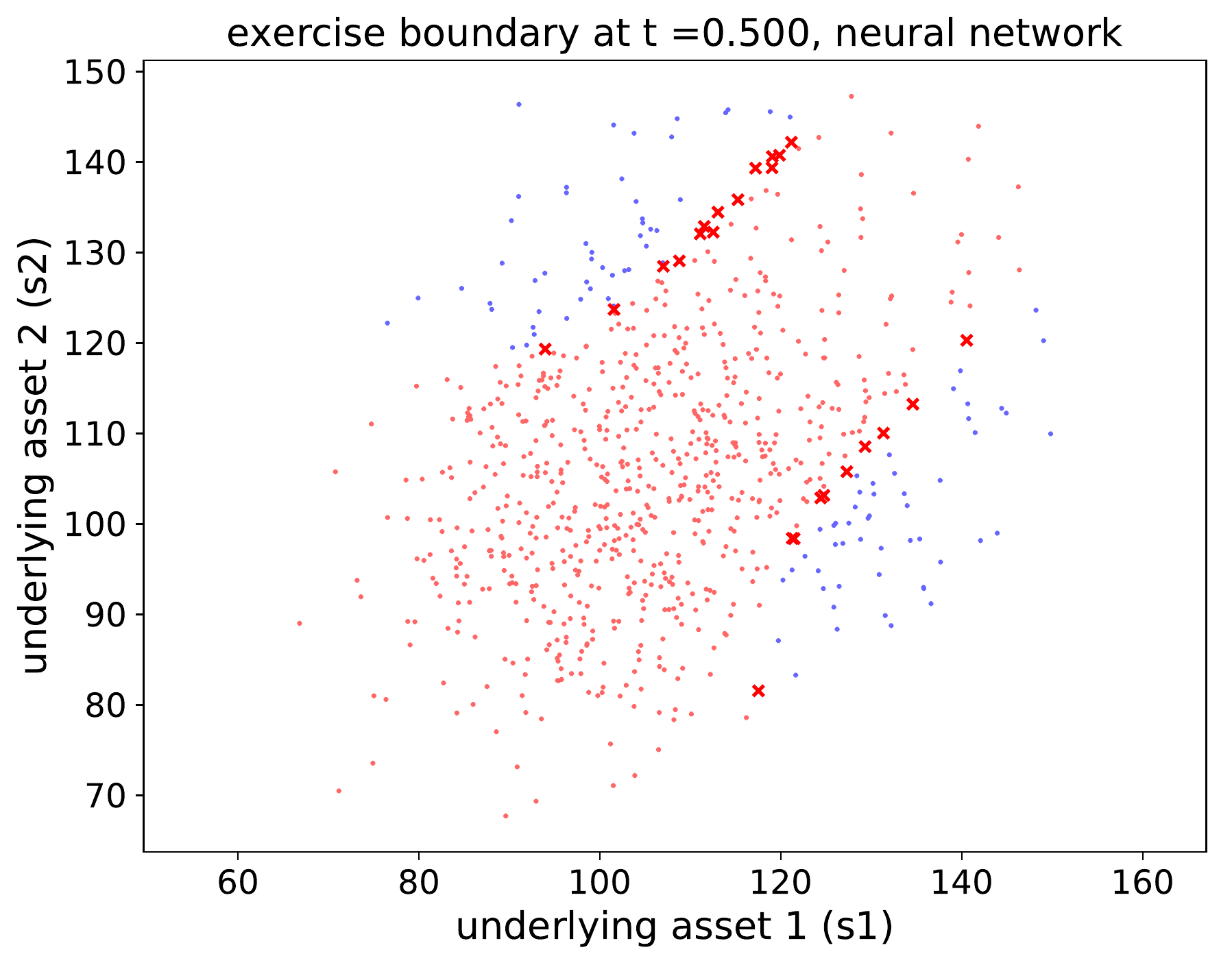}};
\node[below=of img1, node distance=0cm, yshift=1.2cm, xshift=0cm] {1st underlying asset price ($s_1$)};
\node[left=of img1, node distance=0cm, rotate=90, anchor=center, yshift=-0.9cm, xshift=0cm] {2nd underlying asset price ($s_2$)};
\node[above=of img1, node distance=0cm, yshift=-1.15cm, xshift=0cm] {neural network, $t=0.50$};
\end{tikzpicture}
\begin{tikzpicture}
\node (img1)
{\includegraphics[scale=0.42,trim={9mm 9mm 0 8mm},clip]{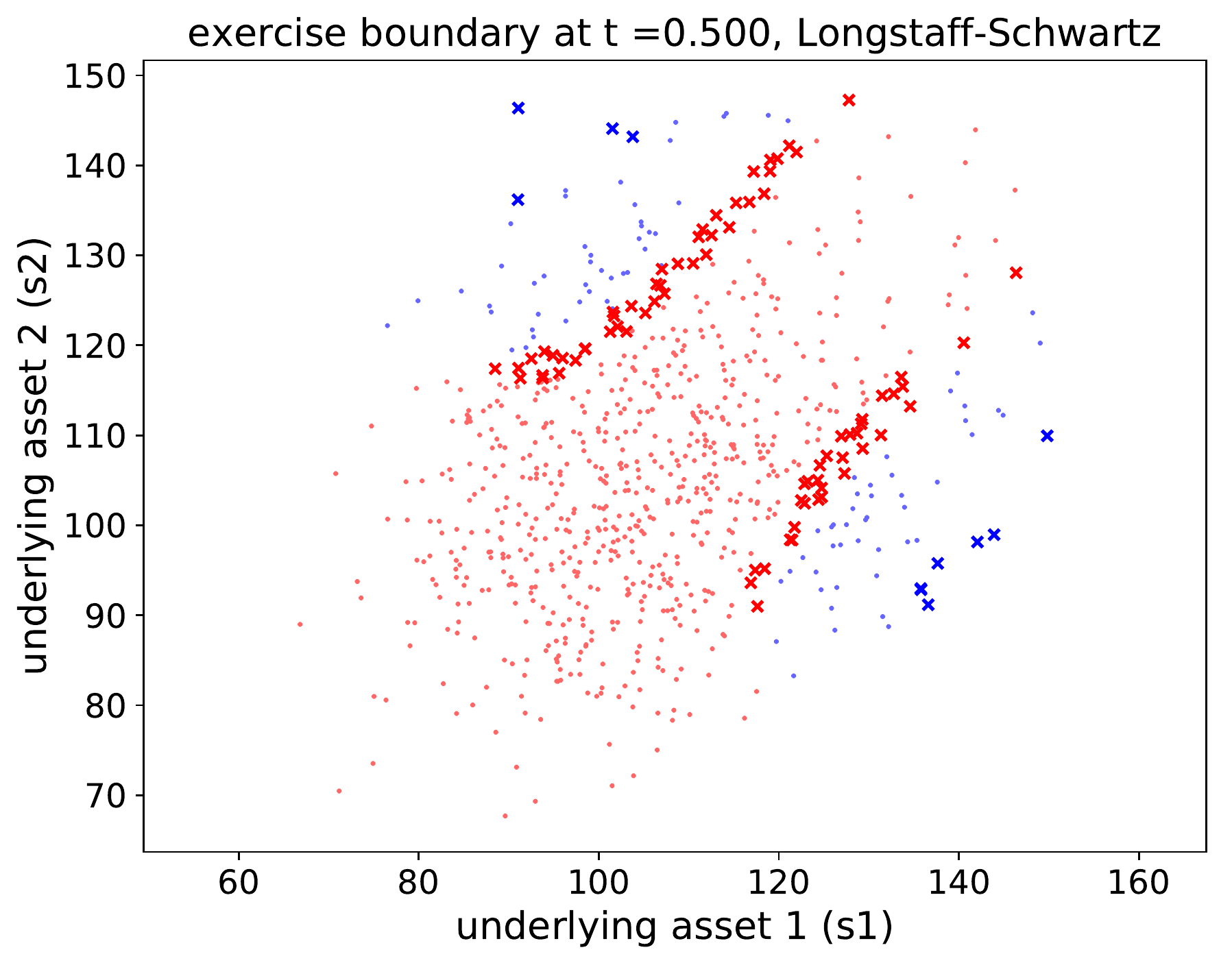}};
\node[below=of img1, node distance=0cm, yshift=1.2cm, xshift=0cm] {1st underlying asset price ($s_1$)};
\node[left=of img1, node distance=0cm, rotate=90, anchor=center, yshift=-0.9cm, xshift=0cm] {2nd underlying asset price ($s_2$)};
\node[above=of img1, node distance=0cm, yshift=-1.2cm, xshift=0cm] {Longstaff-Schwartz, $t=0.50$};
\end{tikzpicture}
\caption{\label{fig:multi_2dmax_boundary}
2-dimensional max call option: Comparison of exercise boundaries between the proposed neural network approach (top left and bottom left) and the Longstaff-Schwartz approach (top right and bottom right). Note that only the time slices of $t=$0.75 and 0.5 are plotted. All blue points: sample points that should be exercised; all red points: sample points that should be continued; bold dark blue points: sample points that should be exercised but are misclassified as continued; bold dark red points: sample points that should be continued but are misclassified as exercised.
}
\end{figure}
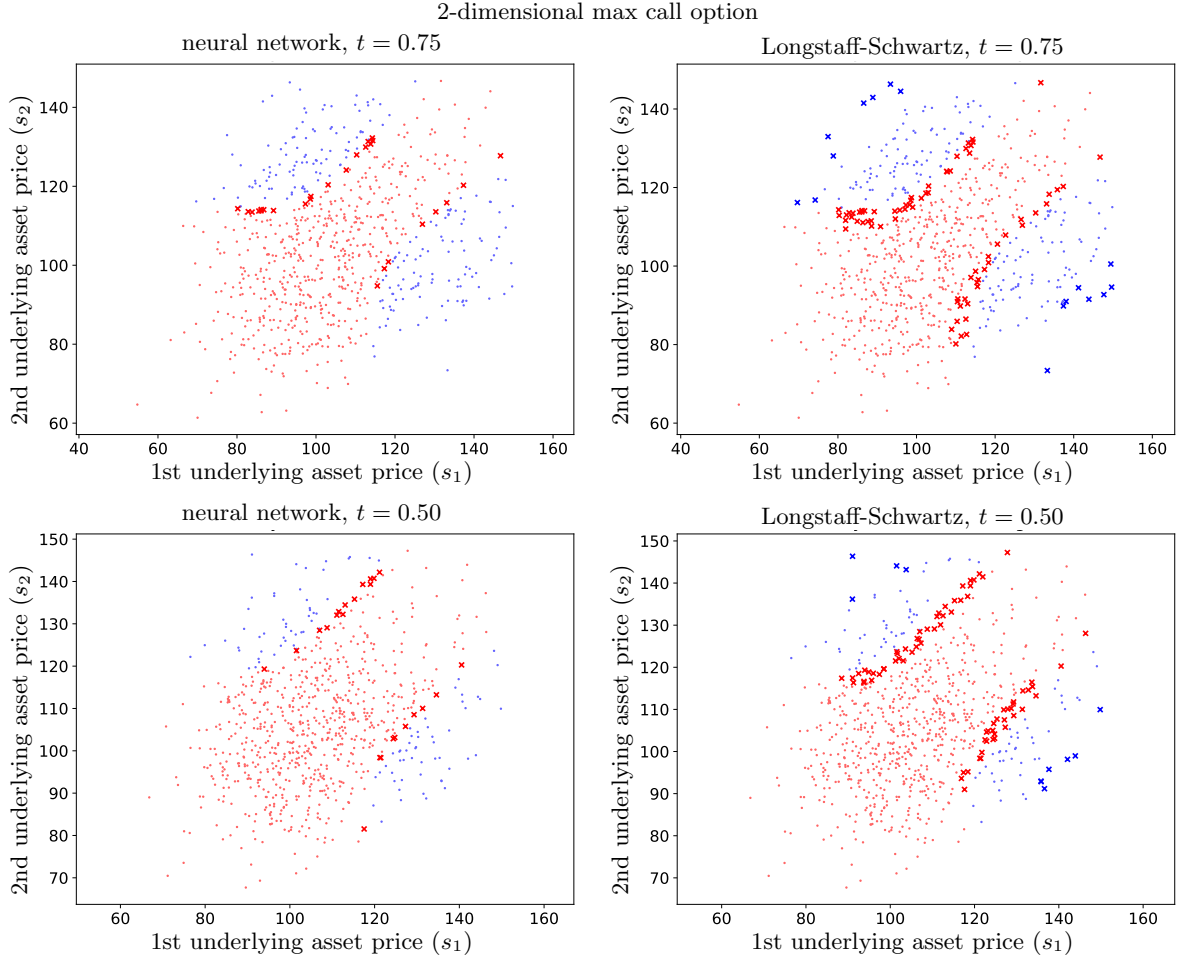


\subsubsection*{Experiment 6}
{\it 2-dimensional max call option.}
Consider the 2-dimensional max call option from Table 3 of \citet{broadie1997pricing}, where the payoff function is
$f(\vec{s}) = \displaystyle\max\left[
\max(s_1, s_2) - K, 0
\right]$, and the parameters are $\rho=0.3$, $\sigma=0.2$, $r=0.05$, $\delta=0.1$, $T=1$. The reason to consider this example is that the exact prices and deltas are available spacetime-wise. More specifically, we approximate the exact prices and deltas by the Crank-Nicolson finite difference method with 1000 timesteps and 2049$\times$2049 space grid points. Hence, we can again benchmark the values computed by our approach with the exact ones.

Using our proposed method, the percent errors of the computed prices at $t=0$ are less than $0.31\%$ (Table \ref{tab:multi_2dmax_1}); the percent errors of the computed deltas at $t=0$ are less than $1.9\%$ (Table \ref{tab:multi_2dmax_2}). These errors are smaller than the corresponding ones computed by the Longstaff-Schwartz method.
In addition, the percent errors of the computed spacetime prices and deltas are less than $1.3\%$ and $4.9\%$ (Table \ref{tab:multi_2dmax_3}).

Here we also compare the exercise boundary computed by the proposed approaches with the one computed by the Longstaff-Schwartz method. Table \ref{tab:multi_2dmax_4} shows that the f1-scores computed by our proposed method are around 0.94, higher than the ones computed by the Longstaff-Schwartz algorithm (around 0.76). Figure \ref{fig:multi_2dmax_boundary} plots the exercise boundaries at the time slices $t=$ 0.75 and 0.5. Similar to Figure \ref{fig:multi_geometric_boundary}, here the misclassified sample points are highlighted by dark cross markers, and we observe again that the proposed neural network approach has fewer misclassified points than the Longstaff-Schwartz method. Both Table \ref{tab:multi_2dmax_4} and Figure \ref{fig:multi_2dmax_boundary} illustrate a more accurate exercise boundary determined by our proposed method than by the Longstaff-Schwartz method.


\begin{table}[t!]
\footnotesize
\begin{center}
2-dimensional max call option
\\
\begin{tabular}{|c|c|c|c|c|}
\hline
\multirow{2}{*}[0em]{$s^0_i$}
& \multicolumn{2}{@{}c@{}|}
{finite difference}
& \multicolumn{2}{@{}c@{}|}
{proposed method}
\\
\cline{2-5}
& mean
& std
& mean
& std
\\
\hline
90
&
0.0025 & 0.1683
&
0.0022 & 0.1932
\\
\hline
100
&
0.0014 & 0.0894
&
0.0016 & 0.0990
\\
\hline
110
&
0.0011 & 0.0544
&
0.0016 & 0.0614
\\
\hline
\end{tabular}
\end{center}
\caption{\label{tab:multi_2dmax_hedge}
2-dimensional max call option: Means and standard deviations of the relative P\&Ls by finite difference versus by the proposed method,  subject to 100 hedging intervals. 
}
\end{table}



\begin{figure}[t!]
\centering
\footnotesize
\begin{tikzpicture}
\node (img1)
{\includegraphics[scale=0.35,trim={9mm 9mm 0 8mm},clip]{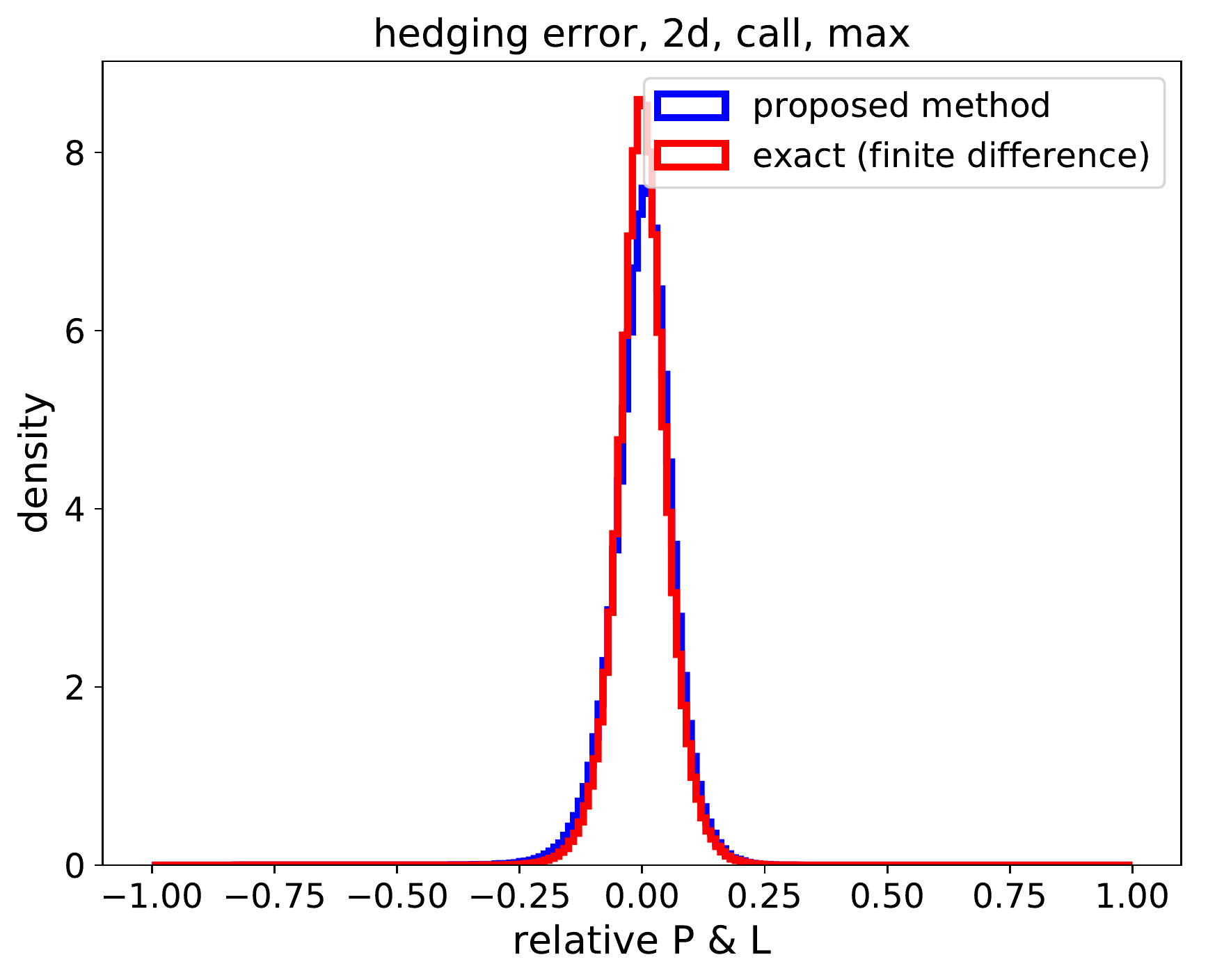}};
\node[below=of img1, node distance=0cm, yshift=1.2cm, xshift=0cm] {relative P\&L};
\node[left=of img1, node distance=0cm, rotate=90, anchor=center, yshift=-0.9cm, xshift=0cm] {density};
\node[above=of img1, node distance=0cm, yshift=-1.2cm, xshift=0cm] {2-dimensional max call option};
\end{tikzpicture}
\caption{\label{fig:multi_2dmax_hedge}
2-dimensional max call option: Comparison of the distributions of the relative P\&Ls computed by the proposed neural network approach (blue) versus by finite difference method (red), subject to 100 hedging intervals.
}
\end{figure}
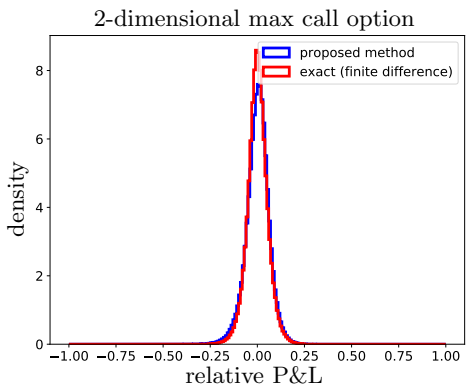


In addition, we compute the relative P\&Ls by the finite difference method\footnote{We note that even though finite difference methods yield nearly exact spacetime prices and deltas, due to the finite number of hedging intervals, the resulting relative P\&Ls are not a Dirac delta distribution.} and compare them with the values computed by our approach.
Table \ref{tab:multi_2dmax_hedge} and Figure \ref{fig:multi_2dmax_hedge} show the means, standard deviations and the distributions of the relative P\&Ls computed by the proposed approach versus by finite difference methods. The results computed by the proposed approach are similar to the ones computed by finite difference methods. This again verifies the accuracy of the spacetime prices and deltas computed by our proposed algorithm.


\begin{table}[h!]
\footnotesize
\begin{center}
5-dimensional max call option
\\
\begin{tabular}{|c|c|c|c|}
\hline
\multirow{2}{*}[0em]{$s^0_i$}
& \multicolumn{2}{@{}c@{}|}
{computed price $v(\vec{s}^0,0)$}
& \multirow{2}{*}[0em]
{\begin{tabular}{@{}c@{}}
computed delta $\vec{\nabla} v(\vec{s}^0,0)$
\\by proposed method\end{tabular}}
\\
\cline{2-3}
& proposed method
& Longstaff-Schwartz
&
\\
\hline
90
&
16.8896
&
16.76
&
(0.1728, 0.1732, 0.1747, 0.1754, 0.1738)
\\
\hline
100
&
26.4876
&
26.28
&
(0.2017, 0.2004, 0.1998, 0.2071, 0.2041)
\\
\hline
110
&
37.0996
&
36.89
&
(0.2157, 0.2198, 0.2190, 0.2149, 0.2202)
\\
\hline
\end{tabular}
\end{center}
\caption{\label{tab:multi_5dmax_1}
5-dimensional max call option: Computed prices and deltas at $t=0$, i.e., $v(\vec{s}^0,0)$ and $\vec{\nabla}v(\vec{s}^0,0)$.
The column ``Longstaff-Schwartz" is the Longstaff-Schwartz prices reported in \citet{firth2005high}.}
\end{table}


\subsubsection*{Experiment 7}
{\it 5-dimensional max call option.}
We study the 5-dimensional max call option from Table 3.5 of \citet{firth2005high}, where $\rho_{i,j}=0$, $\sigma=0.2$, $r=0.05$, $\delta=0.1$, $T=3$. We note that unlike the previous experiments, here the exact solutions are not available. Table \ref{tab:multi_5dmax_1} reports the option prices and deltas at $t=0$ computed by the proposed method. The table also includes the Longstaff-Schwartz prices reported in \citet{firth2005high}. The prices given by the proposed algorithm and the Longstaff-Schwartz method differ by $10^{-2}$. We note that the Longstaff-Schwartz method is a low-biased method due to its sub-optimal computed exercise boundary, as explained in \citet{longstaff2001valuing} and \citet{firth2005high}. The proposed algorithm gives slightly higher prices.


\section{Conclusion}
\label{sec:conclusion}

We propose a neural network framework for high-dimensional American option problems. Our algorithm minimizes the residual of the backward stochastic differential equation that couples both prices and deltas. The neural network is designed to learn the differences between the price functions of the adjacent timesteps. We improve the algorithm by various techniques, including feature selection, weight reuse, ensemble learning, redefining training input ``v", etc. The proposed algorithm yields not only the prices and deltas at $t=0$, but also the prices and deltas on the entire spacetime. The cost of the proposed algorithm grows quadratically with the dimension $d$, which mitigates the curse of dimensionality. In particular, our algorithm outperforms the Longstaff-Schwartz algorithm when $d\geq 20$.

We note that the main drawback of the proposed algorithm is that the computational cost is quadratic (rather than linear) in the number of the timesteps $N$, even though a mitigation is proposed in Section \ref{subsec:NN2}. A potential future work is to re-design the architecture of the neural network in order to improve this drawback.




\bibliographystyle{rQUF}
\bibliography{Reference}

\end{document}